\documentclass[aos,preprint]{imsart}

\RequirePackage[OT1]{fontenc}
\RequirePackage{float}
\RequirePackage{amsthm,amsmath}
\RequirePackage[numbers]{natbib}
\RequirePackage[colorlinks,citecolor=blue,urlcolor=blue]{hyperref}
\usepackage{amsmath}
\usepackage{amssymb}
\usepackage{bm}
\usepackage{enumitem}
\usepackage{multirow}
\usepackage{graphicx}
\usepackage{mathrsfs}
\usepackage{amsthm}
\usepackage{lipsum}

\startlocaldefs
\numberwithin{equation}{section}
\theoremstyle{plain}
\newtheorem{theorem}{Theorem}[section]
\newtheorem{cor}{Corollary}[section]
\newtheorem{lemma}[theorem]{Lemma}
\newtheorem{assumption}{Assumption}
\newtheorem{remark}{remark}
\newtheorem{proposition}[theorem]{Proposition}

\DeclareMathOperator*{\argmax}{arg\,max}
\endlocaldefs

\begin{document}
	\nocite{*}
	\begin{frontmatter}
		\title{Posterior graph selection and estimation consistency for 
		high-dimensional Bayesian DAG models}
		\runtitle{Posterior consistency for Bayesian DAG models}
		
		\begin{aug}
			\author{\fnms{Xuan} \snm{Cao}\ead[label=e1]{caoxuan@ufl.edu}},
			\author{\fnms{Kshitij} \snm{Khare}\ead[label=e2]{kdkhare@stat.ufl.edu}}
			\and
			\author{\fnms{Malay} \snm{Ghosh}\ead[label=e3]{ghoshm@stat.ufl.edu}}
			
			\runauthor{X. Cao et al.}
			
			\affiliation{University of Florida}
			
			
			\address{Department of Statistics\\
				University of Florida\\
				102 Griffin-Floyd Hall\\
				Gainesville, FL 32611\\
				\printead{e1}\\
				\phantom{E-mail:\ }\printead*{e2}\\
				\phantom{E-mail:\ }\printead*{e3}}
		\end{aug}
		
		\begin{abstract}
		Covariance estimation and selection for high-dimensional multivariate 
		datasets is a fundamental problem in modern statistics. Gaussian directed 
		acyclic graph (DAG) models are a popular class of models used for this 
		purpose. Gaussian DAG models introduce sparsity in the Cholesky factor of 
		the inverse covariance matrix, and the sparsity pattern in turn corresponds 
		to specific conditional independence assumptions on the underlying 
		variables. A variety of priors have been developed in recent years for 
		Bayesian inference in DAG models, yet crucial convergence and sparsity 
		selection properties for these models have not been thoroughly 
		investigated. Most of these priors are adaptations/generalizations of the 
		Wishart distribution in the DAG context. In this paper, we consider a flexible 
		and general class of these `DAG-Wishart' priors with multiple shape 
		parameters. Under mild regularity assumptions, we establish 
		strong graph selection consistency and establish posterior convergence 
		rates for estimation when the number of variables $p$ is allowed to grow at 
		an appropriate sub-exponential rate with the sample size $n$. 
		\end{abstract}
		
		
		
	\end{frontmatter}

\section{Introduction} \label{sec1}

\noindent
One of the major challenges in modern day statistics is to formulate models and 
develop inferential procedures to understand the complex multivariate relationships 
present in high dimensional datasets, where the number of variables is much larger 
than the number of samples. The covariance matrix, denoted by $\Sigma$, is one of 
the most fundamental objects that quantifies relationships between variables in 
multivariate datasets. A common and effective approach for covariance estimation in 
sample starved settings is to induce sparsity either in the covariance matrix, its 
inverse, or the Cholesky parameter of the inverse. The sparsity patterns in these 
matrices can be uniquely encoded in terms of appropriate graphs. Hence the 
corresponding models are often referred to as covariance graph models (sparsity in 
$\Sigma$), concentration graph models (sparsity in $\Omega = \Sigma^{-1}$), and 
directed acyclic graph (DAG) models (sparsity in the Cholesky parameter of 
$\Omega$). 

In this paper, we focus on Gaussian DAG models. In particular, suppose we have 
i.i.d. observations ${\bf Y}_1, {\bf Y}_2, \cdots, {\bf Y}_n$ from a $p$-variate 
normal distribution with mean vector ${\bf 0}$ and covariance matrix $\Sigma$. Let 
$\Omega = LD^{-1}L^T$ be the modified Cholesky decomposition of the inverse 
covariance matrix $\Omega = \Sigma^{-1}$, i.e., $L$ is a lower triangular matrix 
with unit diagonal entries, and $D$ is a diagonal matrix with positive diagonal 
entries. For a DAG model, this normal distribution is assumed to be Markov with 
respect to a given directed acyclic graph $\mathscr{D}$ with vertices $\{1,2, \cdots, 
p\}$ (edges directed from larger to smaller vertices). This is equivalent to saying that 
$L_{ij}  = 0$ whenever $\mathscr{D}$ does not have a directed edge from $i$ to $j$ 
(these concepts are discussed in detail in Section \ref{sec2}).  Hence a Gaussian 
DAG model restricts $\Sigma$ (and $\Omega$) to a lower dimensional space 
by imposing sparsity constraints encoded in $\mathscr{D}$ on $L$. 

On the frequentist side, a variety of penalized likelihood methods for sparse estimation 
of $L$ exist in the literature, see \cite{HLPL:2006, Rutimann:Buhlmann:2009, Shojaie:Michailidis:2010, RLZ:2010, 
AAZ:2015, Yu:Bien:2016, KORR:2017}. Some of these methods, such as those in \cite{RLZ:2010, 
Yu:Bien:2016}, constrain the sparsity pattern in $L$ to be banded, whereas others, 
such as those in \cite{HLPL:2006, Shojaie:Michailidis:2010,KORR:2017}, put 
no constraints on the sparsity pattern.  Most 
of the above methods derive asymptotic estimation and model selection consistency 
properties for the resulting estimator in a an appropriate high-dimensional regime. 
See Section \ref{sec:penalized} for more details. On the Bayesian side, the first class of priors on the restricted space of covariance matrices corresponding to a Gaussian DAG model was initially developed in \cite{Geiger:Heckerman:2002, Smith:Kohn:2002}. As pointed out in \cite{BLMR:2016}, the priors in \cite{Geiger:Heckerman:2002} can be considered as 
analogs of the $\mathcal{G}$-Wishart distribution for concentration graph models 
(inducing sparsity in $\Omega$). In fact, for the special case of perfect DAGs the priors 
in \cite{Geiger:Heckerman:2002} are same as the $\mathcal{G}$-Wishart priors. 
As with the $\mathcal{G}$-Wishart priors, the priors in \cite{Geiger:Heckerman:2002} 
have a single shape parameter. Letac and Massam \cite{Letac:Massam:2007} 
introduced a flexible class of priors with multiple shape parameters which facilitate 
differential shrinkage in high-dimensional settings. However, these priors are defined 
only for perfect DAG models. Recently, Ben-David et al. \cite{BLMR:2016} introduce 
a class of DAG-Wishart distributions with multiple shape parameters. This class of 
distributions is defined for arbitrary DAG models, and is identical to the Letac-Massam 
$IW_{P_\mathcal{G}}$ priors for the special case of perfect DAG models. Thus, this 
class of DAG-Wishart distributions offers a flexible framework for Bayesian inference 
in Gaussian DAG models, and generalizes previous Wishart-based priors for DAG 
models. 

The priors above are specified for a known DAG $\mathscr{D}$, and provide a 
Bayesian approach for estimating the covariance matrix. However, if the underlying 
DAG is not known and needs to be selected, one can easily extend this framework 
by specifying a prior on the space of DAGs, and looking at the posterior probabilities 
of the DAGs given the data. Such an approach was used in the context of 
concentration graph models in \cite{Banerjee:Ghosal:2015}. The utility of this Bayesian 
approach in substantially improving finite sample graph selection performance (as 
compared to existing penalized likelihood methods) has been demonstrated in 
\cite{BLMR:2016}. We discuss and demonstrate this further in 
Sections \ref{sec:penalized} and \ref{sec:experiments}. 

Despite the developments in Bayesian methods for analyzing Gaussian DAG models, 
a comprehensive evaluation of the high-dimensional consistency properties 
of these methods has not been undertaken to the best of our knowledge. Assuming 
the data comes from a ``true" DAG model, two aspects of the asymptotic behavior of 
the posterior are of interest: (a) assigning high posterior probability to the ``true" 
underlying graph (graph selection consistency), and (b) estimating the ``true" 
covariance matrix accurately (estimation consistency). 

Gaussian concentration graph models, which induce sparsity in the inverse covariance 
matrix $\Omega$, are a related but markedly different class of models as compared 
to Gaussian DAG models. The two classes of models intersect only at perfect DAGs, 
which are equivalent to decomposable concentration graph models. In the context of 
concentration graph models, high dimensional posterior estimation consistency has 
been explored in recent work 
\cite{Banerjee:Ghosal:2014, Banerjee:Ghosal:2015, XKG:2015}. In 
\cite{Banerjee:Ghosal:2014, XKG:2015} estimation consistency is established for the 
decomposable concentration graph models when the underlying concentration graph 
is known, and the number of variables $p$ is allowed to increase at an appropriate 
sub-exponential rate relative to the sample size $n$. 
Banerjee and Ghosal \cite{Banerjee:Ghosal:2015} get rid of the assumption of 
decomposability and do not assume the true concentration graph is known. They use 
independent Laplace priors for the off-diagonal entries of the inverse covariance 
matrix, and use independent Bernoulli priors for the edges of the concentration 
graph. In this framework, estimation consistency is established in 
\cite{Banerjee:Ghosal:2015} under suitable regularity assumptions when $\sqrt{(p + s) 
\log p/n} \rightarrow 0$ ($s$ denotes the total number of non-zero off-diagonal entries 
in the ``true" inverse covariance matrix). The authors do not address model selection 
consistency, but provide high-dimensional Laplace approximations for the marginal 
posterior probabilities for the graphs, along with a proof of the validity of these 
approximations. 

In this paper, our goal is to explore both model selection and estimation consistency in 
a high-dimensional setting for Gaussian DAG models. In particular, we consider a 
hierarchical Gaussian DAG model with DAG-Wishart priors on the covariance matrix 
and independent Bernoulli priors for each edge in the DAG. Under standard regularity 
assumptions, which include letting $p$ increase at an appropriate sub exponential rate 
with $n$, we establish {\it posterior ratio consistency} (Theorem~\ref{thm1}), i.e., the 
ratio of the maximum marginal posterior probability assigned to a ``non-true" DAG to 
the posterior probability assigned to the ``true" DAG converges to zero in probability 
under the true model. In particular, this implies that the true DAG will be the mode of 
the posterior DAG distribution with probability tending to $1$ as $n \rightarrow \infty$.  
An almost sure version of posterior ratio consistency is established in 
Theorem~\ref{thm2}. Next, under the additional assumption that the prior over DAGs 
is restricted to graphs with edge size less than an appropriate function of the sample 
size $n$, we show {\it strong graph selection consistency} (Theorem~\ref{thm3}) and 
establish a posterior convergence rate for estimation of the inverse covariance matrix 
(Theorem E.1 in the Supplemental document). Strong graph selection consistency 
implies that under the true model, the posterior probability of the true graph converges 
in probability to $1$ as $n \rightarrow \infty$. As pointed out in 
Remark \ref{edgerestrictassumption}, the assumption of restricting the prior over 
models with appropriately bounded parameter size has been used in 
\cite{Narisetty:He:2014} for regression models, and in \cite{Banerjee:Ghosal:2015} 
for concentration graph models. 

Narisetty and He \cite{Narisetty:He:2014} establish strong model selection consistency 
of high-dimensional regression models with spike and slab priors. While there are 
some connections between our model and the one in \cite{Narisetty:He:2014} since 
the entries of $L$ can be interpreted as appropriate regression coefficients, there are 
fundamental differences between the the two models and the corresponding analyses. 
A detailed explanation of this is provided in Remark \ref{spikeslabregression}. 

{ In recent work, Altamore et al. \cite{ACR:2013} develop a class of objective 
non-local priors for Gaussian DAG models. This class of priors is structurally 
different from the DAG Wishart priors of \cite{BLMR:2016}, and we also investigate 
posterior model selection consistency under these non-local priors. In fact, we show 
under almost identical assumptions to the DAG Wishart setting that under the true 
model, the posterior probability of the true graph converges in probability to $1$ as 
$n \rightarrow \infty$ (Theorem \ref{thm4}). Another recent paper \cite{CLP:2017} 
tackles the problem of covariate-adjusted DAG selection, i.e., estimating a 
sparse DAG based covariance matrix in the presence of covariates. Establishing 
consistency in this more complex setup is beyond the scope of our paper, and will be 
an excellent topic for future research. 

The rest of the paper is structured as follows. Section \ref{sec2} provides background 
material from graph theory and Gaussian DAG models. In 
Section \ref{sec:modelspecification} we provide the hierarchical Bayesian DAG 
model. Graph selection consistency results are stated in 
Section \ref{sec:modelselection}, and the proofs are provided in 
Section \ref{sec:modelselectionproofs}. In Section \ref{sec:nonlocalpriors}, we 
establish graph selection consistency for non-local priors. A detailed discussion and 
comparison of the Bayesian approach of \cite{BLMR:2016} and existing penalized 
likelihood approaches is undertaken in Section \ref{sec:penalized}. In 
Section \ref{sec:experiments} we use simulation experiments to illustrate the 
posterior ratio consistency result, and demonstrate the benefits of the 
Bayesian approach for graph selection 
vis-a-vis existing penalized likelihood approaches.}

\section{Preliminaries}\label{sec2}

\noindent
In this section, we provide the necessary background material from graph theory, 
Gaussian DAG models, and DAG-Wishart distributions. 

\subsection{Gaussian DAG models} \label{sec2.1}

\noindent
Throughout this paper, a directed acyclic graph (DAG) $\mathscr{D} = (V,E)$ consists 
of the vertex set $V = \{1,\ldots,p\}$ and an edge set $E$ such that there is no directed 
path starting and ending at the same vertex. As in \cite{BLMR:2016}, we will without 
loss of generality assume a parent ordering, where that all the edges are directed from 
larger vertices to smaller vertices. The set of parents of $i$, denoted by $pa_i(\mathscr D)$, is the 
collection of all vertices which are larger than $i$ and share an edge with $i$. 
Similarly, the set of children of $i$, denoted by $chi_i(\mathscr D)$, is the collection of all vertices 
which are smaller than $i$ and share an edge with $i$. 

A Gaussian DAG model over a given DAG $\mathscr{D}$, denoted by 
$\mathscr{N}_{\mathscr{D}}$, consists of all multivariate Gaussian distributions which 
obey the directed Markov property with respect to a DAG $\mathscr{D}$. In particular, 
if $\bm{y}=(y_1, \ldots, y_p)^T \sim N_p(0,\Sigma)$ and $N_p(0,\Sigma) \in 
\mathscr{N}_{\mathscr{D}}$, then $y_i \perp \bm{y}_{\{i+1,\ldots,p\}\backslash pa_i(\mathscr D)}|
\bm{y}_{pa_i(\mathscr D)}$ for each $i$. 

Any positive definite matrix $\Omega$ can be uniquely decomposed as $\Omega = 
LD^{-1}L^T$, where $L$ is a lower triangular matrix with unit diagonal entries, and $D$ 
is a diagonal matrix with positive diagonal entries. This decomposition is known as the 
modified Cholesky decomposition of $\Omega$ (see for example 
\cite{Pourahmadi:2007}). It is well-known that if $\Omega = LD^{-1}L^T$ is the 
modified Cholesky decomposition of $\Omega$, then $N_p(0,\Omega^{-1}) \in 
\mathscr{N}_{\mathscr{D}}$ if and only if $L_{ij} = 0$ whenever $i \notin pa_j (\mathscr D)$. In other 
words, the structure of the DAG $\mathscr{D}$ is reflected in the Cholesky factor of 
the inverse covariance matrix. In light of this, it is often more convenient to 
reparametrize in terms of the Cholesky parameter of the inverse covariance matrix 
as follows. 

Given a DAG $\mathscr{D}$ on $p$ vertices, denote $\mathscr{L}_{\mathscr{D}}$ as 
the set of lower triangular matrices with unit diagonals and $L_{ij} = 0$ if $i \notin 
pa_j(\mathscr D)$, and let $\mathscr{D}_+^p$ be the set of strictly positive diagonal matrices in 
$\mathbb{R}^{p \times p}$. We refer to $\Theta_{\mathscr{D}} = \mathscr{D}_+^p 
\times \mathscr{L}_{\mathscr{D}}$ as the Cholesky space corresponding to 
$\mathscr{D}$, and $(D,L) \in \Theta_{\mathscr{D}}$ as the Cholesky parameter 
corresponding to $\mathscr{D}$. In fact, the relationship between the DAG and the 
Cholesky parameter implies that 
$$
\mathscr{N}_{\mathscr{D}} = \{N_p(0,(L^T)^{-1}DL^{-1}):(D,L) \in 
\Theta_{\mathscr{D}}\}. 
$$

The skeleton of $\mathscr{D}$, denoted by $\mathscr{D}^u = (V,E^u)$, can be 
obtained by replacing all the directed edges of $\mathscr{D}$ by undirected ones. A 
DAG $\mathscr{D}$ is said to be perfect if the parents of all vertices are adjacent. An 
undirected graph is called decomposable if it has no induced cycle of length $n \ge 4$, 
excluding the loops. It is known that if $\mathscr{D}$ is a perfect directed acyclic graph 
(DAG), then $\mathscr{D}^u$ is a decomposable graph. Conversely, given an 
undirected decomposable graph, one can always direct the edges so that the resulting 
graph is a perfect DAG. This fact can be used to show that the class of normal 
distributions satisfying the directed Markov property with respect to $\mathscr{D}$ 
(DAG models, sparsity in $L$) is identical to the class of normal distributions satisfying 
the undirected Markov property with respect to $\mathscr{D}^u$ (concentration 
graph models, sparsity in $\Omega$) {\it if and only if} $\mathscr{D}$ is a perfect DAG. 
(see \cite{PPS:1989}).

\subsection{DAG-Wishart distribution} \label{sec2.2}

\noindent
In this section, we specify the multiple shape parameter DAG-Wishart distributions 
introduced in \cite{BLMR:2016}. First, we provide required notation. Given a directed 
graph $\mathscr{D} = (V,E)$, with $V = \{1, \ldots, p\}$, and a $p \times p$ matrix $A$, 
denote the column vectors $A_{\mathscr D .i}^> = (A_{ij})_{j \in pa_{i}(\mathscr D)}$ and 
$A_{\mathscr D.i}^{\ge} = (A_{ii}, (A_{\mathscr D.i}^>)^T)^T.$ Also, let $A_{\mathscr 
D}^{>i} = (A_{kj})_{k,j \in pa_{i}(\mathscr D)}$, $$ A_{\mathscr D}^{ \ge i} = 
\left[ \begin{matrix}
A_{ii} & (A_{\mathscr D.i}^>)^T \\
A_{\mathscr D.i}^> & A_{\mathscr D}^{>i}
\end{matrix} \right]. 
$$

\noindent
In particular, $A_{\mathscr D.p}^{\ge} = A_{\mathscr D}^{ \ge p} = A_{pp}$. 

The DAG-Wishart distributions in \cite{BLMR:2016} corresponding to a DAG 
$\mathscr{D}$ are defined on the Cholesky space $\Theta_{\mathscr{D}}$. Given a 
positive definite matrix $U$ and a $p$-dimensional vector ${\boldsymbol \alpha} 
(\mathscr{D})$, the (unnormalized) density of the DAG-Wishart distribution on 
$\Theta_{\mathscr{D}}$ is given by 
\begin{equation} \label{a2}
\exp\{-\frac12\mbox{tr}((LD^{-1}L^T)U)\} \prod_{i=1}^p 
D_{ii}^{-\frac{\alpha_i (\mathscr{D})}2},
\end{equation} 

\noindent
for every $(D,L) \in \Theta_{\mathscr{D}}$. Let $\nu_{i}(\mathscr D) = |pa_{i}
(\mathscr{D})| = |\{j: j>i, (j,i) \in E({\mathscr{D}})\}|$. If $\alpha_i(\mathscr D) - 
\nu_i(\mathscr D) >2$, for all $1 \le i \le p$, the density in (\ref{a2}) can be normalized 
to a probability density, and the normalizing constant is given by 
\begin{equation} \label{a1}
z_{\mathscr{D}}(U,{\boldsymbol \alpha} (\mathscr{D})) = \prod_{i=1}^{p}
\frac{\Gamma(\frac{\alpha_i (\mathscr{D})}2 - \frac{\nu_{i}(\mathscr D)}2 - 
1)2^{\frac{\alpha_i (\mathscr{D})}2 - 1}(\sqrt{\pi})^{\nu_{i}(\mathscr D)} 
det(U_{\mathscr D}^{>i})^{\frac{\alpha_i (\mathscr{D})}2 - \frac{\nu_{i}(\mathscr D)}2 - 
\frac32}}{det(U_{\mathscr D}^{\ge i})^{\frac{\alpha_i (\mathscr{D})}2 - \frac{\nu_{i}
(\mathscr D)}2 - 1}},
\end{equation}

\noindent
In this case, we define the DAG-Wishart density 
$\pi_{U,\bm \alpha(\mathscr D)}^{\Theta_{\mathscr{D}}}$ on the Cholesky space 
$\Theta_{\mathscr{D}}$ by 
$$
\pi_{U,\bm \alpha(\mathscr D)}^{\Theta_{\mathscr{D}}} (D,L) = 
\frac{1}{z_{\mathscr{D}}(U,{\boldsymbol \alpha} (\mathscr{D}))} exp\{-\frac12\mbox{tr}
((LD^{-1}L^T)U)\} \prod_{i=1}^p D_{ii}^{-\frac{\alpha_i (\mathscr{D})}2} 
$$

\noindent
for every $(D,L) \in \Theta_{\mathscr{D}}$. The above density has the same form as 
the classical Wishart density, but is defined on the lower dimensional space 
$\Theta_{\mathscr{D}}$ and has $p$ shape parameters 
$\{\alpha_i (\mathscr{D})\}_{i=1}^p$ which can be used for differential shrinkage 
of the variables in high-dimensional settings. 

The class of densities $\pi_{U,\bm \alpha(\mathscr D)}^{\Theta_{\mathscr{D}}}$ form a 
conjugate family of priors for the Gaussian DAG model $\mathscr{N}(\mathscr{D})$. In 
particular, if the prior on $(D,L) \in \Theta_{\mathscr{D}}$ is $\pi_{U,\bm 
\alpha(\mathscr D)}^{\Theta_{\mathscr{D}}}$ and $\bm{Y}_1, \ldots, \bm{Y}_n$ are 
independent, identically distributed $N_p(\bm{0},(L^T)^{-1}DL^{-1})$ random vectors, 
then the posterior distribution of $(D,L)$ is $\pi_{\tilde{U},\tilde{\bm \alpha}(\mathscr 
D)}^{\Theta_{\mathscr{D}}}$, where $S = \frac1n\sum_{i=1}^n\bm{Y}_i\bm{Y}_i^T$ 
denotes the sample covariance matrix, $\tilde{U} = U + nS$, and $\tilde{\bm \alpha}
(\mathscr D) = (n+\alpha_1(\mathscr D), \ldots, n+\alpha_p(\mathscr D))$.

\section{Model specification} \label{sec:modelspecification}

\noindent
Let ${\bm Y}_1, {\bm Y}_2, \cdots, {\bm Y}_n \in \mathbb{R}^p$ be the observed data. 
The class of DAG-Wishart distributions in Section \ref{sec2} can be used for Bayesian 
covariance estimation and DAG selection through the following hierarchical model. 

\begin{eqnarray} \label{a3}
\begin{split}
&\bm Y | \left((D,L), \mathscr D\right) \sim N_p \left( \bm 0, (LD^{-1}L^T)^{-1}\right),
\\
&(D,L) | \mathscr D \sim \pi_{U,\bm \alpha(\mathscr D)}^{\Theta_{\mathscr{D}}},
\\
&\pi (\mathscr D) = \prod_{i=1}^{p-1} q^{\nu_i(\mathscr{D})} (1-q)^{p-i-
\nu_i(\mathscr{D})},
\end{split}
\end{eqnarray}
	
\noindent
The prior density for DAGs above corresponds to an Erdos-Renyi type of 
distribution on the space of DAGs, where each directed edge is present with 
probability $q$ independently of the other edges. In particular, similar to 
\cite{Banerjee:Ghosal:2015}, define $\gamma_{ij} = \mathbb{I} \{(i,j) \in 
E(\mathscr D)\}$, $1 \le i < j \le p$ to be the edge indicator. Let $\gamma_{ij}$, $1 \le 
i < j < p$ be independent identically distributed Bernoulli($q$) random variables. It 
follows that 
$$
\pi (\mathscr D) = \prod_{(i,j):1 \leq i < j \leq p} q^{\gamma_{ij}} \left(1-q 
\right)^{1-\gamma_{ij}} = \prod_{i=1}^{p-1} q^{\nu_i(\mathscr{D})} (1-q)^{p-i-
\nu_i(\mathscr{D})}. 
$$

\noindent
The model in (\ref{a3}) has three hyperparameters: the scale matrix $U$ (positive 
definite), the shape parameter vector ${\boldsymbol \alpha} (\mathscr{D})$, and the 
edge probability $q$. 

The hierarchical model in (\ref{a3}) can be used to estimate a DAG as follows. By 
(\ref{a3}) and Bayes' rule, the (marginal) posterior DAG probabilities are given by 
\begin{align} \label{m3}
\begin{split}
&\pi({\mathscr{D}}|\bm{Y}) \\
=& \int_{\Theta_{\mathscr{D}}} \frac{\pi(\bm{Y}| {\mathscr{D}},(L,D)) \pi_{U,\bm 
\alpha(\mathscr D)}^{\Theta_{\mathscr{D}}} ((L,D)) \pi({\mathscr{D}})}
{\pi(\bm{Y})} dLdD 
\\
=& \frac{\pi({\mathscr{D}})}{\pi(\bm{Y})}\int_{\Theta_{\mathscr{D}}}\pi(\bm{Y}| {\mathscr{D}},(L,D)) \pi((L,D)|{\mathscr{D}}))dLdD 
\\
=& \frac{\pi({\mathscr{D}})}{\pi(\bm{Y})}\int_{\Theta_{\mathscr{D}}} \frac{\exp\left(-\frac{1}{2} \mbox{tr}(LD^{-1}L^T(U+nS))\right) \prod_{i=1}^p D_{ii}^{-\left(n+\frac{\alpha_i(\mathscr D)}{2}\right)}}{z_{\mathscr{D}}(U,\bm{\alpha}(\mathscr D))}dLdD 
\\
=& \frac{\pi({\mathscr{D}})}{\pi(\bm{Y})(\sqrt{2\pi})^n} \frac{z_{\mathscr{D}}(U+nS,n+\bm{\alpha}(\mathscr D))}{z_{\mathscr{D}}(U,\bm{\alpha}(\mathscr D))}.
\end{split}
\end{align}

\noindent
Hence, the marginal posterior density $\pi({\mathscr{D}}|\bm{Y})$ is available in closed 
form (up to the multiplicative constant $\pi(\bm{Y})$). In particular, these posterior 
probabilities can be used to select a DAG by computing the posterior mode defined by
\begin{equation} \label{a4}
\hat{\mathscr D} =  \argmax_{\mathscr{D}} \pi({\mathscr{D}}|\bm{Y}).
\end{equation}

\section{DAG selection consistency: main results} \label{sec:modelselection}

\noindent
In this section we will explore the high-dimensional asymptotic properties of the 
Bayesian DAG selection approach specified in Section \ref{sec:modelspecification}. 
For this purpose, we will work in a setting where the dimension $p = p_n$ of the data 
vectors, and the edge probabilities $q = q_n$ vary with the sample size $n$. We 
assume that the data is actually being generated from a true model which can be 
specified as follows. Let $\bm{Y}_1^n, \bm{Y}_2^n, \ldots, \bm{Y}_n^n$ be 
independent and identically distributed $p_n$-variate Gaussian vectors with mean 
$0$ and covariance matrix $\Sigma_0^n = (\Omega_0^n)^{-1}$. Let $\Omega_0^n = 
L_0^n(D_0^n)^{-1}(L_0^n)^T$ be the modified Cholesky decomposition of 
$\Omega_0^n$. Let $\mathscr D_0^n$ be the true underlying DAG, i.e, $L_0^n \in \mathscr{L}_{\mathscr{D}_0^n}$. Denote $d_n$ as the maximum number of 
non-zero entries in any column of $L_0^n$, $s_n = \min_{1 \leq j \leq p, i \in pa_j(\mathscr{D}
_0^n)} |(L_0^n)_{ji}|$. Let $\bar P$ and $\bar{E}$ respectively denote the probability 
measure and expected value corresponding to the ``true" Gaussian DAG model 
presented above.

In order to establish our asymptotic results, we need the following mild regularity 
assumptions. Each assumption below is followed by an interpretation/discussion. 
Note that for a symmetric $p \times p$ matrix $A = (A_{ij})_{1\le i,j\le p}$, let $eig_1(A) 
\le eig_2(A) \ldots eig_p(A)$ denote the ordered eigenvalues of $A$. 
\begin{assumption}
	There exists $\epsilon_{0,n} \le 1$, such that for every $n \ge 1,$ $0 < \epsilon_{0,n} \le eig_1({\Omega}_0^n) \le eig_{p_n}({\Omega}_0^n) \le \epsilon_{0,n}^{-1}$, where $\frac{\left(\frac{\log p}{n}\right)^{\frac 1 2 - \frac 1 {2+k}}}{\epsilon_{0,n}^4} \rightarrow 0,$ as $n \rightarrow \infty,$ for some 
	$k > 0$.
\end{assumption}

\noindent
This is a much weaker assumption for high dimensional covariance asymptotics than for example, \cite{Bickel:Levina:2008, ElKaroui:2008, Banerjee:Ghosal:2014, XKG:2015, Banerjee:Ghosal:2015}. Here we allow the lower and upper bounds on the eigenvalues to depend on $p$ and $n$. 
\begin{assumption}
$d_n^{2+k} \sqrt{\frac{\log p_n}{n}}\rightarrow 0$, and
$\left( \sqrt{\frac{\log p_n}{n}} \right)^{\frac{k}{2(k+2)}} \log n \rightarrow 0$, as $n \rightarrow \infty$. 
\end{assumption}
\noindent
This assumption essentially states that the number of variables $p_n$ has to 
grow slower than $e^{n/d_n^{4+2k}}$ (and also $e^{n/(\log n)^{2+k}}$). Again, 
similar assumptions are common in high dimensional covariance asymptotics, 
see for example 
\cite{Bickel:Levina:2008, XKG:2015, Banerjee:Ghosal:2014, Banerjee:Ghosal:2015}. 

\begin{assumption}
	Let $q_n = e^{-\eta_nn}$ in (\ref{a3}), where $\eta_n = d_n (\frac{\log p_n}{n})^{\frac{1/2}{1+k/2}}$. Hence, $q_n \rightarrow 0$, as $ n \rightarrow \infty$.
\end{assumption}

\noindent
This assumption provides the rate at which the edge probability $q_n$ needs to 
approach zero. A similar assumption can be found in \cite{Narisetty:He:2014} in the 
context of linear regression. This can be interpreted as apriori penalizing graphs 
with a large number of edges. 
\begin{assumption}
	$\frac{\eta_nd_n}{\epsilon_{0,n}s_n^2} \rightarrow 0$ as $n \rightarrow \infty$.
\end{assumption}

\noindent
Recall that $s_n$ is the smallest (in absolute value) non-zero off-diagonal entry in 
$L_0^n$, and can be interpreted as the `signal size'. Hence, this assumption provides 
a lower bound for the signal size that is needed for establishing consistency. 
\begin{assumption}
For every $n \ge 1$, the hyperparameters for the DAG-Wishart prior 
$\pi_{U_n,\bm{\alpha}(\mathscr{D}_n)}^{\Theta_{\mathscr{D}_n}}$ in (\ref{a3}) satisfy 
(i) $2 < \alpha_i(\mathscr{D}_n) - \nu_i(\mathscr D_n) < c$ for every $\mathscr{D}_n$ 
and $1 \le i \le p_n$, and (ii) $0 < \delta_1 \le eig_1(U_n) \le eig_{p_n}(U_n) \le 
\delta_2 < \infty$. Here $c, \delta_1, \delta_2$ are constants not depending on $n$. 
\end{assumption}

\noindent
This assumption provides mild restrictions on the hyperparameters for the 
DAG-Wishart distribution. The assumption $2 < \alpha_i(\mathscr D) - 
\nu_i(\mathscr D)$ establishes prior propriety. The assumption $\alpha_i(\mathscr D) - 
\nu_i(\mathscr D) < c$ implies that the shape parameter $\alpha_i(\mathscr D)$ can only differ 
from $\nu_i(\mathscr D)$ (number of parents of $i$ in $\mathscr D$) by a constant which 
does not vary with $n$. Additionally, the eigenvalues of the 
scale matrix $U_n$ are assumed to be uniformly bounded in $n$. While the authors 
in \cite{BLMR:2016} do not specifically discuss hyperparameter choice, they do 
provide some recommendations in their experiments section. For the shape 
parameters, they recommend $\alpha_i(\mathscr{D}_n) = c \nu_i(\mathscr D_n) 
+ b$. They mostly use $c = 1$ in which case Assumption 5 is satisfied. The also 
use $c \in (2.5, 3.5)$ in some examples, in which case Assumption 5 is not satisfied. 
Also, they choose the scale matrix to be a constant multiple of the identity matrix, 
which clearly satisfies Assumption 5. 

For the rest of this paper, $p_n, {\Omega}_0^n, \Sigma_0^n, L_0^n, D_0^n, 
\mathscr{D}_0^n, \hat{\mathscr D}^n, \mathscr{D}^n, d_n, q_n, 
A_n$ will be denoted as $p, {\Omega}_0, \Sigma_0, L_0, D_0, \mathscr{D}_0, 
\hat{\mathscr D}, \mathscr{D}, d, q, A$ as needed for notational convenience 
and ease of exposition. 

We now state and prove the main DAG selection consistency results. Our first result 
establishes what we call as posterior ratio consistency. This notion of consistency 
implies that the true DAG will be the mode of the posterior DAG distribution with 
probability tending to $1$ as $n \rightarrow \infty$. 
\begin{theorem}[Posterior ratio consistency] \label{thm1}
Under Assumptions 1-5, the following holds: 
$$
\max_{{\mathscr{D}} \ne {\mathscr{D}}_0} \frac{\pi(\mathscr{D}|\bm{Y})}
{\pi({\mathscr{D}}_0|\bm{Y})} \stackrel{\bar{P}}{\rightarrow} 0, \mbox{ as } n \rightarrow \infty. 
$$
\end{theorem}

\noindent
A proof of this result is provided in Section \ref{sec:modelselectionproofs}. If one was interested in a point estimate of the 
underlying DAG using the Bayesian approach considered here, the most obvious choice would be the posterior mode 
$\hat{\mathscr D}$ defined in (\ref{a4}). From a frequentist point of view, it would be natural to enquire if we have 
model selection consistency, i.e., if $\hat{\mathscr D}$ is a consistent estimate of $\mathscr D_0$. In fact, the 
model selection consistency of the posterior mode follows immediately from posterior ratio 
consistency established in Theorem \ref{thm1}, by noting that 
$$
\max_{{\mathscr{D}} \ne {\mathscr{D}}_0} \frac{\pi(\mathscr{D}|\bm{Y})}
{\pi({\mathscr{D}}_0|\bm{Y})} < 1 \Rightarrow \hat{\mathscr D}= \mathscr D_0. 
$$

\noindent
We state this result formally in the corollary below. 
\begin{cor}[Model selection consistency for posterior mode] \label{cor1}
Under Assumptions 1-5, the posterior mode $\hat{\mathscr D}$ is equal to the 
true DAG $\mathscr D_0$ with probability tending to $1$, i.e., 
$$
\bar{P}(\hat{\mathscr D} = \mathscr D_0) \rightarrow 1, \mbox{ as } n \rightarrow \infty. 
$$ 
\end{cor}

\noindent
If $p$ is of a larger order than a positive power of $n$, then a stronger version of the posterior ratio consistency 
in Theorem \ref{thm1} can be established. 
\begin{theorem}[Almost sure posterior ratio consistency] \label{thm2}
If $p/n^{\widetilde{k}} \rightarrow \infty$ for some $\widetilde{k} > 0$, then under Assumption 1-5 the following holds: 
$$
\max_{{\mathscr{D}} \ne {\mathscr{D}}_0} \frac{\pi(\mathscr{D}|\bm{Y})}
{\pi({\mathscr{D}}_0|\bm{Y})} \rightarrow 0 \mbox{ almost surely } \bar{P}, 
$$

\noindent
as $n \rightarrow \infty$.
\end{theorem}

\noindent
Next we establish another stronger result (compared to Theorem \ref{thm1}) which implies that the posterior mass assigned to 
the true DAG $\mathscr{D}_0$ converges to 1 in probability (under the true model). Following \cite{Narisetty:He:2014}, we 
refer to this notion of consistency as strong selection consistency. To establish this stronger notion of consistency, we restrict 
our analysis to DAGs with total number of edges bounded by an appropriate function of $n$ (see also 
Remark \ref{edgerestrictassumption}). 
\begin{theorem}[Strong selection consistency] \label{thm3}
Under Assumptions 1-5, if we restrict only to DAG's with number of edges at 
most $\frac{1}{8} d \left( \frac{n}{\log p}\right)^{\frac{1+k}{2+k}}$, the following holds:
$$
\pi(\mathscr{D}_0 | \bm{Y}) \stackrel{\bar{P}}{\rightarrow} 1, \mbox{ as } n \rightarrow 
\infty. 
$$
\end{theorem}

 \begin{remark} \label{spikeslabregression}
 		In the context of linear regression, Narisetty and He \cite{Narisetty:He:2014} consider 
 		the following hierarchical Bayesian model. 
 		\begin{eqnarray*}
 			& & \bm Y \mid X, \bm{\beta}, \sigma^2 \sim N(X \bm{\beta}, \sigma^2 I)\\
 			& & \beta_i \mid \sigma^2, Z_i = 0 \sim N(0, \sigma^2 \tau_{0,n}^2)\\
 			& & \beta_i \mid \sigma^2, Z_i = 1 \sim N(0, \sigma^2 \tau_{1,n}^2)\\
 			& & P(Z_i = 1) = 1 - P(Z_i = 0) = q_n\\
 			& & \sigma^2 \sim \mbox{Inverse-Gamma}(\alpha_1, \alpha_2). 
 		\end{eqnarray*}
 	
 	\noindent
 	In particular, they put an independent spike and slab prior on each linear regression 
 			coefficient (conditional on the variance parameter $\sigma^2$), and an inverse Gamma prior on the variance. 
 			Also, each regression coefficient is present in the model with a probability $q_n$. In 
 			this setting, the authors in \cite{Narisetty:He:2014} establish strong selection 
 			consistency for the regression coefficients (assuming the prior is constrained to 
 			leave out unrealistically large models). There are similarities between the models and the consistency analysis 
			in \cite{Narisetty:He:2014} and this paper. Note that the off-diagonal entries in the $i^{th}$ 
 			column of $L$ are the linear regression coefficients corresponding to fitting the $i^{th}$ 
 			variable against all variables with label greater than $i$, and in our model (\ref{a3}) each coefficient is present 
			independently with a given probability $q_n$. Also, similar to \cite{Narisetty:He:2014}, in terms of proving posterior 
			consistency, we bound the ratio of posterior probabilities for a non-true model and the true model by a `prior term' 
			which is a power of $q_n/(1-q_n)$, and a `data term'. The consistency proof is then a careful exercise in 
			balancing these two terms against each other on a case-by-case basis. However, despite these similarities, there 
			are some fundamental differences in the two models and the corresponding analysis. Firstly, the 
 			DAG-Wishart prior does not in general correspond to assigning an independent spike 
 			and slab prior to each entry of $L$. The columns of $L$ are independent of each other 
 			under this prior, but it introduces correlations among the entries in each column of $L$. 
 			Also, the DAG-Wishart prior introduces exact sparsity in $L$, which is not the case in 
 			\cite{Narisetty:He:2014} as $\tau_{0,n}^2$ is assumed to be strictly positive.   
 			Hence it is structurally different than the prior in \cite{Narisetty:He:2014}. Secondly, 
 			the `design' matrices corresponding to the regression coefficients in each column 
 			of $L$ are random (they are functions of the sample covariance matrix $S$) and are 
 			correlated with each other. In particular, this leads to major differences and further 
 			challenges in analyzing the ratio of posterior graph probabilities (a crucial step in 
 			establishing consistency). 
 	
\end{remark}

\begin{remark} \label{edgerestrictassumption}
We would like to point out that posterior ratio consistency 
		(Theorems \ref{thm1} and \ref{thm2}) does not require any restriction on the maximum number of 
		edges, this requirement is only needed for strong selection consistency 
		(Theorem \ref{thm3}). Similar restrictions on the prior model size have been 
		considered for establishing consistency properties in other contexts. For 
		concentration graph models, Banerjee and Ghosal \cite{Banerjee:Ghosal:2015} use 
		a hierarchical prior where each edge of the concentration graph is independently 
		present with a given probability $q$. For establishing high-dimensional posterior 
		convergence rates, they restrict the prior to graphs with total number of edges 
		bounded by an appropriate fixed constant. A variation where the upper bound 
		on the number of edges is a random variable with sub-exponential tails is 
		also considered. For linear regression, Narisetty and He \cite{Narisetty:He:2014} 
		too restrict the prior model size to an appropriate function of $n$ (number of non-zero regression coefficients) for 
		establishing strong selection consistency (when the variance parameter is 
		random). 
\end{remark}

\section{Proof of Theorems \ref{thm1}, \ref{thm2} and \ref{thm3}} \label{sec:modelselectionproofs}

\noindent
The proof of Theorems \ref{thm1}, \ref{thm2} and \ref{thm3} will be broken up into various steps. 
We begin by presenting a useful lemma that provides an upper bound for the ratio of 
posterior DAG probabilities.
\begin{lemma} \label{newlemma1}
Under Assumption 5, for a large enough constant $M$ and large enough $n$, the ratio 
of posterior probabilities of any DAG $\mathscr{D}$ and the true DAG 
$\mathscr{D}_0$ satisfies:
	\begin{align*}
	\begin{split}
	&\frac{\pi({\mathscr{D}}|\bm{Y})}{\pi({\mathscr{D}}_0|\bm{Y})} \\
	\le& \prod_{i=1}^{p}M\left(\frac{\delta_2}{\delta_1}\right)^{\frac d 2} n^{2c}\left (\sqrt {\frac{\delta_2} n}\frac{q}{1-q}\right)^{\nu_i({\mathscr{D}}) - \nu_i({\mathscr{D}}_0)}  \frac{|\tilde{S}_{\mathscr{D}_0}^{\ge i}|^{\frac12}}{|\tilde{S}_{\mathscr{D}}^{\ge i}|^{\frac12}} \frac{\left(\tilde{S}_{i|pa_i({\mathscr{D}}_0)}\right)^{\frac{n+c_i(\mathscr D_0)-3}{2}}}{\left(\tilde{S}_{i|pa_i({\mathscr{D}})}\right)^{\frac{n+c_i(\mathscr D)-3}{2}}}\\
	\triangleq& \prod_{i=1}^{p}B_i(\mathscr{D},\mathscr{D}_0),
	\end{split}
	\end{align*}
	where $c_i(\mathscr D) = \alpha_i(\mathscr D) - \nu_i(\mathscr D), c_i(\mathscr D_0) = \alpha_i(\mathscr D_0) - \nu_i(\mathscr D_0)$, $\tilde{S} = S + \frac Un$, and 
$\tilde{S}_{i|pa_i({\mathscr{D}})} = \tilde{S}_{ii} - (\tilde{S}_{\mathscr D \cdot i}^>)^T 
(\tilde{S}_{\mathscr{D}}^{>i})^{-1} \tilde{S}_{\mathscr D \cdot i}^>$. 
\end{lemma} 

\noindent
The proof of this lemma is provided in the Supplemental Document. Our goal is to find an upper bound (independent of 
$\mathscr{D}$ and $i$) for $B_i(\mathscr{D},\mathscr{D}_0)$, such that the upper bound converges to $0$ as 
$n \rightarrow \infty$. By Lemma \ref{newlemma1}, this will be enough to establish Theorem \ref{thm1}. Before we 
undertake this goal, we present a proposition that will be useful in further analysis. Note that for any positive definite 
matrix $A$, and $M \subseteq \{1,2 \ldots, p\} \setminus \{i\}$, we denote $A_{i \mid M} = A_{ii} - 
A_{iM} A_{MM}^{-1} A_{Mi}$. 
\begin{proposition} \label{f1}
	Given a DAG $\mathscr D$ with $p$ vertices,  
	\begin{enumerate}[label=(\alph*)]
		\item If $pa_{i}(\mathscr D) \supseteq pa_{i}({\mathscr{D}}_0)$, then 
		$(\Sigma_0)_{i|pa_{i}(\mathscr D)} = (D_0)_{ii} = (\Sigma_0)_{i|pa_{i}({\mathscr{D}}_0)}$ doesn't depend on $\mathscr D$.
		\item If $pa_{i}(\mathscr{D}) \subseteq pa_i (\mathscr D _0)$, then 
		$(\Sigma_0)_{i|pa_{i}(\mathscr{D})} - (\Sigma_0)_{i|pa_i (\mathscr D _0)} \ge \epsilon_{0,n}(\nu_i(\mathscr{D}_0) - \nu_i(\mathscr{D}))s^2,$ where $\epsilon_{0,n} >0$ and $s = \min_{j \in pa_i (\mathscr D _0)}|(L_0)_{ji}|$.
	\end{enumerate}
\end{proposition}

\noindent
The proof of this proposition is provided in the Supplemental Document. Next, we show that in our setting, the sample and 
population covariance matrices are sufficiently close with high probability. It follows by Assumptions 1,2,5, Lemma A.3 
of \cite{Bickel:Levina:2008} and Hanson-Wright inequality from \cite{rudelson2013} that there exists constants $m_1,m_2$ and $\delta$ 
depending on $\epsilon_{0,n}$ only such that for $1 \le i,j \le p$, we have:
\begin{equation*}
\bar{P}(| S_{ij} - (\Sigma_0)_{ij} | \ge t) \le m_1 \exp\{-m_2n(t\epsilon_{0,n})^2\}, \, |t| \le \delta.
\end{equation*}

\noindent
By the union-sum inequality, for a large enough $c'$ such that $2-m_2(c')^2/4 < 0$, we 
get that 
\begin{equation} \label{smplbound1}
\bar{P}\left(\Vert \tilde{S}-\Sigma_0 \Vert_{\max} \ge c' \sqrt{\frac{\log p}{n\epsilon_{0,n}^2}}\right) 
\leq m_1p^{2-m_2 (c')^2/4} \rightarrow 0. 
\end{equation}

\noindent
Define the event $C_n$ as 
\begin{equation} \label{smplbound2}
C_n = \left\{\Vert \tilde{S}-\Sigma_0 \Vert_{\max} \ge c' \sqrt{\frac{\log p}{n\epsilon_{0,n}^2}}\right\}. 
\end{equation}

\noindent
It follows from (\ref{smplbound1}) and (\ref{smplbound2}) that $\bar{P}(C_n) 
\rightarrow 0$ as $n \rightarrow \infty$. 

We now analyze the behavior of $B_i(\mathscr{D},\mathscr{D}_0)$ under different 
scenarios in a sequence of five lemmas (Lemmas \ref{lm1} - \ref{lm5}). Recall that our 
goal is to find an upper bound (independent of $\mathscr{D}$ and $i$) for 
$B_i(\mathscr{D},\mathscr{D}_0)$, such that the upper bound converges to $0$ as $n 
\rightarrow \infty$. {\bf For all these lemmas, we will restrict ourselves to the event 
$C_n^c$}. 
\begin{lemma} \label{lm1}
If $pa_i({\mathscr{D}}) \supset pa_i({\mathscr{D}}_0)$ and $\nu_i({\mathscr{D}}) 
\le 3 \nu_i({\mathscr{D}}_0) + 2$, then there exists $N_1$ (not depending on $i$ or $
\mathscr{D}$) such that for $n \geq N_1$ we have $B_i({\mathscr{D}},{\mathscr{D}}_0) 
\le \epsilon_{1,n}$, where $\epsilon_{1,n}=2e^{-\frac{\eta_n}{2}n}$.
\end{lemma}
\begin{proof}
Since $pa_i (\mathscr D _0) \subset pa_{i}(\mathscr{D})$, we can write 
$|\tilde{S}_{{\mathscr{D}} }^{\ge i}| = |\tilde{S}_{\mathscr{D}_0}^{\ge i}| 
|R_{{\tilde{S}}_{\mathscr{D}_0}^{\ge i}}|$. Here 
$R_{{\tilde{S}}_{\mathscr{D}_0}^{\ge i}}$ is the Schur complement of 
${\tilde{S}}_{\mathscr{D}_0}^{\ge i}$, defined by 
$$
R_{{\tilde{S}}_{\mathscr{D}_0}^{\ge i}} = D - B^T ({\tilde{S}}_{\mathscr{D}_0}^{\ge 
i})^{-1} B 
$$ 

\noindent
for appropriate sub matrices $B$ and $D$ of $\tilde{S}_{{\mathscr{D}} }^{\ge i}$. Since 
${\tilde{S}}_{{\mathscr{D}}}^{\ge i} \geq 
\left(\frac U n\right)_{{\mathscr{D}}}^{\ge i}$~\footnote{For matrices $A$ and $B$, we 
say $A \ge B$ if $A - B$ is positive semi-definite}, and 
$R_{{\tilde{S}}_{\mathscr{D}_0}^{\ge i}}^{-1}$ is a principal submatrix of 
$\left( \tilde{S}_{{\mathscr{D}} }^{\ge i} \right)^{-1}$, it follows from Assumption 5 that 
the largest eigenvalue of $R_{{\tilde{S}}_{\mathscr{D}_0}^{\ge i}}^{-1}$ is bounded 
above by $\frac {n}{\delta_1}$. Therefore,
\begin{equation} \label{new7}
\left(\frac{|\tilde{S}_{{\mathscr{D}_0}}^{\ge i}|}{|\tilde{S}_{\mathscr{D}}^{\ge i}|} \right)^{\frac12} = {|R_{{\tilde{S}}_{\mathscr{D}_0}^{\ge i}}^{-1}|^{1/2}} 
	\le \left(\sqrt{\frac {n} {\delta_1}}\right)^{\nu_i({\mathscr{D}}) - \nu_i({\mathscr{D}_0})}.
	\end{equation}

\noindent
Since we are restricting ourselves to the event $C_n^c$, it follows by 
(\ref{smplbound1}) that 
\begin{equation*}
|| \tilde{S}_{\mathscr{D}_0}^{\ge i} - (\Sigma_0)_{\mathscr{D}_0}^{\ge i} ||_{(2,2)} \le 
(\nu_i(\mathscr{D}_0 )+1)c^\prime \sqrt{\frac{\log p}{n\epsilon_{0,n}^2}}.
\end{equation*}
	Therefore,
	\begin{align} \label{pp1}
	\begin{split}
	&|| (\tilde{S}_{\mathscr{D}_0}^{\ge i})^{-1} - ((\Sigma_0)_{\mathscr{D}_0}^{\ge i})^{-1} ||_{(2,2)} \\ 
	=& || (\tilde{S}_{\mathscr{D}_0}^{\ge i})^{-1}||_{(2,2)} || (\tilde{S}_{\mathscr{D}_0}^{\ge i})^{-1} - ((\Sigma_0)_{\mathscr{D}_0}^{\ge i})^{-1} ||_{(2,2)} || 
	((\Sigma_0)_{\mathscr{D}_0}^{\ge i})^{-1} ||_{(2,2)} \\
	\le &	(|| (\tilde{S}_{\mathscr{D}_0}^{\ge i})^{-1} - ((\Sigma_0)_{\mathscr{D}_0}^{\ge i})^{-1} ||_{(2,2)} + \frac 1 {\epsilon_{0,n}})(\nu_i(\mathscr{D}_0)+1)c^\prime \sqrt{\frac{\log p}{n\epsilon_{0,n}^2}}\frac 1 {\epsilon_{0,n}}, 
	\end{split}
	\end{align}

\noindent
By Assumptions 1,2 and $d > 0$, we have 
\begin{equation} \label{assumptioneigenvalue}
\frac{d\sqrt{\frac{\log p_n}{n}}}{\epsilon_{0,n}^4} \rightarrow 0, \mbox{ as } n \rightarrow \infty.
\end{equation}
Hence, there exists $N_1'$ such that for $n \geq N_1'$, 
$$
\frac {c^\prime} {\epsilon_{0,n}^2}(d +1)\sqrt{\frac{\log p }n} < \frac12, \mbox{ and }
2\frac {c^\prime} {\epsilon_{0,n}^3}(d + 1)\sqrt{\frac{\log p }{n}} < \frac 1 {\epsilon_{0,n}}. 
$$

\noindent
Since $\nu_i (\mathscr{D}_0) \leq d$, it follows by (\ref{pp1}) and Assumption 1 that 
$$
||(\tilde{S}_{\mathscr{D}_0}^{\ge i})^{-1}||_{(2,2)} \leq \frac{2}{\epsilon_{0,n}}, 
$$

\noindent
and 
	\begin{align} \label{nnew9}
	\frac 1 {\tilde{S}_{i|pa_i({\mathscr{D}_0})}} = \left[(\tilde{S}_{\mathscr{D}_0}^{\ge i})^{-1}\right]_{ii}
	\le \frac{2}{\epsilon_{0,n}}.
	\end{align}

\noindent
for $n \geq N_1'$. Since, $pa_i (\mathscr D _0) \subset pa_{i}(\mathscr{D})$, we get
$$
\tilde{S}_{i|pa_i({\mathscr{D}_0})}  \ge \tilde{S}_{i|pa_i({\mathscr{D}})}. 
$$

\noindent
Let $N_1''$ be such that for $n \geq N_1''$, $q \leq \frac{\sqrt{\delta_1}}
{2\sqrt{\delta_2}} \leq \frac{1}{2}$. Using $2 < c_i(\mathscr D), c_i(\mathscr D_0) < c$, 
(\ref{new7}), (\ref{nnew9}), and Lemma \ref{newlemma1}, we get 
	\begin{align} \label{new8}
	\begin{split}
	B_i(\mathscr{D} ,\mathscr{D}_0) 
	\le& M\left(\frac{\delta_2}{\delta_1}\right)^{\frac d 2} n^{2c}\left(\sqrt {\frac{\delta_2} {\delta_1}}\frac {q}{1-q}\right)^{\nu_i({\mathscr{D}})-\nu_i({\mathscr{D}}_0)} \left(\frac 2 {\epsilon_{0,n}}\right)^c \\
	&\times \left (\frac{\tilde{S}_{i|pa_i({\mathscr{D}_0} )}}{\tilde{S}_{i|pa_i({\mathscr{D}})}} \right )^{\frac{n+c-3}2} \\
	\le& M \left(\frac 2 {\epsilon_{0,n}}\right)^c \left(\frac{\delta_2}{\delta_1}\right)^{\frac d 2} n^{2c} \left(\sqrt {\frac{\delta_2} {\delta_1}} 2q\right)^{\nu_i({\mathscr{D}})-\nu_i({\mathscr{D}}_0)}
	\left (\frac {\frac 1 {\tilde{S}_{i|pa_i({\mathscr{D}})}}} {\frac 1 {\tilde{S}_{i|pa_i({\mathscr{D}_0} )}}} \right) ^{\frac{n+c-3}2}, 
	\end{split}
	\end{align}

\noindent
for $n \geq \max(N_1', N_1'')$. We would like to note that the arguments leading up 
to (\ref{new8}) only require the assumption $pa_i ({\mathscr{D}_0}) \subset pa_{i}
(\mathscr{D})$. This observation enables us to use (\ref{new8}) in the proof of 
Lemma \ref{lm2} and Lemma \ref{lm3}. 
	
By following exactly the same sequence of arguments leading up to (\ref{pp1}), 
and replacing $\mathscr{D}$ by $\mathscr{D}_0$, we get 
\begin{align} \label{pp1general}
	&|| (\tilde{S}_{\mathscr{D}}^{\ge i})^{-1} - ((\Sigma_0)_{\mathscr{D}}^{\ge i})^{-1} ||_{(2,2)} \\ 
	\le &	(|| (\tilde{S}_{\mathscr{D}}^{\ge i})^{-1} - ((\Sigma_0)_{\mathscr{D}}^{\ge i})^{-1} ||_{(2,2)} + \frac 1 {\epsilon_{0,n}})(\nu_i(\mathscr{D} )+1)c^\prime \sqrt{\frac{\log p}{n\epsilon_{0,n}^2}}\frac 1 {\epsilon_{0,n}}, 
	\end{align}

\noindent
By (\ref{assumptioneigenvalue}), there exists $N_1'''$ such that for $n \geq N_1'''$, 
\begin{equation} \label{new6}
\frac {c^\prime} {\epsilon_{0,n}^2}(3d +3)\sqrt{\frac{\log p }n} < \frac12, \mbox{ and }
2\frac {c^\prime} {\epsilon_{0,n}^3}(3d + 3)\sqrt{\frac{\log p }n} < \frac{\epsilon_{0,n}}{2}. 
\end{equation}

\noindent
Note that by hypothesis $\nu_i(\mathscr D) + 1 \leq 3\nu_i(\mathscr D_0) + 3 \leq 3d + 
3$. It follows from (\ref{pp1general}) that 
\begin{equation} \label{pp7}
|| (\tilde{S}_{\mathscr{D}}^{\ge i})^{-1} - ((\Sigma_0)_{\mathscr{D}}^{\ge i})^{-1} ||_{(2,2)} \le 
2\frac {c^\prime} {\epsilon_{0,n}^3}(3d +3)\sqrt{\frac{\log p }n} .
\end{equation}

\noindent
for $n \geq N_1'''$. Using $\nu_i (\mathscr{D}) - \nu_i (\mathscr{D}_0) \geq 1$, 
(\ref{new8}), (\ref{pp7}), Proposition \ref{f1} (a) and the definition of $q_n$, it follows 
that for $n \geq \max(N_1', N_1'', N_1''')$, 
\begin{align*}
&B_i(\mathscr{D} ,\mathscr{D}_0 ) \\
	\le& 2 \tilde{M}\frac 1 {\epsilon_{0,n}^c}\left(\frac{\delta_2}{\delta_1}\right)^{\frac d 2} n^{2c} q\left (\frac{|| ((\Sigma_0)_{\mathscr{D}_0}^{\ge i})^{-1} ||_{(2,2)} + 2 \frac {c^\prime} {\epsilon_{0,n}^3}(3d +3)\sqrt{\frac{\log p }n}}{||(\Sigma_0)_{\mathscr{D}_0}^{\ge i})^{-1}||_{(2,2)} - 2\frac {c^\prime} {\epsilon_{0,n}^3}(3d +3)\sqrt{\frac{\log p }n}} \right)^{\frac{n-c+3}2}\\
	\le& 2\exp\left\{-d \left(\frac{\log p}{n}\right)^{\frac{1/2}{1+k/2}} n + d \log\left(\frac{\delta_2}{\delta_1}\right) + 2c\log n + \frac c 4 \log\left(\frac 1 {\epsilon_{0,n}^4}\right) +\log \tilde{M}\right\} \\
	&\times \left (1+\frac{2\frac {c^\prime} {\epsilon_{0,n}^3}(3d +3)\sqrt{\frac{\log p }n}}{|| ((\Sigma_0)_{\mathscr{D}_0}^{\ge i})^{-1} ||_{(2,2)} - \frac {\epsilon_{0,n}}2} \right)^n \\
	\le& 2\exp\left\{-d \left(\frac{\log p}{n}\right)^{\frac{1/2}{1+k/2}} n + d \log\left(\frac{\delta_2}{\delta_1}\right) + 2c\log n + \frac c 4 \log\left(\frac 1 {\epsilon_{0,n}^4}\right) + \log \tilde{M} \right\}\\
	&\times \exp \left\{\frac {12c^\prime} {\epsilon_{0,n}^4}(d +1)\sqrt{n\log p }\right\},
\end{align*}

\noindent
where $\tilde{M} = M2^c\sqrt{\frac{\delta_2}{\delta_1}}$. Since $\eta_n = d \left(\frac{\log p}{n}\right)^{\frac{1/2}
{1+k/2}}$ has a strictly larger order than $\frac{d}{n}$, $\frac{\log n}{n}$,$\frac{\log\left(\frac 1 {\epsilon_{0,n}^4}\right)}{n}$  and $d\frac{\sqrt{\frac{\log p}{n}}}{\epsilon_{0,n}^4}$~\footnote{We say $a_n$ is of a larger order than $b_n$ if 
$\frac{b_n}{a_n} \rightarrow 0$ as $n \rightarrow \infty$} by Assumptions 1 and 2, it follows that there exists 
$N_1''''$ such that for $n \geq N_1''''$, the expression in the exponent is dominated by 
$-\frac{\eta_n}{2}$. It follows that 
$$
B_i(\mathscr{D} ,\mathscr{D}_0 ) \leq 2e^{-\frac{\eta_n}{2} n} 
$$

\noindent
for $n \geq N_1 \stackrel{\Delta}{=} \max(N_1', N_1'', N_1''', N_1'''')$. 
\end{proof}

\begin{lemma} \label{lm2}
Assume $pa_i({\mathscr{D}}) \supset pa_i({\mathscr{D}}_0), \nu_i({\mathscr{D}}) > 3 
\nu_i({\mathscr{D}}_0) + 2$ and $\frac 1 {\epsilon_{0,n}^2}(\nu_i({\mathscr{D}})+1) \sqrt{\frac{\log p}{n}} \le 
\frac{1}{2{c^\prime}},$ then there exists $N_2$ (not depending on $i$ or $
\mathscr{D}$), such that for $n \geq N_2$, $B_i({\mathscr{D}} ,{\mathscr{D}}_0) \le 
\epsilon_{2,n},$ where $\epsilon_{2,n} = e^{-\eta_n n}.$
\end{lemma}

\begin{proof}
By following exactly the same sequence of arguments leading up to (\ref{pp1}), 
and replacing $\mathscr{D}$ by $\mathscr{D}_0$, we get 
\begin{align*}
	&|| (\tilde{S}_{\mathscr{D}}^{\ge i})^{-1} - ((\Sigma_0)_{\mathscr{D}}^{\ge i})^{-1} ||_{(2,2)} \\ 
	\le &	(|| (\tilde{S}_{\mathscr{D}}^{\ge i})^{-1} - ((\Sigma_0)_{\mathscr{D}}^{\ge i})^{-1} ||_{(2,2)} + \frac 1 {\epsilon_{0,n}})(\nu_i(\mathscr{D} )+1)c^\prime \sqrt{\frac{\log p}{n\epsilon_{0,n}^2}}\frac 1 {\epsilon_{0,n}}, 
\end{align*}

\noindent
Using $\frac 1 {\epsilon_{0,n}^2}(\nu_i({\mathscr{D}})+1) \sqrt{\frac{\log p}{n}} \le 
\frac{1}{2{c^\prime}}$, $\nu_i (\mathscr{D}_0) < \nu_i (\mathscr{D})$, (\ref{pp1}) and (\ref{assumptioneigenvalue}), for large enough $n \ge N_2'$, we get
\begin{equation} \label{pp7.1}
|| (\tilde{S}_{\mathscr{D}}^{\ge i})^{-1} - ((\Sigma_0)_{\mathscr{D}}^{\ge i})^{-1} ||_{(2,2)} \le 
\frac{2c^\prime} {\epsilon_{0,n}^3}(\nu_i({\mathscr{D}})+1)\sqrt{\frac{\log p }{n}},
\end{equation}

\begin{equation} \label{e1}
|| (\tilde{S}_{\mathscr{D}_0}^{\ge i})^{-1} - 
((\Sigma_0)_{\mathscr{D}_0}^{\ge i})^{-1} ||_{(2,2)} \le \frac{2c^\prime} {\epsilon_{0,n}^3}
(\nu_i({\mathscr{D}_0})+1)\sqrt{\frac{\log p }{n}}
\end{equation}

\noindent
and
\begin{equation} \label{e2}
\frac {2c^\prime} {\epsilon_{0,n}^3}(\nu_i({\mathscr{D}}_0 )+1)\sqrt{\frac{\log p }{n}} \le 
\frac{\epsilon_{0,n}}{2}. 
\end{equation}

\noindent
Note that the arguments leading up to (\ref{new8}) only use $pa_i (\mathscr{D}_0) 
\subset pa_i (\mathscr{D})$. It follows from (\ref{new8}), Proposition \ref{f1}, 
(\ref{pp7.1}) and (\ref{e1}) that these exists $N_2''$ such that 
\begin{align*}
&B_i({\mathscr{D}} ,{\mathscr{D}}_0)\\ 
\le& \exp\left\{d \log\left(\frac{\delta_2}{\delta_1}\right) + 2c\log n + \frac c 4 \log\left(\frac 1 {\epsilon_{0,n}^4}\right) + \log \tilde{M}\right\}\left( 2q \sqrt{\frac{\delta_1}{\delta_2}} \right)^{\nu_i({\mathscr{D}})-\nu_i({\mathscr{D}}_0 )}\\
&\times \left(\frac{\frac{1}{\tilde{S}_{i|pa_i({\mathscr{D}})}}}{\frac{1}{\tilde{S}_{i|pa_i({\mathscr{D}}_0)}}} \right)^{\frac{n+c-3}{2}} \\
	\le& \exp\left\{d \log\left(\frac{\delta_2}{\delta_1}\right) + 3c\log n\right\}\left( 2q \sqrt{\frac{\delta_1}{\delta_2}} \right)^{\nu_i({\mathscr{D}} )-\nu_i({\mathscr{D}}_0)}\\
	&\times \left(\frac{\frac{1}{(\Sigma_0)_{i|pa_i({\mathscr{D}_0} )}}+\frac{2c^\prime} {\epsilon_{0,n}^3}(\nu_i({\mathscr{D}})+1)\sqrt{\frac{\log p }{n}}}{\frac{1}{(\Sigma_0)_{i|pa_i({\mathscr{D}}_0)}}-\frac {2c^\prime} {\epsilon_{0,n}^3}(\nu_i({\mathscr{D}}_0 )+1)\sqrt{\frac{\log p }{n}}} \right)^{\frac{n+c-3}{2}} 
\end{align*}

\noindent
for $n \geq N_2''$. Note that $\nu_i({\mathscr{D}}) > 3 \nu_i({\mathscr{D}}_0) + 2$ 
implies $\nu_i({\mathscr{D}} ) + \nu_i({\mathscr{D}}_0 ) + 2 \le 2(\nu_i({\mathscr{D}} )-
\nu_i({\mathscr{D}}_0 ))$ It follows by Assumption 1, (\ref{e2}) and $q = q_n = 
e^{-\eta_n n}$ that 
\begin{align*}
&B_i({\mathscr{D}} ,{\mathscr{D}}_0)\\
	\le& \exp\left\{d \log\left(\frac{\delta_2}{\delta_1}\right) + 3c\log n\right\}\left( 2q \sqrt{\frac{\delta_1}{\delta_2}} \right)^{\nu_i({\mathscr{D}})-\nu_i({\mathscr{D}}_0)} \\
	&\times \left(1 + \frac{\frac {2c^\prime} {\epsilon_{0,n}^3}(\nu_i({\mathscr{D}} )+\nu_i({\mathscr{D}}_0)+2)\sqrt{\frac{\log p }{n}}}{\epsilon_{0,n}/2} \right )^{\frac{n+c-3}{2}} \\
	\le& \exp\left\{d \log\left(\frac{\delta_2}{\delta_1}\right) + 3c\log n\right\}\left(2q \sqrt{\frac{\delta_1}{\delta_2}} \right)^{\nu_i({\mathscr{D}})-\nu_i({\mathscr{D}}_0)} \\
	& \times \exp\left\{\frac{8{c^\prime}}{\epsilon_{0,n}^4}(\nu_i({\mathscr{D}} )-\nu_i({\mathscr{D}}_0 ))\sqrt{n \log p}\right\} \\
	\le& \left(2 \sqrt{\frac{\delta_1}{\delta_2}} \exp\left\{-\eta_n n+\frac{8{c^\prime}}{\epsilon_{0,n}^4}\sqrt{n\log p } + d \log\left(\frac{\delta_2}{\delta_1}\right) + 3c\log n\right\} \right)^{\nu_i({\mathscr{D}} )-\nu_i({\mathscr{D}}_0 )}. 
\end{align*}

\noindent
Since $\eta_n$ has a strictly larger order than $\frac{\sqrt{\frac{\log p}{n}}}{\epsilon_{0,n}^4}$, $\frac{d}{n}$ 
and $\frac{\log n}{n}$, there exists $N_2$ such that 
$$
B_i({\mathscr{D}} ,{\mathscr{D}}_0) \le \left(e^{-\frac{\eta_n}{2}n} 
\right)^{\nu_i({\mathscr{D}} )-\nu_i({\mathscr{D}}_0 )} \le e^{-\eta_n n} 
$$

\noindent
for $n \geq N_2$. 
\end{proof}

\begin{lemma} \label{lm3}
If $pa_i({\mathscr{D}}) \supset pa_i({\mathscr{D}}_0)$, $\nu_i({\mathscr{D}}) > 
3\nu_i({\mathscr{D}}_0) + 2$ and $\frac 1 {\epsilon_{0,n}^2}(\nu_i({\mathscr{D}}) + 1)\sqrt{\frac{\log p}{n}} > 
\frac1{2{c^\prime}}$, then there exists $N_3$(not depending on $i$ or 
$\mathscr{D}$), such that for $n \geq N_3$, $B_i({\mathscr{D}} , {\mathscr{D}} _0) \le 
\epsilon_{3,n}$, where $\epsilon_{3,n} = (\frac{1}{\delta_1 n})^n.$	
\end{lemma}
\begin{proof}		
Since $\frac 1 {\tilde{S}_{i|pa_i({\mathscr{D}})}} = \left[(\tilde{S}_{\mathscr{D}}^{\ge 
i})^{-1}\right]_{ii}$ and $\tilde{S}_{\mathscr{D}}^{\ge i} \ge \left(\frac U n
\right)_{\mathscr{D}}^{\ge i}$, we get 
$$
\frac 1 {\tilde{S}_{i|pa_i({\mathscr{D}})}} \le \frac n {\delta_1}. 
$$

\noindent
By (\ref{pp1}) and (\ref{assumptioneigenvalue}), there exists $N_3'$ such that for $n \geq N_3'$, 
$$
\frac {c^\prime} {\epsilon_{0,n}^2}(d +1)\sqrt{\frac{\log p }n} < \frac12, \mbox{ and }
2\frac {c^\prime} {\epsilon_{0,n}^3}(d + 1)\sqrt{\frac{\log p }n} < 
\frac{\epsilon_{0,n}}{2}. 
$$

\noindent
Since $\nu_i (\mathscr{D}_0) \leq d$, it follows by (\ref{pp1}) and Assumption 1 that 
$$
||(\tilde{S}_{\mathscr{D}_0}^{\ge i})^{-1}||_{(2,2)} \geq \frac{\epsilon_{0,n}}{2}, 
\mbox{ and } \frac 1 {\tilde{S}_{i|pa_i({\mathscr{D}_0})}} = \left[ 
(\tilde{S}_{\mathscr{D}_0}^{\ge i})^{-1} \right]_{ii} \geq \frac{\epsilon_{0,n}}{2}. 
$$

\noindent
for $n \geq N_3'$. Note that by hypothesis, we have 
$$
\nu_i({\mathscr{D}}) > \frac{\epsilon_{0,n}^2}{2{c^\prime}}\sqrt{\frac{n}{\log p}} - 1. 
$$

\noindent
Since the arguments leading up to (\ref{new8}) only require $pa_i (\mathscr{D}) 
\supset pa_i (\mathscr{D}_0)$, using the above facts along with Assumption 2, there 
exists $N_3''$ such that for $n \geq N_3''$, we get 
\begin{align*}
B_i({\mathscr{D}} , {\mathscr{D}}_0 ) & \le \exp\left\{d \log\left(\frac{\delta_2}{\delta_1}\right) + 
3c\log n\right\}\left( 2q \sqrt{\frac{\delta_1}{\delta_2}} 
\right)^{\frac{\epsilon_{0,n}^2}{2{c^\prime}}\sqrt{\frac{n}{\log p}} - 2d}\left(\frac{2n}{\delta_1 
\epsilon_{0,n}} \right)^{\frac{n+c-3}{2}}\\
&\le \exp\left\{d \log\left(\frac{\delta_2}{\delta_1}\right) + 3c\log n\right\}\left (2q 
\sqrt{\frac{\delta_1}{\delta_2}} \right)^{\frac{\epsilon_{0,n}^2}{2{c^\prime}}\sqrt{\frac{n}{\log 
p}} - 2d}\left(\frac{2n}{\delta_1 \epsilon_{0,n}}\right)^n\\
&\le \left(\frac{n^2}{\epsilon_{0,n} \delta_1} \right)^n (2q)^{\frac{\epsilon_{0,n}^2}{2{c^\prime}}
\sqrt{\frac{n}{\log p}} - 2d} \exp\left\{d\log{\left(\frac{\delta_2}{\delta_1}\right)} + 3c\log n\right\}\\
&= \left (\frac{1}{\delta_1}\exp\left\{-\frac{\epsilon_{0,n}^2}{{2c^\prime}} \eta_n  
\sqrt{\frac{n}{\log p }} + 2 \eta_n d + (2 + 3c)\log n + d\log{\left(\frac{\delta_2}{\delta_1}\right)}+ \frac 1 4 \log\left(\frac 1 {\epsilon_{0,n}^4}\right)\right\} \right)^n. 
\end{align*}

\noindent
By Assumption 1, we have $\frac 1 {\epsilon_{0,n}^2} = o\left(\left(\frac {\log p} n\right)^{-\frac 1 2\left(\frac 1 2 - \frac 1 {2+k}\right)}\right)$. Then, by (\ref{assumptioneigenvalue}), Assumptions 2 and 3, we obtain $ \epsilon_{0,n}^2 \eta_n  
\sqrt{\frac{n}{\log p }}$ has 
a larger order than $\eta_n d$, $\log n$ and $\log\left(\frac 1 {\epsilon_{0,n}^4}\right)$. It follows that there exists $N_3$ such 
that 
\begin{align*}
B_i({\mathscr{D}} , {\mathscr{D}}_0 ) & \leq \left(\frac{1}{\delta_1}\exp\left\{- 
\frac{\epsilon_{0,n}^2}{{4c^\prime}} \eta_n  
\sqrt{\frac{n}{\log p }} \right\} \right)^n\\
&\leq \left(\frac{1}{\delta_1}\exp\left\{-\log n\right\} \right)^n\\
&= \left(\frac{1}{\delta_1 n} \right)^n 
\end{align*}

\noindent
for $n \geq N_3$. 
\end{proof}

\begin{lemma} \label{lm4}
If $pa_i({\mathscr{D}}) \subset pa_i({\mathscr{D}}_0)$, then there exists $N_4$ (not 
depending on $i$ or $\mathscr{D}$), such that for $n \geq N_4$, $B_i({\mathscr{D}} , 
{\mathscr{D}}_0) \le \epsilon_{4,n}$, where $\epsilon_{4,n} = e^{-d\eta_n n}.$
\end{lemma}

\begin{proof}
Since $pa_i (\mathscr D _0) \supset pa_{i}(\mathscr{D})$, we can write 
$|\tilde{S}_{{\mathscr{D}_0}}^{\ge i}| = |\tilde{S}_{\mathscr{D}}^{\ge i}| 
|R_{{\tilde{S}}_{\mathscr{D}}^{\ge i}}|$. Here $R_{{\tilde{S}}_{\mathscr{D}}^{\ge i}}$ is 
the Schur complement of ${\tilde{S}}_{\mathscr{D}}^{\ge i}$, defined by 
$$
R_{{\tilde{S}}_{\mathscr{D}}^{\ge i}} = \tilde{D} - \tilde{B}^T 
({\tilde{S}}_{\mathscr{D}}^{\ge i})^{-1} \tilde{B} 
$$ 

\noindent
for appropriate sub matrices $\tilde{B}$ and $\tilde{D}$ of 
$\tilde{S}_{{\mathscr{D}_0}}^{\ge i}$. It follows by (\ref{pp1}) that if restrict to $C_n^c$, 
$$
||(\tilde{S}_{{\mathscr{D}_0}}^{\ge i})^{-1} - 
({(\Sigma_0)}_{{\mathscr{D}_0}}^{\ge i})^{-1}||_{(2,2)} \le \frac{4c'}{\epsilon_{0,n}^3} d 
\sqrt{\frac{\log p }{n}}, 
$$

\noindent
for $n > N_4'$. It follows that 
$$
||R_{{\tilde{S}}_{\mathscr{D}}^{\ge i}}^{-1} - 
R_{(\Sigma_0)_{\mathscr{D}}^{\ge i}}^{-1}||_{(2,2)} \le \frac{4c'}{\epsilon_{0,n}^3} d 
\sqrt{\frac{\log p }{n}}, 
$$

\noindent
where $R_{(\Sigma_0)_{\mathscr{D}}^{\ge i}}$ represents the Schur complement of 
$(\Sigma_0)_{\mathscr{D}}^{\ge i}$ defined by 
$$
R_{(\Sigma_0)_{\mathscr{D}}^{\ge i}} = \bar{D} - \bar{B}^T 
((\Sigma_0)_{\mathscr{D}}^{\ge i})^{-1} \bar{B} 
$$ 

\noindent
for appropriate sub matrices $\bar{B}$ and $\bar{D}$ of 
$(\Sigma_0)_{{\mathscr{D}_0}}^{\ge i}$. Let $\lambda_{min} (A)$ denote the smallest 
eigenvalue of a positive definite matrix $A$. By Assumptions 1 and 2, it follows that 
there exists $N_4''$ such that 
\begin{align*}
\left(\frac{|\tilde{S}_{{\mathscr{D}}_0 }^{\ge i}|}{|\tilde{S}_{\mathscr{D}}^{\ge i}|} \right)^{\frac12} = \frac{1}{|R_{{\tilde{S}}_{\mathscr{D}}^{\ge i}}^{-1}|^{1/2}} 
	&\le \frac{1}{\left(\lambda_{\mbox{min}}\left(R_{{(\Sigma_0)}_{\mathscr{D}}^{\ge i}}^{-1}\right)- K\frac d {\epsilon_{0,n}^3}\sqrt{\frac{\log p }{n}} \right)^{\frac{\nu_i({\mathscr{D}}_0) - \nu_i({\mathscr{D}})}{2}}}\\
	&\le \left(\frac{1}{\epsilon_{0,n}/2} \right)^{\frac{\nu_i({\mathscr{D}}_0 ) - \nu_i({\mathscr{D}} )}{2}} \mbox{ for large enough $n$}.
	\end{align*}
	
\noindent
for $n \geq N_4''$. Since $pa_{i}(\mathscr{D}) \subset pa_i (\mathscr D _0)$, we get
$$
\tilde{S}_{i|pa_i({\mathscr{D}_0})}  \le \tilde{S}_{i|pa_i({\mathscr{D}})}. 
$$

\noindent
Let $K_1 = 4c'$. By Lemma \ref{newlemma1} and 
Proposition \ref{f1}, and $2 < c_i (\mathscr{D}), c_i (\mathscr{D}_0) < c$, it follows 
that there exists $N_4'''$ such that for $n \geq N_4'''$, we get 
\begin{align} \label{pp6}
\begin{split}
&B_i({\mathscr{D}} ,{\mathscr{D}}_0 ) \\
\le& M \left( \frac{2}{\epsilon_{0,n}} \right)^c \left(\frac{\delta_2}
{\delta_1}\right)^{\frac d 2} n^{2c}\left(\sqrt{\frac{2n}{\delta_2\epsilon_{0,n}}} q^{-1} 
\right)^{\nu_i({\mathscr{D}}_0 ) - \nu_i({\mathscr{D}} )} \\
&\times \left (\frac{\frac{1}{(\Sigma_0)_{i|pa_i({\mathscr{D}} )}}+K_1\frac d {\epsilon_{0,n}^3} \sqrt{\frac{\log p }{n}}}
{\frac{1}{(\Sigma_0)_{i|pa_i({\mathscr{D}}_0 )}}-K_1\frac d {\epsilon_{0,n}^3}\sqrt{\frac{\log p }{n}}} \right)^{\frac{n
		+2-3}{2}} \\
\le& \left(\exp\left\{\frac{2d \log\left(\frac{\delta_2}{\delta_1}\right) + 6c\log n + (c+d)\log\left(\frac 1 {\epsilon_{0,n}}\right)}{n-1} + 
8\eta_n \left(\nu_i({\mathscr{D}}_0 )-\nu_i({\mathscr{D}} )\right) \right\}\right)^{\frac{n-1}
	{2}} \\
&\times \left(1+ \frac{(\frac{1}{(\Sigma_0)_{i|pa_i({\mathscr{D}}_0 )}} - \frac{1}{(\Sigma_0)_{i|
			pa_i({\mathscr{D}} )}}) - 2K_1 \frac d {\epsilon_{0,n}^3} \sqrt{\frac{\log p }{n}} }{\frac{1}{(\Sigma_0)_{i|
			pa_i({\mathscr{D}} )}} + K_1\frac d {\epsilon_{0,n}^3} \sqrt{\frac{\log p }{n}}} \right)^{- \frac{n-1}{2}} \\
\le& \left(\exp\left\{\frac{2d \log\left(\frac{\delta_2}{\delta_1}\right) + 6c\log n + (c+d)\log\left(\frac 1 {\epsilon_{0,n}}\right)}{n-1} + 
8\eta_n \left(\nu_i({\mathscr{D}}_0 )-\nu_i({\mathscr{D}} )\right) \right\}\right)^{\frac{n-1}
	{2}} \\
&\times \left(1+ \frac{\epsilon_{0,n} s_n^2 (\nu_i({\mathscr{D}} _0) - \nu_i({\mathscr{D}} )) - 
	2K_1\frac d {\epsilon_{0,n}^3} \sqrt{\frac{\log p }{n}}}{2/\epsilon_{0,n}} \right)^{- \frac{n-1}{2}}. 
\end{split}
\end{align}

\noindent
Note that, by Assumptions 2, 3 and 4, $\frac{d \eta_n}{\epsilon_{0,n}^2 s_n^2} \rightarrow 0$, 
${\frac{2d \log\left(\frac{\delta_2}{\delta_1}\right) + 6c\log n + (c+d)\log\left(\frac 1 {\epsilon_{0,n}}\right)}{(n-1)\epsilon_{0,n}^2 s_n^2}} \rightarrow 0$, and 
$\frac{\eta_n}{s_n^2} \rightarrow 0$ as $n \rightarrow \infty$. Since $e^x \leq 1 + 2x$ for $x < \frac{1}{2}$, by Assumptions 1 and 4, there exists $N_4''''$ 
such that for $n \geq N_4''''$, 
$$
2K_1\frac d {\epsilon_{0,n}^3} \sqrt{\frac{\log p }{n}} \le \epsilon_{0,n}\eta \le \frac{\epsilon_{0,n} s_n^2}{2}, 
$$

\noindent
and 
\begin{align*}
&\exp\left\{\frac{2d \log\left(\frac{\delta_2}{\delta_1}\right) + 6c\log n + (c+d)\log\left(\frac 1 {\epsilon_{0,n}}\right)}{n-1} + 4\eta 
\left(\nu_i({\mathscr{D}}_0 )-\nu_i({\mathscr{D}} )\right) \right\} \\
\le& 1 + 8\eta (\nu_i({\mathscr{D}}_0 )-\nu_i({\mathscr{D}} )) + \frac{4d \log
\left(\frac{\delta_2}{\delta_1}\right) + 12c\log n + 2(c+d)\log\left(\frac 1 {\epsilon_{0,n}}\right)}{n-1}\\
\leq& 1 + \frac{\epsilon_{0,n}^2 s_n^2}{8}. 
\end{align*}

\noindent
It follows by (\ref{pp6}) and the above observations that 
\begin{align*}
B_i({\mathscr{D}} ,{\mathscr{D}}_0 ) \le \left (\frac{1+\frac{\epsilon_{0,n}^2}{8}s_n^2}{1+\frac{\epsilon_{0,n}^2}{4}s_n^2} 
\right)^{\frac{n-1}{2}} 
\end{align*}

\noindent
for $n \geq \max(N_4',N_4'', N_4''',N_4'''')$. The last step follows by noting that $\nu_i 
(\mathscr{D}_0) - \nu_i (\mathscr{D}) \geq 1$. Since $pa_i (\mathscr{D}_0)$ is 
non-empty by hypothesis, $\exists j \in \nu_i({\mathscr{D}}_0 )$ such that $|(L_0)_{ji}| 
\ge s_n$, which implies that $s_n^2 \le (L_0) _{ji}^2 \le \frac{1}
{\epsilon_{0,n}}\left(\frac{[(L_0)_{ji}]^2}{(D_0)_{ii}}\right) \le \frac{(\Omega_0)_{jj}}
{\epsilon_{0,n}} \le \frac{1}{\epsilon_{0,n}^2}$ and $\epsilon_{0,n}^2 s_n^2 \le 1$. 
Hence, following from Assumption 4, there exists $N_4$ such that for $n \ge N_4$, 
we get 
\begin{align*}
B_i({\mathscr{D}} , {\mathscr{D}}_0 ) \le \left (1 - \frac{\frac{\epsilon_{0,n}^2}{8}s_n^2}{1 
+ \frac{\epsilon_{0,n}^2}{4}s_n^2} \right)^{\frac{n-1}{2}} &\le 
\exp\left\{-\left(\frac{\frac{\epsilon_{0,n}^2}{8}s_n^2}{1 + \frac{\epsilon_{0,n}^2}{4}
s_n^2}\right)\left(\frac{n-1}{2} \right)\right\} \\
&\le e^{-\frac{1}{10}\epsilon_{0,n}^2 s_n^2(\frac{n-1}{2})} \le e^{-d\eta_n n}. 
\end{align*}
\end{proof}

\begin{lemma} \label{lm5}
Suppose ${\mathscr{D}}$ is such that $pa_i({\mathscr{D}}_0) \neq pa_i({\mathscr{D}})$, $pa_i({\mathscr{D}}_0) \nsubseteq 
pa_i({\mathscr{D}})$, and $pa_i({\mathscr{D}}_0) \nsupseteq pa_i({\mathscr{D}})$, then, for $n \ge N_5$(not depending 
on $i$ or $\mathscr{D}$), $B_i({\mathscr{D}} , {\mathscr{D}} _0) \le \epsilon_{5,n}$, where $\epsilon_{5,n} = 
\max(\epsilon_{1,n},\epsilon_{2,n},\epsilon_{3,n})\epsilon_{4,n}.$
\end{lemma}

\noindent
The proof of this lemma is provided in the Supplemental Document. With these lemmas in hand, 
Theorem \ref{thm1} can be proved as follows. By Lemmas \ref{lm1} - \ref{lm5}, if we restrict to 
the event $C_n^c$, and $pa_i (\mathscr{D}) \neq pa_i (\mathscr{D}_0)$, then 
$B_i({\mathscr{D}} , {\mathscr{D}} _0) \le \epsilon^*_n$ for every $n \geq 
\max (N_1, N_2, N_3, N_4)$, where $\epsilon^*_n \stackrel{\Delta}
{=} \max \{\epsilon_{1,n}, \epsilon_{2,n},\epsilon_{3,n}, \epsilon_{4,n}, \epsilon_{5,n}\}$ 
converges to $0$ as $n \rightarrow \infty$ (by Assumption 3). Note that if $\mathscr{D} 
\neq \mathscr{D}_0$, then there exists at least one i, such that $pa_i (\mathscr{D}) 
\neq pa_i (\mathscr{D}_0)$. It follows by Lemma \ref{newlemma1}, that if we restrict to 
$C_n^c$, then 
\begin{equation} \label{pp6.1}
\max_{\mathscr{D} \neq \mathscr{D}_0} \frac{\pi({\mathscr{D}}|\bm{Y})}
{\pi({\mathscr{D}}_0 |\bm{Y})} \leq \max_{\mathscr{D} \neq \mathscr{D}_0} 
\prod_{i=1}^{p}B_i(\mathscr{D} ,\mathscr{D}_0) \leq \epsilon^*_n 
\end{equation}

\noindent
for every $n \geq \max (N_1, N_2, N_3, N_4)$. By (\ref{smplbound1}), $P(C_n^c) 
\rightarrow 1$ as $n \rightarrow \infty$. Theorem \ref{thm1} follows immediately. 

To prove Theorem \ref{thm2}, note that if $p/n^{\widetilde{k}} \rightarrow \infty$, then one can 
choose $c'$ in (\ref{smplbound1}) such that $m_2(c')^2/4 = 2 + 2/\widetilde{k}$. It follows that 
$P(C_n) \leq m_1/n^2$ for large enough $n$. The result follows by (\ref{pp6.1}) and the 
Borel-Cantelli lemma. 

We now move on to the proof of Theorem \ref{thm3}, and only consider DAGs with 
number of edges at most $h = \frac{1}{8} d \left( \frac{n}{\log p}
\right)^{\frac{1+k}{2+k}}$. By Lemmas \ref{lm1} - \ref{lm5}, it follows that if we 
restrict to $C_n^c$, then  
\begin{eqnarray} \label{pfm2}
\frac{1 - \pi(\mathscr{D}_0 | \bm{Y})}{\pi(\mathscr{D}_0 | \bm{Y})} 
&=& \sum_{\mathscr D \neq \mathscr D_0, \mathscr{D} \mbox{ has atmost } h 
\mbox{ edges}} \frac{\pi(\mathscr{D} | \bm{Y})}{\pi(\mathscr{D}_0 | \bm{Y})} \nonumber\\
&\le& \sum_{i=0}^{h}\binom{\binom{p}{2}}{i} \max_{{\mathscr{D}} \ne {\mathscr{D}}_0} 
\frac{\pi(\mathscr{D}|\bm{Y})}{\pi({\mathscr{D}}_0|\bm{Y})} \nonumber\\
&\le& p^{3h} e^{-\frac{\eta_n n}{2}} = e^{3h \log p - \frac{\eta_n n}{2}} = 
e^{-\frac{1}{8} d n^{\frac{1+k}{2+k}} (\log p)^{\frac{1}{2+k}}} 
\end{eqnarray}

\noindent
for $n \geq \max (N_1, N_2, N_3, N_4)$. Theorem \ref{thm3} follows immediately.

\section{Results for non-local priors} \label{sec:nonlocalpriors}

\noindent
{ In \cite{ACR:2013}, the authors present an alternative to the Wishart-based Bayesian framework for Gaussian DAG 
models by using non-local priors. Non-local priors were first introduced in \cite{John:Rossell:nonlocaltesting:2010} as densities 
that are identically zero whenever a model parameter is equal to its null value in the context of hypothesis testing (compared to 
local priors, which still preserve positive values at null parameter values). Non-local priors tend to discard spurious covariates 
faster as the sample size $n$ grows, while preserving exponential learning rates to detect non-zero coefficients as indicated in 
\cite{John:Rossell:nonlocaltesting:2010}. These priors were further extended to Bayesian model selection problems in 
\cite{Johnson:Rossell:2012} by imposing non-local prior densities on a vector of regression coefficients. The non-local 
prior based approach for Gaussian DAG models proposed in \cite{ACR:2013}, adapted to our notation and framework, can be
described as follows: 
\begin{eqnarray} \label{nl1}
\begin{split}
&\bm Y | \left((D,L), \mathscr D\right) \sim N_p \left( \bm 0, (LD^{-1}L^T)^{-1}\right),
\\
& \pi_{NL} ((D,L) | \mathscr D) \propto \prod_{j=1}^p \frac{1}{D_{jj}} \left( \prod_{i \in pa_j (\mathscr D)} L_{ij}^{2r} \right), 
\\
&\pi_{NL} (\mathscr D) = \prod_{i=1}^{p-1} q^{\nu_i(\mathscr{D})} (1-q)^{p-i-
\nu_i(\mathscr{D})},
\end{split}
\end{eqnarray}

\noindent
where $r$ is a fixed positive integer. Note that the prior on $(D,L)$ is an improper objective prior. If $\max_{1 \leq j \leq p} 
pa_j (\mathscr D) > n$, then this objective prior leads to an improper posterior for $(D,L)$ as well. In such cases, the authors 
in \cite{ACR:2013} propose using fractional Bayes factors. However, for the purposes of proving strong selection consistency, 
similar to Theorem \ref{thm3}, we will restrict the prior on the space of DAGs to graphs whose total number of edges is 
appropriately bounded (leaving out unrealistically large models, in the terminology of \cite{Narisetty:He:2014}). This will ensure 
that the posterior impropriety issue never arises. 

The next result establishes strong selection consistency for the objective non-local prior based approach of \cite{ACR:2013}. 
The proof is provided in the Supplementary Document. 
\begin{theorem}[(Strong selection consistency for non-local priors] \label{thm4}
Consider the non-local prior based model described in (\ref{nl1}). Under Assumptions 1-4, if we restrict the prior to DAG's 
with total number of edges at most $ d \left( \frac{n}{\log p} \right)^{\frac{1}{2(2 + k)}}$, the following holds: 
$$
\pi_{NL} (\mathscr D_0 \mid \bm{Y}) \stackrel{\bar{P}}{\rightarrow} 1, 
$$

\noindent
as $n \rightarrow \infty$. 
\end{theorem}

\noindent
Note that the only difference between the assumptions need for Theorem \ref{thm3} (for DAG-Wshart priors) and 
Theorem \ref{thm4} (non-local priors) is that the in Theorem \ref{thm4} we restrict to DAG's with number of edges at most 
$d \left( \frac{n}{\log p} \right)^{\frac{1}{2(2 + k)}}$ (as opposed to 
$\frac{1}{8} d \left( \frac{n}{\log p} \right)^{\frac{1+k}{2 + k}}$ for Theorem \ref{thm3}). All the remaining assumptions 
(Assumptions 1-4) are identical. Assumption 5 relates to the hyperparameters of the DAG-Wishart distribution, and hence is 
not relevant for the non-local prior setting.}

\section{Discussion: Comparison of penalized likelihood and Bayesian approaches} 
\label{sec:penalized}

\noindent
As mentioned in the introduction, several penalized likelihood approaches for sparse 
estimation in Gaussian DAG models have been proposed in the literature. The objective 
of this section is to compare and contrast these methods with the Bayesian approach of 
Ben-David et al. \cite{BLMR:2016} considered in this paper, and discuss advantages and 
disadvantages. 

For this discussion, we will focus on the approaches in \cite{HLPL:2006, Shojaie:Michailidis:2010, 
KORR:2017}, because these do not put any restrictions on the resulting sparsity pattern and 
focus on DAG models with ordering similar to the work of 
Ben-David et al. \cite{BLMR:2016}. For several applications in genetics, finance, and 
climate sciences, a location or time based ordering of variables is naturally available. For temporal data, 
a natural ordering of variables is provided by the time at which they are observed. In genetic datasets, 
the variables can be genes or SNPs located contiguously on a chromosome, and their spatial location 
provides a natural ordering.  More examples can be found in 
\cite{HLPL:2006, Shojaie:Michailidis:2010, Yu:Bien:2016, KORR:2017}. 

The more complex case of DAG models where a domain-specific ordering of the vertices is not known 
has also been studied in the literature, see 
\cite{Rutimann:Buhlmann:2009, vandeGeer:Buhlmann:2013, AAZ:2015} and the references therein. 
In \cite{Rutimann:Buhlmann:2009}, the authors first recover the underlying conditional independence 
relationships by estimating the equivalence class of DAGs (CPDAG class). Then, a DAG is chosen 
from this class, and then a covariance matrix obeying the conditional independence relationships in 
this DAG is estimated. In \cite{AAZ:2015}, the authors simultaneously estimate the DAG and the 
covariance matrix using a penalized regression approach. 

\subsection{Brief description of penalized likelihood methods for DAG models with given ordering} \label{sec:brief:description}

\noindent
All of the penalized likelihood methods in \cite{HLPL:2006, Shojaie:Michailidis:2010, KORR:2017} 
consider the decomposition $\Omega = L^T D^{-1} L$ where $L$ is a lower triangular matrix with 
ones on the diagonal, and $D$ is a diagonal matrix. Then, they minimize an objective function 
comprised of the log Gaussian likelihood and appropriate $\ell_1$ penalty which induces sparsity 
in the resulting estimator of the Cholesky factor $L$ (corresponding to a DAG with inverse parent 
ordering). Huang et al. \cite{HLPL:2006} obtain a sparse estimate of $L$ by minimizing the 
objective function 
\begin{equation} \label{eq:penalized1}
Q_{Chol} (L,D) = tr \left( L^T D^{-1} L S \right) + \log |D| + \lambda \sum_{1 \leq i < j \leq p} 
|L_{ij}|. 
\end{equation}

\noindent
with respect to $L$ and $D$, where $S$ is the sample covariance matrix. The objective function 
$Q_{Chol} (L, D)$ is not jointly convex in its arguments. Furthermore, as demonstrated in 
\cite{KORR:2017}, this approach can lead to singular estimates of the covariance matrix when 
$n < p$. Shojaie and Michailidis  \cite{Shojaie:Michailidis:2010} obtain a sparse estimate of $L$ 
by minimizing the convex objective function 
\begin{equation} \label{eq:penalized2}
Q_{Chol} (L, I_p) = tr \left( L^T L S \right) + \lambda \sum_{1 \leq i < j \leq p} |L_{ij}|, 
\end{equation}

\noindent
with respect to $L$, where $I_p$ denotes the identity matrix of order $p$ (an adaptive lasso 
version of the above objective function is also considered in \cite{Shojaie:Michailidis:2010}). Note 
that the objective function in (\ref{eq:penalized2}) is a special case of the objective function in 
(\ref{eq:penalized1}) with $D = I_p$. The approach in \cite{Shojaie:Michailidis:2010} provides a 
sparse estimate of $L$ and hence is useful for sparsity/DAG selection, but cannot be used for 
estimation of $\Omega$ as it does not provide an estimate of the conditional variances 
$\{D_{ii}\}_{i=1}^p$. In \cite{KORR:2017}, the authors reparametrize in terms of $T = 
D^{-\frac{1}{2}} L$, and obtain their estimate by minimizing the jointly convex objective function 
\begin{equation} \label{eq:penalized3}
Q_{CSCS} (T) = tr \left( T^T T S \right) - 2 \log |T| + \lambda \sum_{1 \leq i < j \leq p} |T_{ij}|, 
\end{equation}

\noindent
with respect to $T$. This approach preserves the desirable properties of the methods in 
\cite{HLPL:2006, Shojaie:Michailidis:2010} while overcoming their shortcomings. See 
\cite{KORR:2017} for more details. 

\subsection{Comparison: Graph search complexity and accuracy} \label{sec:graph:complexity:accuracy}

\noindent
For all the penalized likelihood methods, the user-specified penalty parameter $\lambda$ 
controls the level of sparsity of the resulting estimator. Varying $\lambda$ provides a range of 
possible DAG models to choose from. This set of graphs, often obtained over a grid of $\lambda$ 
values, is often refered to as the solution path for the particular penalized likelihood method. The 
choice of $\lambda$ is typically made by assigning a `score' to each DAG on the solution path 
using the Bayesian Information Criterion (BIC) or cross-validation, and choosing the DAG with 
the highest score. Another approach is to choose a specific value of $\lambda$ based on 
asymptotic normal approximation (see \cite{Shojaie:Michailidis:2010}). For the Bayesian 
approach, the posterior probabilities naturally assign a `score' for {\bf all the $2^{p \choose 2}$ 
DAGs}, not just the graphs on the solution path produced by the penalized likelihood methods. 
Of course, the entire space of DAGs is prohibitively large to search in high-dimensional settings. 
To address this, Ben-David et al. \cite{BLMR:2016} develop a computationally feasible approach 
which searches around the graphs on the penalized likelihood solution path by adding or 
removing edges, and demonstrate that significant improvement in accuracy can be obtained by 
searching beyond the penalized likelihood solution paths using posterior probabilities. Hence, 
this Bayesian procedure maintains the advantage of being able to do a principled broader search 
(for improved accuracy) in a computationally feasible way. One can extend this procedure by 
also searching on and around the solution paths of other methods, such as the CSCS method in 
\cite{KORR:2017}, and choose the graph with the maximum posterior probability. We implement 
such a Bayesian procedure in Section \ref{sec:illustration:graph:selection} and demonstrate the improvement 
in graph selection performance that can be obtained as compared to penalized likelihood approaches. 

\subsection{Comparison: Uncertainty quantification and prior information}

\noindent
While Bayesian methods naturally provide uncertainty quantification through the posterior distribution, it 
is crucial to establish the accuracy of this uncertainty quantification especially in modern 
high-dimensional settings. The high-dimensional asymptotic results in this paper provide justification for such uncertainty 
quantification using the Bayesian approach of \cite{BLMR:2016}. Uncertainty quantification 
for estimates produced by the penalized likelihood methods can be achieved through 
a CLT or through resampling methods such as bootstrap. To the best of our knowledge, 
a high-dimensional CLT, or results establishing high-dimensional accuracy of the bootstrap in this 
context are not available for the penalized likelihood based estimators in 
\cite{HLPL:2006, Shojaie:Michailidis:2010, KORR:2017}. 

A natural benefit of Bayesian approaches is the ability to incorporate prior knowledge. 
However, this can be done in a principled way only when the hyperparameters are interpretable, 
and the class of priors is flexible. The distributional and moment results 
in \cite{BLMR:2016} provide a natural interpretability for the hyperparameters $U$ and 
${\boldsymbol \alpha}$. Also, as mentioned in \cite{BLMR:2016}, a separate shape parameter 
$\alpha_i$ for each variable allows for differential shrinkage. 

\subsection{Comparison: Convergence rates for estimation of $\Omega$}

\noindent
In \cite{KORR:2017}, the authors provide convergence rates of the estimate of the precision 
matrix $\Omega$ resulting from CSCS, their penalized likelihood estimation procedure. The 
framework in \cite{Shojaie:Michailidis:2010} is restrictive for estimation of $\Omega$, as it 
assumes that $D_{ii} = 1$ for every $1 \leq i \leq p$. Also, to the best of our knowledge, 
high-dimensional asymptotic convergence rates are not available for the estimates obtained 
from the procedure in \cite{HLPL:2006}. Hence, in this section, we will undertake a comparison 
of the assumptions and convergence rates between the $\Omega$-estimate using the 
CSCS procedure in \cite{KORR:2017} and the posterior distribution convergence rate for 
$\Omega$ in Theorem E.1 in the Supplemental document. 

We start with a point-by-point comparison of the parallel/related assumptions used for 
these high-dimensional asymptotic results. 
\begin{enumerate}
\item For CSCS $p = p_n$ is assumed to be bounded above by a polynomial in $n$, whereas 
in this paper $p_n$ can grow much faster than a polynomial in $n$ (at an appropriate 
sub-exponential rate, see Assumption $2$). 
\item For CSCS the eigenvalues of the $\Omega_0$ are assumed to be uniformly 
bounded in $n$, whereas in this paper we allow the eigenvalues of $\Omega$ to grow with 
$n$ (see Assumption $1$). 
\item As with any $\ell_1$-penalized method, \cite{KORR:2017} use an incoherence 
condition for their asymptotic results. This condition is algebraically complex and 
hard to interpret. We do not need any such assumption for our asymptotic results. 
\item For CSCS mild assumptions are specified regarding the rate at which the penalty 
parameter $\lambda_n$ goes to zero, and relates to the total number of non-zero 
off-diagonal entries in the Cholesky factor of the true concentration matrix (we will denote 
this quantity by $m_n$). In this paper, we need to make analogous mild assumptions on the prior 
parameters $q_n$, $U_n$ and ${\boldsymbol \alpha}(\mathscr{D}_n)$ (see Assumptions $3$ and 
$5$). 
\item Recall that $s_n$ is the smallest (in absolute value) non-zero off-diagonal entry of $L_0$. 
For CSCS, it is assumed that $\frac{s_n}{\sqrt{d_n} \lambda_n} \rightarrow \infty$ as 
$n \rightarrow \infty$, where $\lambda_n$ is the penalty parameter, whereas we assume that
$\frac{s_n}{\sqrt{\eta_n d_n}} \rightarrow \infty$ (see Assumption 4 with $\epsilon_{0,n}$ as a 
constant for a fair comparison with CSCS).  There are other assumptions in \cite{KORR:2017} 
regarding the rate at which $\lambda_n$ goes to $0$, but these do not enable a 
direct comparison of the two rates for $s_n$. One can construct situations where the 
CSCS assumption on $s_n$ is weaker than Assumption 4, and vice-versa. 
\end{enumerate}

\noindent
As far as the convergence rate for the estimate/posterior of $\Omega$ is concerned, the 
convergence rate of the CSCS estimate is $m_n \lambda_n$, whereas the posterior 
convergence rate for $\Omega$ in Theorem E. 1 in the Supplemental document is 
$d_n^2 \sqrt{\frac{\log p_n}{n}}$ (treating $\epsilon_{0,n}$ as a constant for a fair 
comparison with CSCS). Using other assumptions regarding $\lambda_n$ in 
\cite{KORR:2017} it can be shown that $m_n^{3/2} \sqrt{\frac{\log p_n}{n}} = 
o(m_n \lambda_n)$. Hence, if $m_n^{3/2} > d_n^2$, in which case the 
Bayesian approach leads to a faster convergence rate. Of course, one can construct 
situations where $m_n^{3/2} < d_n^2$ and choose $\lambda_n$ such that CSCS would 
lead to a faster convergence rate than the Bayesian approach. Since $m_n$ is the 
{\it total} number of non-zeros in the true Cholesky factor one would expect that for a 
large majority of graphs, Theorem E.1 would lead to a faster convergence rate than 
CSCS.

\section{Experiments} \label{sec:experiments}

\subsection{Simulation I: Illustration of posterior ratio consistency} \label{sec:illustration:ratio:consistency}

\noindent
In this section, we illustrate the DAG selection consistency result in Theorems \ref{thm1} and \ref{thm2} 
using a simulation experiment. We consider $10$ different values of $p$ ranging from 
$250$ to $2500$, and choose $n = p/5$. Then, for each fixed $p$, we construct a $p \times p$ 
lower triangular matrix with diagonal entries $1$ and off-diagonal entries $0.5$. Then, each lower 
triangular entry is independently set to zero with a certain probability such that the expected 
value of non-zero entries for each column does not exceed $3$. We refer to this matrix as 
$L_0$. The matrix $L_0$ also gives us the true DAG $\mathscr D_0$. Next, we generate 
$n$ i.i.d. observations from the $N(\bm 0_p, (L_0^{-1})^T L_0^{-1})$ distribution, and set the 
hyperparameters as $U = I_p$ and $\alpha_i(\mathscr D) = \nu_i(\mathscr D) + 10$ for 
$i = 1,2,\ldots,p$. The above process ensures Assumptions 1-5 are satisfied. We then examine 
posterior ratio consistency under four different cases by computing the log posterior ratio of 
a ``non-true" graph $\mathscr D$ and $\mathscr D_0$ as follows. 
\begin{enumerate}
\item Case $1$: $\mathscr D$ is a subgraph of $\mathscr D_0$ and the number of total edges of 
$\mathscr D$ is exactly half of $\mathscr D_0$, i.e. $|E(\mathscr D)| = \frac 1 2 
|E(\mathscr D_0)|$. 
\item Case $2$: $\mathscr D$ is a supergraph of $\mathscr D_0$ and the number of total edges of 
$\mathscr D$ is exactly twice of $\mathscr D_0$, i.e. $|E(\mathscr D)| = 2 |E(\mathscr D_0)|$. 
\item Case $3$: $\mathscr D$ is not necessarily a subgraph of $\mathscr D_0$, but the number 
of total edges in $\mathscr D$ is half the number of total edges in $\mathscr D_0$. 
\item Case $4$: $\mathscr D$ is not necessarily a supergraph of $\mathscr D_0$, but the number 
of total edges in $\mathscr D$ is twice the number of total edges in $\mathscr D_0$. 
\end{enumerate}

\noindent
The log of the posterior probability ratio for various cases is provided in Table \ref{logratio}. 
As expected the log of the posterior probability ratio eventually decreases as $n$ 
becomes large in all four cases, thereby providing a numerical illustration of 
Theorems \ref{thm1} and \ref{thm2}. 
\begin{table}
	\centering
	\begin{tabular}{|c|c|c|c|c|c|}
	\hline
			 \mbox{ } & \mbox{ } & $\mathscr D \subset \mathscr D_0 $ & $\mathscr D \supset \mathscr D_0$  & \mbox{ } & \mbox{ } \\
		$p$ & $n$ &	$|E(\mathscr D)| = \frac 1 2 |E(\mathscr D_0)|$ & $|E(\mathscr D)| = 2|E(\mathscr D_0)|$ & $|E(\mathscr D)| = \frac 1 2 |E(\mathscr D_0)|$   & $|E(\mathscr D)| = 2|E(\mathscr D_0)|$ \\
	\hline
			250  & 50  & 38553  & -133007   & 24723  & -139677 \\
			500  & 100 & 93634  & -458799 & 51553  & -438377 \\
			750  & 150 & 41935  & -784866 & 60731  & -1042449  \\
			1000 & 200 & 249342  & -1118384  & -28657 & -1791276  \\
			1250 & 250 & 18847  & -1787260  & -245769 & -2633731  \\
			1500 & 300 & -79566 & -2603779  & -452125 & -3873151  \\
			1750 & 350 & -512894 & -2971286  & -455941 & -5808992  \\
			2000 & 400 & -443457 & -4082005  & -1388037  & -7139952  \\
			2250 & 450 & -558718 & -4533967  & -1883472  & -8744044  \\
			2500 & 500 & -571653 & -4708833  & -2644104  & -9910277  \\
	\hline
	\end{tabular}
	\caption{Log of posterior probability ratio for $\mathscr D$ and $\mathscr D_0$ for various 
	choices of the ``non-true" DAG $\mathscr D$. Here $\mathscr D_0$ denotes the true 
	underlying DAG.}
	\label{logratio}
\end{table}

\subsection{Simulation II: Illustration of graph selection} \label{sec:illustration:graph:selection}

\noindent
In this section, we perform a simulation experiment to illustrate the potential advantages of using the hybrid Bayesian graph 
selection approach outlined in Section \ref{sec:graph:complexity:accuracy}. We consider $7$ values of $p$ ranging from 
$2500$ to $4000$, with $n = p/5$. For each fixed $p$, the Cholesky factor $L_0$ of the true concentration matrix, 
and the subsequent dataset, is generated by the same mechanism as in Section \ref{sec:illustration:ratio:consistency}. Then, we 
perform graph selection using the four procedures outlined below. 
\begin{enumerate}
\item {\it Lasso-DAG BIC path search}: We implement the Lasso-DAG approach in \cite{Shojaie:Michailidis:2010} discussed in 
Section \ref{sec:brief:description}. The penalty parameter $\lambda$ is varied on a grid so that the resulting graphs range from 
approximately three times the edges compared to the true graph with approximately one third edges compared 
to the true graph. We then select the best graph according to the ``BIC"-like measure defined as 
\begin{equation} \label{BIC}
BIC(\lambda) = ntr(S\hat{\Omega}) - n\log|\hat{\Omega}| + \log n * E,
\end{equation}

\noindent
where $\widehat{L}$ is the resulting estimator from Lasso-DAG, $E$ denotes the total numbers of non-zero entries in $\hat{L}$ 
and $\hat{\Omega} = \hat{L}^T\hat{L}.$ 
\item {\it Lasso-DAG with quantile based tuning}: We again implement the Lasso-DAG approach in \cite{Shojaie:Michailidis:2010}, 
but choose penalty parameters (separate for each variable $i$) given by $\lambda_i(\alpha) = 
2 n^{-\frac 1 2}Z^*_{\frac{0.1}{2p(i-1)}}$, where $Z_q^*$ denotes the $(1-q)^{th}$ quantile of the standard normal distribution. 
This choice is justified in \cite{Shojaie:Michailidis:2010} based on asymptotic considerations.
\item {\it CSCS BIC path search}: We implement the CSCS approach \cite{KORR:2017} discussed in 
Section \ref{sec:brief:description}. The penalty parameter $\lambda$ is varied on a grid so that the 
resulting graphs range from three times the edges compared to the true graph with one third edges 
compared to the true graph. The best graph us selected using the ``BIC"-like measure in (\ref{BIC}). 
\item {\it Bayesian approach}: We construct two sets of candidate graphs as follows. 
\begin{enumerate}
\item All the graphs on the solution paths for Lasso-DAG and CSCS are included in the candidate set. To increase the search 
range, we generate additional graphs by thresholding the modified Cholesky factor of $(S + 0.5 I)^{-1}$ ($S$ is the sample 
covariance matrix) to get a sequence of $300$ additional graphs, and include them in the candidate set. We then search around 
all the above graphs using Shotgun Stochastic Search to generate even more candidate graphs. Then we implement Algorithm 
A.$8$ in (\cite{Koller:Friedman:09}), the Greedy Hill-climbing algorithm, to our candidate graphs. For each graph, this particular 
search procedure first generates a new DAG by adding one random edge and only chooses it if the new DAG has a higher 
posterior score. Then, we generate another graph by deleting one random edge from the chosen DAG and select the one with 
higher score. We repeat the whole process $20$ times for every graph in the previous candidate set and all the chosen DAGs 
are included in the candidate set. 
\item We combine Algorithm $18.1$ in (\cite{Koller:Friedman:09}) and the idea of cross-validation to form our second set of 
candidate graphs. The original data set of $n$ observations is randomly partitioned into $10$ equal sized subsets. Of the $10$ 
subsets, a single subset is excluded, and the remaining $9$ subsets are used as our new sample. The same thresholding 
procedure to generate $300$ graphs is performed for the new sample covariance matrix. The process is then repeated $10$ 
times, with each of the $10$ subsamples removed exactly once. We then have a total of $3000$ graphs as the second candidate 
set.
\end{enumerate}

\noindent
The log posterior probabilities are computed for all graphs in the two candidate sets, and the graph with the highest probability is 
chosen. 
\end{enumerate}

\noindent
The model selection performance of these four methods is then compared using several different measures of structure such as 
positive predictive value, true positive rate and false positive rate (average over $20$ independent repetitions). 
Positive Predictive Value (PPV) represents the proportion of true edges among all the edges detected by the given procedure, 
True Positive Rate (TPR) measures the proportion of true edges detected by the given procedure among all the edges 
from the true graph, and False Positive Rate (FPR) represents the proportion of false edges detected by the given procedure 
among all the non-edges in the true graph. One would like the PPV and TPR values to be as close to $1$ as possible, and the 
FPR value to be as close to $0$ as possible. The results are provided in Table \ref{my-label}. It is clear that the Bayesian 
approach outperforms the penalized likelihood approaches based on all measures. The PPV values for the Bayesian approach 
are all above $0.95$, while the ones for the penalized likelihood approaches are around $0.1$. The TPR values for the Bayesian 
approach are all above $0.39$, while the ones for the penalized likelihood approaches are all below $0.21$. The FPR values for 
the Bayesian approach are all significantly smaller than the penalized approaches. Overall, this experiment 
illustrates that the Bayesian approach can be used for a broader yet computationally feasible graph search, and can lead 
to a significant improvement in graph selection performance. 
\begin{table}[H]
	\small
	\centering
	\caption{Model selection performance table}
	\label{my-label}
	\scalebox{0.75}{
		\begin{tabular}{|c|c|c|c|c|c|c|c|c|c|c|c|c|c|}
			\hline
			\multicolumn{2}{|c|}{\multirow{2}{*}{}} & \multicolumn{3}{c|}{Lasso-DAG}       & \multicolumn{3}{c|}{Lasso-DAG}              & \multicolumn{3}{c|}{CSCS}            & \multicolumn{3}{c|}{Bayesian}                               \\ \cline{3-14} 
			\multicolumn{2}{|c|}{}                  & \multicolumn{3}{c|}{BIC path search} & \multicolumn{3}{c|}{Quantile-based lambdas} & \multicolumn{3}{c|}{BIC path search} & \multicolumn{3}{c|}{Log-score path search}                  \\ \hline
			p                   & n                 & PPV        & TPR        & FPR        & PPV           & TPR          & FPR          & PPV        & TPR        & FPR        & PPV       & TPR       & FPR                                 \\ \hline
			2500                & 500               & 0.0822  & 0.2007  & 0.0027  & 0.1031     & 0.1820    & 0.0019    & 0.0864  & 0.2046  & 0.0025  & 0.9933 & 0.3956 & $3.76 \times 10^{-6}$    \\ 
			2750                & 550               & 0.0530  & 0.1730  & 0.0034  & 0.0649    & 0.1626    & 0.0026    & 0.0816  & 0.1849  & 0.0023  & 0.9878 & 0.4578 & $6.75 \times 10^{-6}$  \\ 
			3000                & 600               & 0.0604  & 0.1868  & 0.0029  & 0.0713     & 0.1706    & 0.0022    & 0.0780  & 0.1828  & 0.0021  & 0.9927 & 0.4525 & $3.26 \times 10^{-6}$  \\ 
			3250                & 650               & 0.0619  & 0.1713  & 0.0024  & 0.0712     & 0.1562    & 0.0019     & 0.0803  & 0.1843  & 0.0019  & 0.9686 & 0.5023 & $1.73 \times 10^{-5}$ \\ 
			3500                & 700               & 0.0668  & 0.1765  & 0.0021  & 0.0782     & 0.1646   & 0.0017     & 0.0815  & 0.1884  & 0.0018   & 0.9719 & 0.5247 & $1.31  \times 10^{-5}$  \\ 
			3750                & 750               & 0.0684  & 0.1741  & 0.0019  & 0.0787     & 0.1643    & 0.0016     & 0.0756  & 0.1745  & 0.0017  & 0.9782 & 0.5242 & $9.47 \times 10^{-6}$ \\ 
			4000               & 800               & 0.0681  & 0.1865  & 0.0019  & 0.0770     & 0.1758   & 0.0016     & 0.0732  & 0.1654  & 0.0016   & 0.9570 & 0.5653 & $1.91 \times 10^{-5}$  \\ \hline
		\end{tabular}
	}
\end{table}

\section{Acknowledgment} \label{sec:acknowledgment}
The authors are grateful to anonymous referees and an Associate Editor for their encouraging and helpful comments which substantially improved the paper.

\bibliographystyle{plain}
\bibliography{references}

\begin{thebibliography}{10}

\bibitem{ACR:2013}
D.~Altamore, G.~Consonni, and L.~La~Rocca.
\newblock Objective bayesian search of gaussian directed acyclic graphical
  models for ordered variables with non-local priors.
\newblock {\em Biometrics}, 69(478-487), 2013.

\bibitem{AAZ:2015}
B.~Aragam, A.~Amini, and Q.~Zhou.
\newblock Learning directed acyclic graphs with penalized neighbourhood
  regression.
\newblock {\em https://arxiv.org/abs/1511.08963}, 2015.

\bibitem{AAZ:2016}
B.~Aragam, A.A. Amini, and Q.~Zhou.
\newblock Learning directed acyclic graphs with penalized neighbourhood
  regression.
\newblock {\em arxiv}, 2016.

\bibitem{Banerjee:Ghosal:2014}
S.~Banerjee and S.~Ghosal.
\newblock Posterior convergence rates for estimating large precision matrices
  using graphical models.
\newblock {\em Electronic Journal of Statistics}, 8:2111--2137, 2014.

\bibitem{Banerjee:Ghosal:2015}
S.~Banerjee and S.~Ghosal.
\newblock Bayesian structure learning in graphical models.
\newblock {\em Journal of Multivariate Analysis}, 136:147--162, 2015.

\bibitem{BLMR:2016}
E.~Ben-David, T.~Li, H.~Massam, and B.~Rajaratnam.
\newblock High dimensional bayesian inference for gaussian directed acyclic
  graph models.
\newblock {\em Technical Report}, http://arxiv.org/abs/1109.4371, 2016.

\bibitem{Bickel:Levina:2008}
P.~J. Bickel and E.~Levina.
\newblock Regularized estimation of large covariance matrices.
\newblock {\em Ann. Statist.}, 36:199--227, 2008.

\bibitem{CLP:2017}
G.~Consonni, L.~La~Rocca, and S.~Peluso.
\newblock Objective bayes covariate- adjusted sparse graphical model selection.
\newblock {\em Scand. J. Statist.}, 44:741--764, 2017.

\bibitem{datta1996}
Gauri~Sankar Datta and Malay Ghosh.
\newblock On the invariance of noninformative priors.
\newblock {\em Ann. Statist.}, 24(1):141--159, 02 1996.

\bibitem{ElKaroui:2008}
N.~El~Karoui.
\newblock Spectrum estimation for large dimensional covariance matrices using
  random matrix theory.
\newblock {\em Annals of Statistics}, 36:2757--2790, 2008.

\bibitem{Geiger:Heckerman:2002}
D.~Geiger and D.~Heckerman.
\newblock Parameter priors for directed acyclic graphical models and the
  characterization of several probability distributions.
\newblock {\em Ann. Statist.}, 30:1412--1440, 2002.

\bibitem{HLPL:2006}
J.~Huang, N.~Liu, M.~Pourahmadi, and L.~Liu.
\newblock Covariance selection and estimation via penalised normal likelihood.
\newblock {\em Biometrika}, 93:85--98, 2006.

\bibitem{Johnson:Rossell:2012}
V.~Johnson and D.~Rossell.
\newblock Bayesian model selection in high-dimensional settings.
\newblock {\em J. Amer. Statist. Assoc}, 107(498):649--660, 201.

\bibitem{John:Rossell:nonlocaltesting:2010}
V.~Johnson and D.~Rossell.
\newblock On the use of non-local prior densities in bayesian hvoothesis tests
  hypothesis.
\newblock {\em J. Royal Stat. Soc, Ser. B}, 72:143--170, 2010.

\bibitem{Johnson:Rossell:supplemental:2012}
V.~Johnson and D.~Rossell.
\newblock Supplementary material to ``bayesian model selection in
  high-dimensional settings".
\newblock {\em J. Amer. Statist. Assoc}, 2012.

\bibitem{KORR:2017}
K.~Khare, S.~Oh, S.~Rahman, and B.~Rajaratnam.
\newblock A convex framework for high-dimensional sparse cholesky based
  covariance estimation in gaussian dag models.
\newblock {\em Preprint, Department of Statisics, University of Florida}, 2017.

\bibitem{Koller:Friedman:09}
D.~Koller and N.~Friedman.
\newblock {\em Probabilistic Graphical Models: Principles and Techniques}.
\newblock MIT Press, 2009.

\bibitem{Letac:Massam:2007}
G.~Letac and H.~Massam.
\newblock Wishart distributions for decomposable graphs.
\newblock {\em Ann. Statist.}, 35:1278--1323, 2007.

\bibitem{Mazumder:Hastie:2012}
R.~Mazumder and T.~Hastie.
\newblock Exact covariance thresholding into connected components for
  large-scale graphical lasso.
\newblock {\em The Journal of Machine Learning Research}, 13:781--794, 2012.

\bibitem{Narisetty:He:2014}
N.~Narisetty and X.~He.
\newblock Bayesian variable selection with shrinking and diffusing priors.
\newblock {\em Ann. Statist.}, 42:789--817, 2014.

\bibitem{PPS:1989}
V.I. Paulsen, S.C. Power, and R.R. Smith.
\newblock Schur products and matrix completions.
\newblock {\em J. Funct. Anal.}, 81:151--178, 1989.

\bibitem{Pourahmadi:2007}
M.~Pourahmadi.
\newblock Cholesky decompositions and estimation of a covariance matrix:
  Orthogonality of variance--correlation parameters.
\newblock {\em Biometrika}, 94:1006--1013, 2007.

\bibitem{RLZ:2010}
A.~J. Rothman, E.~Levina, and J.~Zhu.
\newblock A new approach to cholesky-based covariance regularization in high
  dimensions.
\newblock {\em Biometrika}, 97:539--550, 2010.

\bibitem{rudelson2013}
Mark Rudelson and Roman Vershynin.
\newblock Hanson-wright inequality and sub-gaussian concentration.
\newblock {\em Electronic Communications in Probability}, 18:9 pp., 2013.

\bibitem{Rutimann:Buhlmann:2009}
P.~Rutimann and P.~Buhlmann.
\newblock High dimensional sparse covariance estimation via directed acyclic
  graphs.
\newblock {\em Electronic Journal of Statistics}, 3:1133--1160, 2009.

\bibitem{Shojaie:Michailidis:2010}
A.~Shojaie and G.~Michailidis.
\newblock Penalized likelihood methods for estimation of sparse
  high-dimensional directed acyclic graphs.
\newblock {\em Biometrika}, 97:519--538, 2010.

\bibitem{Smith:Kohn:2002}
M.~Smith and R.~Kohn.
\newblock Parsimonious covariance matrix estimation for longitudinal data.
\newblock {\em Journal of the American Statistical Association}, 97:1141--1153,
  2002.

\bibitem{vandeGeer:Buhlmann:2013}
Sara van~de Geer and Peter B{\"u}hlmann.
\newblock $\ell_{0}$-penalized maximum likelihood for sparse directed acyclic
  graphs.
\newblock {\em Ann. Statist.}, 41(2):536--567, 04 2013.

\bibitem{Watson:1959}
G.N. Watson.
\newblock A note on gamma functions.
\newblock {\em Proc. Edinburgh Math. Soc.}, 11:7--9, 1959.

\bibitem{XKG:2015}
R.~Xiang, K.~Khare, and M.~Ghosh.
\newblock High dimensional posterior convergence rates for decomposable
  graphical models.
\newblock {\em Electronic Journal of Statistics}, 9:2828--2854, 2015.

\bibitem{Yu:Bien:2016}
G.~Yu and J.~Bien.
\newblock Learning local dependence in ordered data.
\newblock {\em arXiv:1604.07451}, 2016.

\end{thebibliography}

\begin{frontmatter}
	\title{Supplemental Document for ``Posterior graph selection and estimation consistency for 
		high-dimensional Bayesian DAG models"}
	\runtitle{Posterior consistency for Bayesian DAG models}
	
	\begin{aug}
		\author{\fnms{Xuan} \snm{Cao}\ead[label=e1s]{caoxuan@ufl.edu}},
		\author{\fnms{Kshitij} \snm{Khare}\ead[label=e2s]{kdkhare@stat.ufl.edu}}
		\and
		\author{\fnms{Malay} \snm{Ghosh}\ead[label=e3s]{ghoshm@stat.ufl.edu}}
		
		\runauthor{X. Cao et al.}
		
		\affiliation{University of Florida}
		
		
		\address{Department of Statistics\\
			University of Florida\\
			102 Griffin-Floyd Hall\\
			Gainesville, FL 32611\\
			\printead{e1s}\\
			\phantom{E-mail:\ }\printead*{e2s}\\
			\phantom{E-mail:\ }\printead*{e3s}}
	\end{aug}
\end{frontmatter}

\renewcommand{\theequation}{A.\arabic{equation}}
\setcounter{equation}{0}

\renewcommand\thesection{\Alph{section}}

\section{Proof of Lemma \ref{newlemma1}}

\noindent
By (\ref{m3}), the posterior ratio is given by 
\begin{align} \label{m4}
\begin{split}
&\frac{\pi({\mathscr{D}}|\bm{Y})}{\pi({\mathscr{D}}_0|\bm{Y})}\\ =& \frac{\pi({\mathscr{D}})}{\pi({\mathscr{D}}_0)} \frac{z_{\mathscr{D}}(U+nS, n+\bm{\alpha}(\mathscr D)) / z_{ \mathscr{D}}(U, \bm{\alpha}(\mathscr D))}{z_{{\mathscr{D}}_0}(U+nS, n+\bm{\alpha}(\mathscr D)) / z_{{\mathscr{D}}_0}(U, \bm{\alpha}(\mathscr D))} \\
=& \left( \prod_{i=1}^{p -1} (\frac{q}{1-q})^{\nu_i(\mathscr{D}) - \nu_i({{\mathscr{D}}_0})} \right) \frac{z_{\mathscr{D}}(U+nS, n+\bm{\alpha}(\mathscr D)) / z_{\mathscr{D}}(U, \bm{\alpha}(\mathscr D))}{z_{\mathscr{D}_0}(U+nS, n+\bm{\alpha}(\mathscr D)) / z_{\mathscr{D}_0}(U, \bm{\alpha}(\mathscr D))}.
\end{split}
\end{align}

\noindent
Using \cite[eq. (11)]{BLMR:2016}, we get 
\begin{align*} 
\begin{split}
&\frac{z_{\mathscr{D}}(U+nS, n+\bm{\alpha}(\mathscr D))}{z_{\mathscr{D}}(U,\bm{\alpha}(\mathscr D))}\\
=&\prod_{i = 1}^p \frac{\Gamma(\frac{n + \alpha_i(\mathscr D) - \nu_i(\mathscr{D}) -2}{2}) 2^{\frac{\alpha_i(\mathscr D) + n-2}{2}} \frac{|(U+nS)_{\mathscr{D}}^{>i}|^{\frac{n+\alpha_i(\mathscr D) - \nu_i(\mathscr{D}) -3}{2}}} { |(U+nS)_{\mathscr{D}}^{\ge i}|^{\frac{n+\alpha_i(\mathscr D) - \nu_i(\mathscr{D}) - 2}{2} }}}{\Gamma( \frac{\alpha_i(\mathscr D) - \nu_i(\mathscr{D}) -2}{2}) 2^{\frac{\alpha_i(\mathscr D) - 2}{2}} \frac{|U_{\mathscr{D}}^{>i}|^{\frac{\alpha_i(\mathscr D) - \nu_i(\mathscr{D}) -3}{2}}} {|U_{\mathscr{D}}^{\ge i}|^{\frac{\alpha_i(\mathscr D) - \nu_i(\mathscr{D}) -2}{2}}}}\\
=& \prod_{i=1}^p \frac{\Gamma(\frac{n}{2} + \frac{c_i(\mathscr D)}{2} - 1)}{\Gamma(\frac{c_i(\mathscr D)}{2} - 1)} 2^{\frac n 2} \frac{n^{-\frac{n+\alpha_i(\mathscr D)-2}{2}}}{\frac{|U_{\mathscr{D}}^{>i}|^{\frac{\alpha_i(\mathscr D) - \nu_i(\mathscr{D}) -3}{2}}} {|U_{\mathscr{D}}^{\ge i}|^{\frac{\alpha_i(\mathscr D) - \nu_i(\mathscr{D}) -2}{2}}}}\frac{|\tilde{S}_{\mathscr{D}}^{>i}|^{\frac{n+c_i(\mathscr D)-3}{2}}}{|\tilde{S}_{\mathscr{D}}^{\ge i}|^{\frac{n+c_i(\mathscr D)-2}{2}}}.
\end{split}
\end{align*}
Note that $$|U_{\mathscr{D}}^{\ge i}| = |U_{\mathscr{D}}^{>i}|\left(U_{ii} - (U_{\mathscr D \cdot i}^>)^T (U_{\mathscr{D}}^{>i})^{-1} U_{\mathscr D \cdot i}^>\right) = \left|U_{\mathscr{D}}^{>i}\right|\left(\left[(U_{\mathscr{D}}^{\ge i})^{-1}\right]_{ii}\right)^{-1},$$ following from Assumption 5,
\begin{align*}
\frac{|U_{\mathscr{D}}^{>i}|^{\frac{\alpha_i(\mathscr D) - \nu_i(\mathscr{D}) -3}{2}}} {|U_{\mathscr{D}}^{\ge i}|^{\frac{\alpha_i(\mathscr D) - \nu_i(\mathscr{D}) -2}{2}}} &= \frac {\left[(U_{\mathscr{D}}^{\ge i})^{-1}\right]_{ii}^{\frac{\alpha_i(\mathscr D) - \nu_i(\mathscr{D}) -3}{2}}} {\left|U_{\mathscr{D}}^{\ge i}\right|^{\frac12}} \\
&\ge \left(\frac{1}{\delta_2}\right)^{\frac{\alpha_i(\mathscr D) - \nu_i(\mathscr{D}) -3}{2}} \times \left(\frac{1}{\delta_2}\right)^{\frac{\nu_i(\mathscr D) + 1}{2}}\\
&\ge \left(\frac 1 {\delta_2}\right)^{\frac {\alpha_i(\mathscr D) - 2} 2}.
\end{align*}
Similarly, $$\frac{|U_{\mathscr{D}}^{>i}|^{\frac{\alpha_i(\mathscr D) - \nu_i(\mathscr{D}) -3}{2}}} {|U_{\mathscr{D}}^{\ge i}|^{\frac{\alpha_i(\mathscr D) - \nu_i(\mathscr{D}) -2}{2}}} \le \left(\frac 1 {\delta_1}\right)^{\frac {\alpha_i(\mathscr D) - 2} 2}.$$

\noindent	
Hence, 
\begin{equation} \label{m4.1}
\frac{z_{\mathscr{D}}(U+nS, n+\bm{\alpha}(\mathscr D))}{z_{\mathscr{D}}(U,\bm{\alpha}(\mathscr D))} \le \prod_{i=1}^p \frac{\Gamma(\frac{n}{2} + \frac{c_i(\mathscr D)}{2} - 1)}{\Gamma(\frac{c_i(\mathscr D)}{2} - 1)} 2^{\frac n 2} \frac{n^{-\frac{n+\alpha_i(\mathscr D)-2}{2}}}{{\delta_2}^{-\frac {\alpha_i(\mathscr D) - 2} 2}}\frac{|\tilde{S}_{\mathscr{D}}^{>i}|^{\frac{n+c_i(\mathscr D)-3}{2}}}{|\tilde{S}_{\mathscr{D}}^{\ge i}|^{\frac{n+c_i(\mathscr D)-2}{2}}},
\end{equation}

\noindent
and
\begin{equation} \label{m4.2}
\frac{z_{\mathscr{D}}(U+nS, n+\bm{\alpha}(\mathscr D))}{z_{\mathscr{D}}(U,
	\bm{\alpha}(\mathscr D))} \ge \prod_{i=1}^p \frac{\Gamma(\frac{n}{2} + 
	\frac{c_i(\mathscr D)}{2} - 1)}{\Gamma(\frac{c_i(\mathscr D)}{2} - 1)} 2^{\frac n 2} 
\frac{n^{-\frac{n+\alpha_i(\mathscr D)-2}{2}}}{{\delta_1}^{-\frac {\alpha_i(\mathscr D) - 
			2} 2}}\frac{|\tilde{S}_{\mathscr{D}}^{>i}|^{\frac{n+c_i(\mathscr D)-3}{2}}}
{|\tilde{S}_{\mathscr{D}}^{\ge i}|^{\frac{n+c_i(\mathscr D)-2}{2}}}. 
\end{equation}

\noindent
Note that all the bounds above hold for an arbitrary DAG $\mathscr{D}$, and in 
particular for the true DAG $\mathscr{D}_0$. 

By Assumption 5, it follows that $2 < c_i(\mathscr{D}), c_i(\mathscr{D}_0) < c$. Using 
the fact that 
$$
\sqrt{x + \frac{1}{4}} \leq \frac{\Gamma (x+1)}{\Gamma \left( x+\frac{1}{2} \right)} 
\leq \sqrt{x + \frac{1}{2}} 
$$

\noindent
for $x > 0$ (see \cite{Watson:1959}), it follows that for a large enough constant $M'$ 
and large enough $n$, we have 
\begin{align} \label{new5}
\frac{\frac{\Gamma(\frac{n}{2} + \frac{c_i(\mathscr D)}{2} - 1)}{\Gamma(\frac{c_i(\mathscr D)}{2} - 1)}}{\frac{\Gamma(\frac{n}{2} + \frac{c_i(\mathscr D_0)}{2} - 1)}{\Gamma(\frac{c_i(\mathscr D_0)}{2} - 1)}} \le M'n^c,
\end{align}

\noindent
Using Assumption 5, (\ref{m4.1}), (\ref{m4.2}), (\ref{new5}) and the definition 
of $d = d_n$, it follows that 
\begin{align*} 
\frac{\frac{z_{\mathscr{D}}(U+nS, n+\bm{\alpha}(\mathscr D))}{z_{\mathscr{D}}(U,
		\bm{\alpha}(\mathscr D))}}{\frac{z_{\mathscr{D}_0}(U+nS, n+\bm{\alpha}(\mathscr 
		D_0))}{z_{\mathscr{D}_0}(U,\bm{\alpha}(\mathscr D_0))}} \le \prod_{i=1}^{p} M
\left(\frac{\delta_2}{\delta_1}\right)^{\frac d 2} n^{2c} \left(\sqrt {\frac{\delta_2} n}
\right)^{\nu_i({\mathscr{D}}) - \nu_i({\mathscr{D}}_0)} \frac{|\tilde{S}_{\mathscr{D}}^{>i}|
	^{\frac{n+c_i(\mathscr D)-3}{2}}}{|\tilde{S}_{\mathscr{D}}^{\ge i}|^{\frac{n+c_i(\mathscr 
			D)-2}{2}}}.
\end{align*}

\noindent
where $M = M' \max (\delta_2^c, 1) \max (\delta_1^{-c}, 1)$. It follows from (\ref{m4}) 
and $\nu_p (\mathscr{D}) = \nu_p (\mathscr{D_0}) = 0$ that 
\begin{align*} 
\begin{split}
\frac{\pi({\mathscr{D}}|\bm{Y})}{\pi({\mathscr{D}}_0|\bm{Y})} \le \prod_{i=1}^{p} M
\left(\frac{\delta_2}{\delta_1}\right)^{\frac d 2} n^{2c}\left (\sqrt {\frac{\delta_2} n}\frac{q}
{1-q}\right)^{\nu_i({\mathscr{D}}) - \nu_i({\mathscr{D}}_0)} 
\frac{\frac{|\tilde{S}_{\mathscr{D}}^{>i}|^{\frac{n+c_i(\mathscr D)-3}{2}}}
	{|\tilde{S}_{\mathscr{D}}^{\ge i}|^{\frac{n+c_i(\mathscr D)-2}{2}}}}{\frac{
		|\tilde{S}_{{\mathscr{D}}_0}^{>i}|^{\frac{n+c_i(\mathscr D_0)-3}{2}}}
	{|\tilde{S}_{{\mathscr{D}}_0}^{\ge i}|^{\frac{n+c_i(\mathscr D_0)-2}{2}}}}.
\end{split}
\end{align*}

\noindent
Next, note that
\begin{align*} 
|\tilde{S}_{\mathscr{D}}^{\ge i}| &= |\tilde{S}_{\mathscr{D}}^{>i}| |\tilde{S}_{ii} - 
(\tilde{S}_{\mathscr D \cdot i}^>)^T (\tilde{S}_{\mathscr{D}}^{>i})^{-1} 
\tilde{S}_{\mathscr D \cdot i}^>| \\
&=  |\tilde{S}_{\mathscr{D}}^{> i}| \tilde{S}_{i|pa_i({\mathscr{D}})}.
\end{align*}

\noindent
Hence, the upper bound can be simplified as
\begin{align} \label{m5}
\begin{split}
\frac{\pi({\mathscr{D}}|\bm{Y})}{\pi({\mathscr{D}}_0|\bm{Y})} \le& \prod_{i=1}^{p} M
\left(\frac{\delta_2}{\delta_1}\right)^{\frac d 2} n^{2c}\left (\sqrt {\frac{\delta_2} n}\frac{q}
{1-q}\right)^{\nu_i({\mathscr{D}}) - \nu_i({\mathscr{D}}_0)} 
\frac{|\tilde{S}_{\mathscr{D}_0}^{\ge i}|^{\frac12}}
{|\tilde{S}_{\mathscr{D}}^{\ge i}|^{\frac12}} \\
&\times \frac{\left(\tilde{S}_{i|pa_i({\mathscr{D}}_0)}\right)^{\frac{n+c_i(\mathscr 
			D_0)-3}{2}}}{\left(\tilde{S}_{i|pa_i({\mathscr{D}})}\right)^{\frac{n+c_i(\mathscr D)-3}{2}}}. 
\end{split}
\end{align}

\section{Proof of Proposition \ref{f1}} 

\noindent
(a) Note that under the true model, we have
\begin{align*}
&\bm{Y}\sim MVN(0, (L_0(D_0)^{-1}L_0^T)^{-1}), (L_0,D_0) \in \Theta_{{\mathscr{D}}_0},\\
&Var \bm{Y} = (L_0(D_0)^{-1}L_0^T)^{-1},\\
&Var(L_0^T \bm{Y})=D.
\end{align*}
Let
$$\bm{Z}=L_0^T\bm{Y}=\begin{pmatrix}
Y_1 + \sum_{i \ge 2}(L_0)_{i1}Y_i\\
\vdots \\
Y_j + \sum_{i \ge j+1}(L_0)_{ij}Y_i\\
\vdots\\
Y_p
\end{pmatrix}.$$

\noindent	
Since $Var(\bm{Z}) = D$, $Z_1, \cdots, Z_p$ are mutually independent. Since 
$\bm{Y} = (L_0^T)^{-1} \bm{Z}$, it follows that 
$$
Z_i \perp \{Y_{i+1}, \cdots, Y_p\}. 
$$

\noindent
Suppose $M \supseteq pa_{i}({\mathscr{D}}_0), M \subseteq \{i+1, \cdots, p\}$. Since 
$(L_0)_{ki} = 0$ for $k \in M \setminus pa_{i}({\mathscr{D}}_0)$, it follows that 
\begin{align*}
(\Sigma_0)_{i|M} &= Var(Y_i |\bm{Y}_M)\\
&= Var(Y_i+\sum_{k \in pa_{i}({\mathscr{D}}_0)} (L_0)_{ki}Y_k |\bm{Y_M})\\
&= Var(Z_i |\bm{Y_M}) = Var(Z_i) = (D_0)_{ii} = 
(\Sigma_0)_{i|pa_{i}({\mathscr{D}}_0)},
\end{align*}

\noindent
The result follows, since by hypothesis $pa_{i}({\mathscr{D}}) \supseteq 
pa_{i}({\mathscr{D}}_0)$, and by definition $pa_{i}({\mathscr{D}}) \subseteq 
\{i+1, \ldots, p\}$. 

(b) Let $\bm{Z}$ be as defined in the proof of part (a). since $pa_{i}(\mathscr{D}) 
\subseteq pa_i (\mathscr{D}_0)$, we have 
$$
Z_i \perp (\bm{Y}_{pa_i (\mathscr{D}_0) \backslash pa_i (\mathscr{D})}, 
\bm{Y}_{pa_i (\mathscr{D})}). 
$$

\noindent
Hence, 
\begin{align*}
(\Sigma_0)_{i|pa_{i}(\mathscr{D})} &= Var \left( Z_i - \sum_{j \in 
	pa_{i}({\mathscr{D}}_0)} (L_0)_{ji} Y_j | \bm{Y}_{pa_{i}(\mathscr{D})} \right) \\
& = Var \left( Z_i - \sum_{j \in pa_{i}({\mathscr{D}}_0)\backslash pa_{i}({\mathscr{D}})}(L_0)_{ji}y_j | \bm{Y}_{pa_{i}(\mathscr{D})} \right)\\
& = (D_0)_{ii} + Var \left( \sum_{j \in pa_{i}({\mathscr{D}}_0)\backslash pa_{i}({\mathscr{D}})}(L_0)_{ji}y_j | \bm{Y}_{pa_{i}(\mathscr{D})} \right)\\
& \ge  (\Sigma_0)_{i|pa_i (\mathscr D _0)} + L_0^T Var(\bm{Y}_{pa_{i}({\mathscr{D}}_0)\backslash pa_{i}({\mathscr{D}})} | \bm{Y}_{pa_{i}(\mathscr{D})})L\\
& \ge (\Sigma_0)_{i|pa_i (\mathscr D _0)} + \epsilon_{0,n}\sum_{j \in pa_{i}({\mathscr{D}}_0)\backslash pa_{i}({\mathscr{D}})}(L_0)_{ji}^2\\
& \ge (\Sigma_0)_{i|pa_i (\mathscr D _0)} + \epsilon_{0,n}(\nu_i(\mathscr{D}_0) - \nu_i(\mathscr{D}))s^2.
\end{align*}

\noindent
The second last inequality follows from Assumption 1.

\section{Proof of Lemma \ref{lm5}}

\noindent
Let $pa_i^* = pa_i({\mathscr{D}} ) \cap pa_i({\mathscr{D}}_0 )$, $\nu_i^* = |
pa_i({\mathscr{D}}^* )|$ and ${{\mathscr{D}}^*} $ be an arbitrary DAG with 
$pa_i({{\mathscr{D}}^*} ) = pa_i^*$. It follows that $pa_i^* \subset 
pa_i({\mathscr{D}} _0)$ and $pa_i^* \subset pa_i({\mathscr{D}} )$. Recall by 
Lemma \ref{newlemma1} that, 
\begin{align*} 
B_i({\mathscr{D}} , {\mathscr{D}}_0 ) =&  M\left(\frac{\delta_2}{\delta_1}\right)^{\frac d 2} n^{2c}\left (\sqrt {\frac{\delta_2} n}\frac{q}{1-q}\right)^{\nu_i({\mathscr{D}}) - \nu_i({\mathscr{D}}_0)}  \\
&\times \frac{|\tilde{S}_{\mathscr{D}_0}^{\ge i}|^{\frac12}}{|\tilde{S}_{\mathscr{D}}^{\ge i}|^{\frac12}} \frac{\left(\tilde{S}_{i|pa_i({\mathscr{D}}_0)}\right)^{\frac{n+c_i(\mathscr D_0)-3}{2}}}{\left(\tilde{S}_{i|pa_i({\mathscr{D}})}\right)^{\frac{n+c_i(\mathscr D)-3}{2}}}. 
\end{align*}

\noindent
Therefore,
\begin{align} \label{pp3}
\begin{split}
&B_i({\mathscr{D}} , {\mathscr{D}}_0 )\\
\le& M\left(\frac{\delta_2}{\delta_1}\right)^{\frac d 2} n^{2c} \left (\sqrt {\frac{\delta_2} n}
\frac{q}{1-q}\right)^{\nu_i({\mathscr{D}}) - \nu_i({\mathscr{D}}^*)} \\
&\times  \frac{|\tilde{S}_{\mathscr{D}^*}^{\ge i}|^{\frac12}}{|\tilde{S}_{\mathscr{D}}^{\ge 
		i}|^{\frac12}} \frac{\left(\tilde{S}_{i|pa_i({\mathscr{D}}^*)}\right)^{\frac{n+c_i(\mathscr 
			D^*)-3}{2}}}{\left(\tilde{S}_{i|pa_i({\mathscr{D}})}\right)^{\frac{n+c_i(\mathscr D)-3}{2}}}\\
&\times M\left(\frac{\delta_2}{\delta_1}\right)^{\frac d 2} n^{2c} \left (\sqrt 
{\frac{\delta_2} n}\frac{q}{1-q}\right)^{\nu_i({\mathscr{D}^*}) - \nu_i({\mathscr{D}}_0)}\\
&\times \frac{|\tilde{S}_{\mathscr{D}_0}^{\ge i}|^{\frac12}}{|\tilde{S}_{\mathscr{D}^*}
	^{\ge i}|^{\frac12}} \frac{\left(\tilde{S}_{i|pa_i({\mathscr{D}}_0)}\right)^{\frac{n
			+c_i(\mathscr D_0)-3}{2}}}{\left(\tilde{S}_{i|pa_i({\mathscr{D}^*})}\right)^{\frac{n+c_i(\mathscr D^*)-3}{2}}} \\
=& B_i({\mathscr{D}} , {{\mathscr{D}}^*} ) B_i({{\mathscr{D}}^*} , {\mathscr{D}}_0 ).	
\end{split}
\end{align}

\noindent
Note that $pa_i^* \subseteq pa_i (\mathscr{D})$. If $\nu_i({\mathscr{D}} ) \le 
3\nu_i({\mathscr{D}}^* ) + 2$, then using the fact that $pa_i^* \subseteq pa_i 
(\mathscr{D}_0)$ and following exactly the same arguments as in the proof of 
Lemma \ref{lm1}, it can be shown that 
\begin{align*} 
B_i({\mathscr{D}} , {{\mathscr{D}}^*} ) < \epsilon_{1,n}, \mbox{ for } n \ge N_1.
\end{align*}

\noindent
If $\nu_i({\mathscr{D}}) > 3\nu_i({\mathscr{D}}^* ) + 2$, then using the 
fact that $pa_i^* \subseteq pa_i (\mathscr{D}_0)$ and following exactly the same 
arguments in the proofs of Lemma \ref{lm2} and Lemma \ref{lm3}, it can be shown that 
$$
B_i({\mathscr{D}} , {{\mathscr{D}}^*} )  < \mbox{max}(\epsilon_{2,n},\epsilon_{3,n}), 
\mbox{ for } n \ge \mbox{max}(N_2,N_3). 
$$

\noindent
Hence, 
$$
B_i({\mathscr{D}} ,{{\mathscr{D}}^*} ) \le \mbox{max}(\epsilon_{1,n},
\epsilon_{2,n},\epsilon_{3,n}), \mbox{ for } n \ge \mbox{max}(N_1, N_2, N_3). 
$$

\noindent
Since $pa_i^* \subset pa_i({\mathscr{D}}_0 )$, it follows by Lemma \ref{lm4} that 
$$
B_i({{\mathscr{D}}^*} , {\mathscr{D}}_0 ) \le \epsilon_{4, n}, \mbox{ for } n \ge N_4. 
$$ 

\noindent
By (\ref{pp3}), we get 
\begin{align*} 
B_i({\mathscr{D}} , {\mathscr{D}}_0 ) \le \mbox{max}(\epsilon_{1,n},\epsilon_{2,n},
\epsilon_{3,n})\epsilon_{4,n}, \mbox{ for } n \ge \max (N_1, N_2, N_3, N_4). 
\end{align*}

\section{Proof of Theorem \ref{thm4}}

\noindent
Note that 
\begin{equation} \label{posterior1}
\pi(\mathscr D | \bm Y) \propto \pi(\mathscr D) \times \int\int \prod_{j = 1}^{p-1} \left\{D_{jj}^{-(\frac n 2 + 1)} \left( \prod_{i > j, (i,j) \in 
	E}L_{ij}^{2r} \right) e^{-\frac{nL_{.j}^TSL_{.j}}{2D_{jj}}} \right\} dL dD,
\end{equation}

\noindent
where $L_{.j}^TSL_{.j} = (L_{\mathscr D.j}^\ge)^TS_{\mathscr D}^{\ge j}L_{\mathscr D.j}^\ge.$ Therefore,
\begin{align*}
\pi(\mathscr D | \bm Y) \propto \pi(\mathscr D)\times\int\int \prod_{j = 1}^{p-1} \left\{D_{jj}^{-(\frac n 2 + 1)} \left( \prod_{i \in pa_j(\mathscr D)}L_{ij}^{2r} \right) e^{-\frac{n \left(L_{\mathscr D.j}^\ge\right)^TS_{\mathscr D}^{\ge j}L_{\mathscr D.j}^\ge}{2D_{jj}}} \right\} dL dD.
\end{align*}
Note that 
\begin{align*}
\left(L_{\mathscr D.j}^\ge\right)^TS_{\mathscr D}^{\ge j}L_{\mathscr D.j}^\ge = \left(1,\left(L_{\mathscr D.j}^>\right)^T\right) \times \quad
\begin{pmatrix} 
S_{jj} & \left(S_{\mathscr D.j}^>\right)^T \\
S_{\mathscr D.j}^> & S_\mathscr D^{>j}
\end{pmatrix}
\quad
\times  \left(1,L_{\mathscr D.j}^>\right).
\end{align*}
Therefore, we have
\begin{align*}
&\int \left( \prod_{i \in pa_j(\mathscr D)}L_{ij}^{2r} \right) e^{-\frac{n \left(L_{\mathscr D.j}^\ge\right)^TS_{\mathscr D}^{\ge j}L_{\mathscr D.j}^\ge}{2D_{jj}}} dL\\
=& \int \prod_{i \in pa_j(\mathscr D)} \left(\left[L_{\mathscr D.j}^>\right]_i\right)^{2r}\exp\left\{-\frac{\left(L_{\mathscr D.j}^> + \left(S_\mathscr D^{>j}\right)^{-1}S_{\mathscr D.j}^>\right)^TS_\mathscr D^{>j}\left(L_{\mathscr D.j}^> + \left(S_\mathscr D^{>j}\right)^{-1}S_{\mathscr D.j}^>\right)}{\frac{2D_{jj}}{n}}\right\}\\
&\times \exp\left\{-\frac{S_{jj} - \left(S_{\mathscr D.j}^>\right)^T\left(S_\mathscr D^{>j}\right)^{-1}S_{\mathscr D.j}^>}{\frac{2D_{jj}}{n}}\right\}dL.
\end{align*}

\noindent
Note by hypothesis that the number of edges in $\mathscr{D}$ is bounded by $d \left( \frac{n}{\log p} \right)^{\frac{1}{2(2+k)}}$. It 
follows by Assumption 1 and Assumption 2 that $\frac{\nu_j(\mathscr D)\sqrt{\frac{\log p}n}}{\epsilon_{0,n}^4} \rightarrow 0,$ as 
$n \rightarrow \infty$. Using this, we obtain 
$$
\frac{2}{\epsilon_{0,n}} \ge \lVert S_\mathscr D^{>j}  \rVert_{(2,2)} \ge \frac{\epsilon_{0,n}}{2}. 
$$

\noindent
It follows from \cite[Lemma 6]{Johnson:Rossell:supplemental:2012} that 
\begin{align*}
&\int \left( \prod_{i \in pa_j(\mathscr D)}L_{ij}^{2r} \right) e^{-\frac{n \left(L_{\mathscr D.j}^\ge\right)^TS_{\mathscr D}^{\ge j}L_{\mathscr D.j}^\ge}{2D_{jj}}} dL\\
\le &\left(\frac{4\pi}{n\epsilon_{0,n}}D_{jj}\right)^{\frac{\nu_j(\mathscr D)}{2}}\times\left(\frac{4\left(S_{\mathscr D.j}^>\right)^T\left(S_\mathscr D^{>j}\right)^{-2}S_{\mathscr D.j}^>}{\nu_j(\mathscr D)} + \frac{8\left((2r-1)!!\right)^{\frac 1 r}}{n\epsilon_{0,n}}D_{jj}\right)^{r\nu_j(\mathscr D)}\times \exp\left\{-\frac{\frac {nS_{j|pa_j(\mathscr D)}}2}{D_{jj}}\right\}\\
\le &\left(M_1\frac{\left(\frac 2 {\epsilon_{0,n}}\right)^{2r\nu_j(\mathscr D) + 2r + \frac 1 2}}{\sqrt{n}\nu_j(\mathscr D)^r}\right)^{\nu_j(\mathscr D)} \times D_{jj}^{\frac {\nu_j(\mathscr D)}2}\times \exp\left\{-\frac{\frac {nS_{j|pa_j(\mathscr D)}}2}{D_{jj}}\right\}\\
&+ \left(\frac{M_2}{(n\epsilon_{0,n})^{r+\frac 1 2}}\right)^{\nu_j(\mathscr D)}\times D_{jj}^{r\nu_j(\mathscr D) + \frac{\nu_j(\mathscr D)}{2}} \times \exp\left\{-\frac{\frac {nS_{j|pa_j(\mathscr D)}}2}{D_{jj}}\right\}, 
\end{align*}

\noindent
where $M_1$ and $M_2$ are constants not depending on $n$. Hence, 
\begin{align*}
&\pi(\mathscr D)\times\int\int \prod_{j = 1}^{p-1} \left\{D_{jj}^{-(\frac n 2 + 1)} \left( \prod_{i > j, (i,j) \in E}L_{ij}^{2r} \right) e^{-\frac{nL_{.j}^TSL_{.j}}{2D_{jj}}} \right\} dL dD\\
\le &\pi(\mathscr D) \times \prod_{j = 1}^{p-1}  \Biggl[\int  \left(M_1\frac{\left(\frac 2 {\epsilon_{0,n}}\right)^{2r\nu_j(\mathscr D) + 2r + \frac 1 2}}{\sqrt{n}\nu_j(\mathscr D)^r}\right)^{\nu_j(\mathscr D)} \times D_{jj}^{-\left(\frac {n-\nu_j(\mathscr D)}2+1\right)}\times \exp\left\{-\frac{\frac {nS_{j|pa_j(\mathscr D)}}2}{D_{jj}}\right\}\\
&+ \left(\frac{M_2}{(n\epsilon_{0,n})^{r+\frac 1 2}}\right)^{\nu_j(\mathscr D)}\times D_{jj}^{-\left(\frac{n - \nu_j(\mathscr D)}2 - r\nu_j(\mathscr D) + 1\right)  } \times \exp\left\{-\frac{\frac {nS_{j|pa_j(\mathscr D)}}2}{D_{jj}}\right\} dD_{jj}\Biggr]\\
\le & \pi(\mathscr D) \times \prod_{j = 1}^{p-1} \Biggl[\left(M_1\frac{\left(\frac 2 {\epsilon_{0,n}}\right)^{2r\nu_j(\mathscr D) + 2r + \frac 1 2}}{\sqrt{n}\nu_j(\mathscr D)^r}\right)^{\nu_j(\mathscr D)} \frac{\Gamma\left(\frac {n-\nu_j(\mathscr D)}2\right)}{\left(\frac {nS_{j|pa_j(\mathscr D)}}2\right)^{\frac {n-\nu_j(\mathscr D)}2}} \\
&+ \left(\frac{M_2}{(n\epsilon_{0,n})^{r+\frac 1 2}}\right)^{\nu_j(\mathscr D)}\frac{\Gamma\left(\frac{n - \nu_j(\mathscr D)}2 - r\nu_j(\mathscr D)\right)}{\left(\frac {nS_{j|pa_j(\mathscr D)}}2\right)^{\frac{n - \nu_j(\mathscr D)}2 - r\nu_j(\mathscr D)}} \Biggr],
\end{align*} 

\noindent
for large enough $M_1$ and $M_2$. Now consider the true DAG $\mathscr D_0$. Recall that $s = 
\min_{1 \leq j \leq p, i \in pa_j(\mathscr{D}_0^n)} |(L_0^n)_{ij}|$, and $\left( (L_0^n)_{ij} \right)_{i \in pa_j(\mathscr{D}_0^n)} 
= \left((\Sigma_0)_{\mathscr D_0}^{>j}\right)^{-1} (\Sigma_0)_{\mathscr D_0}^{>j}$. By Assumptions 1, 2, 4 and (\ref{pp1}) in 
the main paper, it follows that $\left\| \left(S_{\mathscr D_0}^{>j}\right)^{-1} \right\|_{(2,2)} \leq 2/\epsilon_{0,n}$ and 
$\left\lvert \left( \left(S_{\mathscr D_0}^{>j}\right)^{-1}S_{\mathscr D_0}^{>j} \right)_i \right\rvert > \frac{s}{2}$ for every 
$i \in pa_j (\mathscr{D}_0^n)$ (for large enough $n$ on the event $C_n^c$, as defined in (\ref{smplbound2}) in the main 
paper). Using the fact that $E[(\sigma Z + \mu)^{2r}] \geq \mu^{2r}$ for a standard normal variable $Z$, we get 
\begin{align*}
\int \left( \prod_{i \in pa_j({\mathscr D_0})}L_{ij}^{2r} \right) e^{-\frac{n \left(L_{{\mathscr D_0}.j}^\ge\right)^TS_{{\mathscr D_0}}^{\ge j}L_{{\mathscr D_0}.j}^\ge}{2D_{jj}}} dL &\ge \left(\frac{4\pi}{n\epsilon_{0,n}}D_{jj}\right)^{\frac{\nu_j({\mathscr D_0})}{2}} \left(\frac s 2\right)^{2r\nu_j(\mathscr D_0)}\exp\left\{-\frac{\frac {nS_{j|pa_j(\mathscr D_0)}}2}{D_{jj}}\right\}\\
&\ge \left(\frac{M_3s^{2r}}{\sqrt{n\epsilon_{0,n}}}\right)^{\nu_j(\mathscr D_0)}D_{jj}^{\frac{\nu_j({\mathscr D_0})}{2}}\exp\left\{-\frac{\frac {nS_{j|pa_j(\mathscr D_0)}}2}{D_{jj}}\right\},
\end{align*} 

\noindent
for an appropriate constant $M_3$. Therefore, under the true DAG $\mathscr D_0$,
\begin{align*}
&\pi({\mathscr D_0})\times\int\int \prod_{j = 1}^{p-1} \left\{D_{jj}^{-(\frac n 2 + 1)} \left( \prod_{i > j, (i,j) \in E}L_{ij}^{2r} \right) 
e^{-\frac{nL_{.j}^TSL_{.j}}{2D_{jj}}} \right\} dL dD\\
\ge & \pi({\mathscr D_0}) \times \prod_{j = 1}^{p-1} \int D_{jj}^{-(\frac n 2 +1)}\left(\frac{M_3s^{2r}}{\sqrt{n\epsilon_{0,n}}}
\right)^{\nu_j({\mathscr D_0})}D_{jj}^{\frac{\nu_j({{\mathscr D_0}})}{2}}\exp\left\{-\frac{\frac {nS_{j|pa_j(\mathscr D_0)}}2}{D_{jj}}
\right\} dD_{jj}\\
\ge &  \pi({\mathscr D_0}) \times \prod_{j = 1}^{p-1} \left(\frac{M_3s^{2r}}{\sqrt{n\epsilon_{0,n}}}\right)^{\nu_j({\mathscr D_0})} 
\frac{\Gamma\left(\frac {n-\nu_j({\mathscr D_0})}2\right)}
{\left(\frac {nS_{j|pa_j(\mathscr D_0)}}2\right)^{\frac {n-\nu_j(\mathscr D_0)}2}}.
\end{align*} 

\noindent
It follows that the posterior ratio of any DAG $\mathscr D$ and true DAG $\mathscr D_0$ is bounded by
\begin{align} 
\begin{split}
\frac{\pi(\mathscr D | \bm Y)}{\pi(\mathscr D_0 | \bm Y)}
&\le \prod_{j = 1}^{p-1} \left(\frac{q}{1-q}\right)^{\nu_j(\mathscr D) - \nu_j(\mathscr D_0)}  \Biggl[ \frac{\left(M_1\frac{\left(\frac 2 {\epsilon_{0,n}}\right)^{2r\nu_j(\mathscr D) + 2r + \frac 1 2}}{\sqrt{n}\nu_j(\mathscr D)^r}\right)^{\nu_j(\mathscr D)} \frac{\Gamma\left(\frac {n-\nu_j(\mathscr D)}2\right)}{\left(\frac {nS_{j|pa_j(\mathscr D)}}2\right)^{\frac {n-\nu_j(\mathscr D)}2}}}{\left(\frac{M_3s^{2r}}{\sqrt{n\epsilon_{0,n}}}\right)^{\nu_j({\mathscr D_0})} \frac{\Gamma\left(\frac {n-\nu_j({\mathscr D_0})}2\right)}{\left(\frac {nS_{j|pa_j(\mathscr D_0)}}2\right)^{\frac {n-\nu_j(\mathscr D_0)}2}}} \\
&+ \frac{\left(\frac{M_2}{(n\epsilon_{0,n})^{r+\frac 1 2}}\right)^{\nu_j(\mathscr D)}\frac{\Gamma\left(\frac{n - \nu_j(\mathscr D)}2 - r\nu_j(\mathscr D)\right)}{\left(\frac {nS_{j|pa_j(\mathscr D)}}2\right)^{\frac{n - \nu_j(\mathscr D)}2 - r\nu_j(\mathscr D)}}}{\left(\frac{M_3s^{2r}}{\sqrt{n\epsilon_{0,n}}}\right)^{\nu_j({\mathscr D_0})} \frac{\Gamma\left(\frac {n-\nu_j({\mathscr D_0})}2\right)}{\left(\frac {nS_{j|pa_j(\mathscr D_0)}}2\right)^{\frac {n-\nu_j(\mathscr D_0)}2}}} \Biggr].
\end{split}
\end{align}

\noindent
Using Stirling bounds for the Gamma function, we have 
\begin{align} \label{nonlocal_ratio_bound}
\begin{split}
\frac{\pi(\mathscr D | \bm Y)}{\pi(\mathscr D_0 | \bm Y)}
&\le \prod_{j = 1}^{p-1} \Biggl[ \left(\frac{Mq}{(1-q)\sqrt{n\epsilon_{0,n}}}\right)^{\nu_j(\mathscr D) - \nu_j(\mathscr D_0)} \frac{\left(\frac{2}{\epsilon_{0,n}}\right)^{2r\nu_j(\mathscr D)^2 + 2r \nu_j(\mathscr D)}}{s^{2r\nu_j(\mathscr D_0)}\nu_j(\mathscr D)^{r\nu_j(\mathscr D)}} \frac{\left(\frac{1}{S_{j|pa_j(\mathscr D)}}\right)^{\frac{n-\nu_j(\mathscr D)}{2}}}{\left(\frac{1}{S_{|pa_j(\mathscr D_0)}}\right)^{\frac{n-\nu_j(\mathscr D_0)}{2}}}\\
&+ \left(\frac{Mq}{(1-q)\sqrt{n\epsilon_{0,n}}}\right)^{\nu_j(\mathscr D) - \nu_j(\mathscr D_0)} \frac{1}{n^{r\nu_j(\mathscr D)}} \frac{\left(\frac{1}{\epsilon_{0,n}}\right)^{2r\nu_j(\mathscr D)}}{s^{2r\nu_j(\mathscr D_0)}} \frac{\left(\frac{1}{S_{j|pa_j(\mathscr D)}}\right)^{\frac{n-\nu_j(\mathscr D)}{2}}}{\left(\frac{1}{S_{|pa_j(\mathscr D_0)}}\right)^{\frac{n-\nu_j(\mathscr D_0)}{2}}} \Biggr]\\
&\le \prod_{j = 1}^{p-1} \left(2\left(\frac{Mq}{(1-q)\sqrt{n\epsilon_{0,n}}}\right)^{\nu_j(\mathscr D) - \nu_j(\mathscr D_0)} \frac{\left(\frac{2}{\epsilon_{0,n}}\right)^{2r\nu_j(\mathscr D)^2 + 2r \nu_j(\mathscr D)}}{s^{2r\nu_j(\mathscr D_0)}\nu_j(\mathscr D)^{r\nu_j(\mathscr D)}} \frac{\left(\frac{1}{S_{j|pa_j(\mathscr D)}}\right)^{\frac{n-\nu_j(\mathscr D)}{2}}}{\left(\frac{1}{S_{|pa_j(\mathscr D_0)}}\right)^{\frac{n-\nu_j(\mathscr D_0)}{2}}} \right) \\
\triangleq& \prod_{j=1}^{p-1}C_i(\mathscr{D},\mathscr{D}_0),
\end{split} 
\end{align}

\noindent
for an appropriate constant $M$. The last inequality follows from $\nu_j(\mathscr D) < n$. 

We now analyze $C_j (\mathscr D, \mathscr D_0)$ when $\nu_j(\mathscr D) \neq \nu_j(\mathscr D_0)$. Let $pa_j^* = pa_j({\mathscr{D}} ) \cap pa_j({\mathscr{D}}_0 )$ and ${{\mathscr{D}}^*} $ be an arbitrary DAG with 
$pa_j({{\mathscr{D}}^*} ) = pa_j^*$. It follows that $pa_j({{\mathscr{D}}^*} ) \subset 
pa_j({\mathscr{D}} _0)$ and $pa_j^* \subseteq pa_j({\mathscr{D}} )$. Now we first consider the case when $\nu_j(\mathscr D) 
< \nu_j(\mathscr D_0).$ Note that $d - 1 \ge \nu_j(\mathscr D_0) - \nu_j(\mathscr D) \ge 0$ and $\nu_j(\mathscr D^*) \le  
\nu_j(\mathscr D_0) - 1.$ It follows from (\ref{nonlocal_ratio_bound}) that
\begin{align} \label{smallDAG_ratio_bound}
\begin{split}
&C_j(\mathscr{D},\mathscr{D}_0) \\
\le & 2\left(\frac{Mq}{(1-q)\sqrt{n\epsilon_{0,n}}}\right)^{\nu_j(\mathscr D) - \nu_j(\mathscr D_0)} \frac{\left(\frac{2}{\epsilon_{0,n}}\right)^{2r\nu_j(\mathscr D)^2 + 2r \nu_j(\mathscr D)}}{s^{2r\nu_j(\mathscr D_0)}\nu_j(\mathscr D)^{r\nu_j(\mathscr D)}} \frac{\left(\frac{1}{S_{j|pa_j(\mathscr D)}}\right)^{\frac{n-\nu_j(\mathscr D)}{2}}}{\left(\frac{1}{S_{|pa_j(\mathscr D^*)}}\right)^{\frac{n-\nu_j(\mathscr D_0)}{2}}}\frac{\left(\frac{1}{S_{j|pa_j(\mathscr D^*)}}\right)^{\frac{n-\nu_j(\mathscr D_0)}{2}}}{\left(\frac{1}{S_{|pa_j(\mathscr D_0)}}\right)^{\frac{n-\nu_j(\mathscr D_0)}{2}}} \\
\le& 2\left(\frac{M\sqrt{n}}{q\sqrt{\epsilon_{0,n}}}\right)^{\nu_j(\mathscr D_0) - \nu_j(\mathscr D)} \left(\frac{2}{\epsilon_{0,n}}\right)^{2rd^2 + 2rd}  \frac{1}{s^{2r\nu_j(\mathscr D_0)}}  \left(\frac{\frac{1}{S_{j|pa_j(\mathscr D)}}}{\frac{1}{S_{j|pa_j(\mathscr D^*)}}}\right)^{\frac{n-\nu_j(\mathscr D_0)}{2}} \left(\frac{\frac{1}{S_{j|pa_j(\mathscr D_0)}}}{\frac{1}{S_{j|pa_j(\mathscr D^*)}}}\right)^{-\frac{n-\nu_j(\mathscr D_0)}{2}}\\
\le & 2 \exp\left\{\frac{(2rd^2+2rd)\log\left(\frac 2 {\epsilon_{0,n}}\right) - 2r\log s + d\log\left(\frac{M\sqrt n}{\sqrt{\epsilon_{0,n}}}\right)}{\frac{n-d}2} + 2\eta(\nu_j(\mathscr D_0) - \nu_j(\mathscr D)) \right\}^{\frac{n-d}{2}}\\\
&\times \left(1 + \frac{\frac {2c^\prime} {\epsilon_{0,n}^3}(\nu_j({\mathscr{D}} )+\nu_j({\mathscr{D}}^*)+2)\sqrt{\frac{\log p }{n}}}{\epsilon_{0,n}/2} \right )^{\frac{n-1}{2}}   \\
&\times \left(1+ \frac{(\frac{1}{(\Sigma_0)_{i|pa_j({\mathscr{D}}_0 )}} - \frac{1}{(\Sigma_0)_{i|pa_j({\mathscr{D}^*} )}}) - 2K_1 \frac d {\epsilon_{0,n}^3} \sqrt{\frac{\log p }{n}} }{\frac{1}{(\Sigma_0)_{i|pa_j({\mathscr{D}^*} )}} + K_1\frac d {\epsilon_{0,n}^3} \sqrt{\frac{\log p }{n}}} \right)^{- \frac{n-d}{2}} \\
\le & 2 \exp\left\{\frac{(2rd^2+2rd)\log\left(\frac 2 {\epsilon_{0,n}}\right) + d\log\left(\frac{M\sqrt n}{\sqrt{\epsilon_{0,n}}}\right)}{\frac{n-d}2} + 2\eta(\nu_j(\mathscr D_0) - \nu_j(\mathscr D)) + \frac{n-1}{n-d}\frac{16c^\prime}{\epsilon_{0,n}^4}d\sqrt{\frac{\log p }{n}}\right\}^{\frac{n-d}{2}}\\
&\times \left(1+ \frac{\epsilon_{0,n} s^2 (\nu_j({\mathscr{D}} _0) - \nu_j({\mathscr{D}^*} )) - 
	2K_1\frac d {\epsilon_{0,n}^3} \sqrt{\frac{\log p }{n}}}{2/\epsilon_{0,n}} \right)^{- \frac{n-d}{2}}.
\end{split}
\end{align}
By following the exact steps in the proof of Lemma 5.6, when $\mathscr D \subset \mathscr D_0$, we get
\begin{align}
\begin{split}
C_j(\mathscr{D},\mathscr{D}_0)  &\le 
2\exp\left\{-\left(\frac{\frac{\epsilon_{0,n}^2}{8}s_n^2}{1 + \frac{\epsilon_{0,n}^2}{4}
	s_n^2}\right)\left(\frac{n-d}{2} \right)\right\} \\
&\le 2e^{-\frac{1}{10}\epsilon_{0,n}^2 s_n^2(\frac{n-d}{2})} \le 2 e^{-d\eta n}. 
\end{split}
\end{align}
Now consider the case when $\nu_j(\mathscr D) > \nu_j(\mathscr D_0).$ By (\ref{nonlocal_ratio_bound}), similar to the proof of Lemma 5.4, note that $S_{j|pa_j(\mathscr D^*)} \ge S_{j|pa_j(\mathscr D_0)}$. We have
\begin{align}
\begin{split}
&C_j(\mathscr{D},\mathscr{D}_0) \\
\le & 2\left(\frac{Mq}{(1-q)\sqrt{n\epsilon_{0,n}}}\right)^{\nu_j(\mathscr D) - \nu_j(\mathscr D_0)} \frac{\left(\frac{2}{\epsilon_{0,n}}\right)^{2r\nu_j(\mathscr D)^2 + 2r \nu_j(\mathscr D)}}{s^{2r\nu_j(\mathscr D_0)}\nu_j(\mathscr D)^{r\nu_j(\mathscr D)}} \frac{\left(\frac{1}{S_{j|pa_j(\mathscr D)}}\right)^{\frac{n-\nu_j(\mathscr D)}{2}}}{\left(\frac{1}{S_{|pa_j(\mathscr D^*)}}\right)^{\frac{n-\nu_j(\mathscr D_0)}{2}}}\frac{\left(\frac{1}{S_{j|pa_j(\mathscr D^*)}}\right)^{\frac{n-\nu_j(\mathscr D_0)}{2}}}{\left(\frac{1}{S_{|pa_j(\mathscr D_0)}}\right)^{\frac{n-\nu_j(\mathscr D_0)}{2}}} \\	
\le& 2\left(\frac{2Mq}{{\epsilon_{0,n}}^{\frac 3 2}\sqrt n}\right)^{\nu_j(\mathscr D) - \nu_j(\mathscr D_0)}\frac{\left(\frac{2}{\epsilon_{0,n}}\right)^{2r\nu_j(\mathscr D)^2 + 2r \nu_j(\mathscr D)}}{s^{2r}d^{rd}} \left(\frac{\frac{1}{S_{j|pa_j(\mathscr D)}}}{\frac{1}{S_{|pa_j(\mathscr D^*)}}}\right)^{\frac{n-\nu_j(\mathscr D)}{2}}\\
\le& 2\exp\left\{-\eta n  + \log \left(2M\right) + \log\left(\frac{1}{\epsilon_{0,n}^{\frac 3 2}}\right) - \frac 1 2 \log n + \frac{2r\nu_j(\mathscr D)^2+2r\nu_j(\mathscr D)}{\nu_j(\mathscr D) - \nu_j(\mathscr D_0)}\log \left(\frac{1}{\epsilon_{0,n}}\right)\right\}^{\nu_j(\mathscr D) - \nu_j(\mathscr D_0)}\\
&\times \left(1 + \frac{\frac {2c^\prime} {\epsilon_{0,n}^3}(\nu_j({\mathscr{D}} )+\nu_j({\mathscr{D}}^*)+2)\sqrt{\frac{\log p }{n}}}{\epsilon_{0,n}/2} \right )^{\frac{n-d}{2}} 
\end{split}
\end{align}
First consider the scenario when $\nu_j(\mathscr D) \ge 3\nu_j(\mathscr D_0).$ In particular, we have $\nu_j(\mathscr D) + \nu_j(\mathscr D^*) \le  \nu_j(\mathscr D) + \nu_j(\mathscr D_0) \le 2(\nu_j(\mathscr D) - \nu_j(\mathscr D_0)).$ Therefore,
\begin{align} \label{proof2}
\begin{split}
&C_j(\mathscr{D},\mathscr{D}_0) \\
\le& 2\exp\left\{-\eta n  + \log (2M) + \log\left(\frac{1}{\epsilon_{0,n}^{\frac 3 2}}\right) - \frac 1 2 \log n + \frac{2r\nu_j(\mathscr D)^2+2r\nu_j(\mathscr D)}{\nu_j(\mathscr D) - \nu_j(\mathscr D_0)}\log \left(\frac{1}{\epsilon_{0,n}}\right)\right\}^{\nu_j(\mathscr D) - \nu_j(\mathscr D_0)}\\
&\times \left(1 + \frac{\frac {2c^\prime} {\epsilon_{0,n}^3}(\nu_j({\mathscr{D}} )+\nu_j({\mathscr{D}}^*)+2)\sqrt{\frac{\log p }{n}}}{\epsilon_{0,n}/2} \right )^{\frac{n-d}{2}} \\
\le & 2\exp\left\{-\eta n  + \log (2M) + \frac 3 2\log\left(\frac{1}{\epsilon_{0,n}}\right) - \frac 1 2 \log n + 6r\nu_j(\mathscr D)\log \left(\frac{1}{\epsilon_{0,n}}\right)\right\}^{\nu_j(\mathscr D) - \nu_j(\mathscr D_0)}\\
&\times \left(1 + \frac{\frac {8c^\prime} {\epsilon_{0,n}^3}(\nu_j({\mathscr{D}} )-\nu_j({\mathscr{D}}_0))\sqrt{\frac{\log p }{n}}}{\epsilon_{0,n}/2} \right )^{\frac{n-d}{2}}\\
&\le 2\exp\left\{-\eta n   + 6r\nu_j(\mathscr D)\log \left(\frac{1}{\epsilon_{0,n}}\right) +\frac{16{c^\prime}}{\epsilon_{0,n}^4}\sqrt{n\log p } \right\}^{\nu_j(\mathscr D) - \nu_j(\mathscr D_0)}.
\end{split}
\end{align}
Note that $\nu_j(\mathscr D) = O \left( d \left( \frac{n}{\log p} \right)^{\frac{1}{2(2+k)}} \right)$ by hypothesis. Using 
Assumption 1 and Assumption 2, we get $\nu_j(\mathscr D) = o\left(\sqrt{\frac{n}{\log p}}\epsilon_{0,n}^4\right)$. 
Since $\eta$ has a strictly larger order than $\frac{\sqrt{\frac{\log p}{n}}}{\epsilon_{0,n}^4}$ and 
$\frac{\nu_j(\mathscr D)}{n}\log \left(\frac{1}{\epsilon_{0,n}}\right)$, there exists $N_2$ such that 
$$
C_j({\mathscr{D}} ,{\mathscr{D}}_0) \le 2\left(e^{-\frac{\eta}{2}n} 
\right)^{\nu_j({\mathscr{D}} )-\nu_j({\mathscr{D}}_0 )} \le 2e^{-\frac 1 2 \eta n} 
$$

\noindent
Now consider the scenario when $\nu_j(\mathscr D_0)+1 \le \nu_j(\mathscr D) < 3\nu_j(\mathscr D_0)$. Using the fact that 
$\nu_j(\mathscr D) - \nu_j(\mathscr D_0) \geq 1$, $\nu_j(\mathscr D) + \nu_j(\mathscr D*) \leq 4d$ and following similar steps 
as in (\ref{proof2}), we also have 
$$
C_j({\mathscr{D}} ,{\mathscr{D}}_0) \le 2e^{-\frac 1 2 \eta n}. 
$$

\noindent
Note that if $\mathscr D \neq \mathscr D_0$, then $\nu_j(\mathscr D) \neq \nu_j(\mathscr D_0)$ for atleast one $j$. 
It follows that for every $\mathscr D \neq \mathscr D_0$
$$
\frac{\pi(\mathscr D | \bm Y)}{\pi(\mathscr D_0 | \bm Y)} \leq 2e^{-\frac 1 2 \eta n} 
$$

\noindent
on the event $C_n^c$ (defined in (\ref{smplbound2}) in the main paper). Theprem \ref{thm4} now follows by the same 
sequence of arguments as in (\ref{pp6.1})-(\ref{pfm2}) in the main paper.

\section{Posterior convergence rates for DAG Wishart priors}

\noindent
In this section, we will provide the convergence rate for the posterior distribution of the 
precision matrix under the conditions in Theorem \ref{thm3}. Let $|| A ||_{(2,2)} = 
\{eig_p(A^TA)\}^{\frac12}$ for any $p \times p$ matrix $A$, and 
$\Pi(\cdot \; \mid \bm{Y})$ denote the probability measure corresponding to the 
posterior distribution. 
\begin{theorem} \label{3.1}
	Under Assumptions 1-5, if $p = p_n \rightarrow \infty$ and we only consider DAG's 
	with number of edges at most $\frac 1 8 d \left( \frac{n}{\log p}\right)^{\frac{1+k}{2+k}}$, 
	then 
	$$
	\bar{E} \left [\Pi \left\{|| {\Omega} - \Omega_0 ||_{(2,2)} \ge K d^{2}\sqrt{\frac{\log{p}}{\epsilon_{0,n}^4 n}} | \bm{Y} 
	\right\} \right] \rightarrow 0, \mbox{ as } n \rightarrow \infty. 
	$$
\end{theorem}

\noindent
{\it Proof} Let $\epsilon_n = d^{2}\sqrt{\frac{\log{p}}{\epsilon_{0,n}^4 n}}$. First note that,
\begin{align}
& \bar{E}[ \Pi \{|| {\Omega} - \Omega_0 || _{(2,2)} \ge K\epsilon_n | \bm{Y} \}] 
\nonumber\\
=&  \sum_{\mathscr D} \bar{E}[ \Pi(|| {\Omega} - \Omega_0 ||_{(2,2)} \ge K
\epsilon_n | \bm{Y},\mathscr D ) \pi(\mathscr D | \bm Y)] \nonumber\\
\le& \bar{E} \left[ \Pi \{|| {\Omega} - \Omega_0 ||_{(2,2)} \ge K\epsilon_n | \bm{Y},
\mathscr{D}_0 \}\right] + \bar{E}[\Pi(\mathscr D \neq \mathscr D_0) | \bm{Y}]. 
\label{post1}
\end{align}

\noindent
By Theorem \ref{thm3} in the main paper, it suffices to show that 
$$
\bar{E}[\Pi \{|| {\Omega} - \Omega_0 || _{(2,2)} \ge K\epsilon_n | \bm{Y},
\mathscr{D}_0 \}] \rightarrow 0, \mbox{ as } n \rightarrow \infty. 
$$

\noindent
Since all further analysis is going to be conditional of $\mathscr{D}_0$ (based on the 
above observation), for notational convenience we will use $A_{\cdot i}^>, A^{>i}, 
A^{\geq i}$ to denote $A_{\mathscr{D}_0 \cdot i}^>, A_{\mathscr{D}_0}^{>i}, 
A_{\mathscr{D}_0}^{\geq i}$ respectively. Next, we state some distributional properties 
for the Cholesky parameter $(D,L)$ under the DAG-Wishart prior. Under the 
DAG-Wishart prior $\pi_{U,\bm \alpha}^{\Theta_{\mathscr{D}_0}}$ on $(D,L)$, the 
posterior distribution of $(D,L)$ (given $\bm{Y}$ and $\mathscr{D}_0$) satisfies:
\begin{align}
& L_{\cdot i}^> | D_{ii}, \bm{Y}, \mathscr{D}_0 \sim N(-(\tilde{S}^{>i})^{-1} 
\tilde{S}_{\cdot i}^>, \frac 1 n D_{ii}(\tilde{S}^{>i})^{-1}), \mbox{ for } 1 \le i \le p-1, 
\label{d1}\\
& D_{ii}^{-1} \mid \bm{Y}, \mathscr{D}_0 \sim \mbox{Gamma}(\frac{n+\alpha_i}{2} + 1, 
\frac{nc_i}{2}), \mbox{ for } 1 \le i \le p, \label{d2}
\end{align}

\noindent
where $\tilde{S} = S + \frac U n$, and $c_i = \tilde{S}_{ii} - ({\tilde{S}_{\cdot i}^>})^T 
(\tilde{S}^{>i})^{-1} (\tilde{S}_{\cdot i}^>)$. These properties follow immediately from 
\cite[Eq. (9),(10)]{BLMR:2016}, and will be useful in proving Theorem \ref{3.1}. 

Note that to prove Theorem \ref{3.1} we only need to show that $\Pi \{|| \Omega - 
{\Omega_0} ||_{(2,2)} \ge K\epsilon_n | \bm{Y}, \mathscr{D}_0\} \stackrel{\bar{P}}
{\rightarrow} 0$ for a large enough $K$ (where $\Pi(\cdot \mid \bm{Y}, 
\mathscr{D}_0)$ denotes the probability measure corresponding to the posterior 
distribution given $\bm{Y}, \mathscr{D}_0$). By the triangle inequality, 
\begin{align} \label{p1}
\begin{split}
&\Pi \{|| \Omega - {\Omega_0} ||_{(2,2)} \ge K\epsilon_n | \bm{Y}, \mathscr{D}_0\} \\
=& \Pi \{|| LD^{-1}L^T - L_0{D_0}^{-1}L_0^T || _{(2,2)} \ge K\epsilon_n|\bm{Y}, 
\mathscr{D}_0\}\\
\le& \Pi \{|| L || _{(2,2)}|| D^{-1}-{D_0}^{-1}||_{(2,2)} || L^T || _{(2,2)} \ge \frac{K
	\epsilon_n}{3}|\bm{Y}, \mathscr{D}_0\} \\
&+ \Pi \{|| L || _{(2,2)}|| L-L_0 || _{(2,2)} || {D_0}^{-1} || _{(2,2)} \ge \frac{K\epsilon_n}
{3}|\bm{Y}, \mathscr{D}_0\} \\ 
&+ \Pi \{|| L_0 || _{(2,2)} || L-L_0 || _{(2,2)} || {D_0}^{-1} || _{(2,2)} \ge \frac{K
	\epsilon_n}{3}|\bm{Y}, \mathscr{D}_0\}.
\end{split}
\end{align}

\noindent
For the first part of (\ref{p1}), by the triangle inequality again, 
\begin{align} \label{p2}
\begin{split}
&\Pi \{|| L ||_{(2,2)}|| D^{-1}-{D_0}^{-1} || _{(2,2)}|| L^T || _{(2,2)} \ge \frac{K
	\epsilon_n}{3}|\bm{Y}, \mathscr{D}_0\} \\
\le & \Pi \{|| L-L_0 || _{(2,2)}^2|| D^{-1}-{D_0}^{-1} || _{(2,2)} \ge \frac{K\epsilon_n}
{12}|\bm{Y}, \mathscr{D}_0 \} \\ 
&+ 2\Pi \{|| L-L_0 || _{(2,2)}|| D^{-1}-{D_0}^{-1} || _{(2,2)} || L_0 || _{(2,2)} \ge \frac{K
	\epsilon_n}{12}|\bm{Y}, \mathscr{D}_0 \} \\
&+ \Pi \{|| L_0 || _{(2,2)}^2 || D^{-1}-{D_0}^{-1} || _{(2,2)} \ge \frac{K\epsilon_n}{12}|
\bm{Y}, \mathscr{D}_0\}.
\end{split}
\end{align}

\noindent
Similarly, for the other two parts of (\ref{p1}), using $\|L_0\|_{(2,2)} = 1$, we have 
\begin{align}
&\Pi \{|| L||_{(2,2)}|| L-L_0||_{(2,2)} || {D_0}^{-1}||_{(2,2)} \ge \frac{K\epsilon_n}
{3}|\bm{\bm{Y}, \mathscr{D}_0}\} \nonumber\\
\le& \Pi \{|| L-L_0||_{(2,2)} \ge \sqrt{\frac{K\epsilon_n}{6 
		|| {D_0}^{-1} ||_{(2,2)}}} |\bm{Y}, \mathscr{D}_0\} \nonumber\\
&+ \Pi \{|| L-L_0||_{(2,2)} \ge \frac{K\epsilon_n}{6 || {D_0}^{-1} || _{(2,2)}}|
\bm{Y}, \mathscr{D}_0\}, \label{p4}
\end{align}

\noindent
and 
\begin{align}
&\Pi \{|| L_0 ||_{(2,2)} || L-L_0 ||_{(2,2)} || {D_0}^{-1} ||_{(2,2)}\ge\frac{K\epsilon_n}
{3}|\bm{Y}, \mathscr{D}_0\} \nonumber\\
\le& \Pi \{|| L-L_0 || _{(2,2)} \ge \frac{K\epsilon_n}{3 || {D_0}^{-1}||_{(2,2)}}|
\bm{Y}, \mathscr{D}_0\}. \label{p4.1}
\end{align}

\noindent
By Assumption 1, we have $\frac{\epsilon_{0,n}}{2} \le (D_0)_{ii}^{-1} \le \frac{2}
{\epsilon_{0,n}}, \mbox{ for } 1 \le i \le p$. Also, by Assumption 2, $\sqrt[4]{\epsilon_n} \ge 
\sqrt{\epsilon_n} \ge \epsilon_n$, for $n$ large enough. By (\ref{p1}), (\ref{p2}), 
(\ref{p4}) and (\ref{p4.1}), to prove the required result, it suffices to show that both 
\begin{equation*}
\Pi \{ || L-L_0 ||_{(2,2)} \ge K_1 \epsilon_n\epsilon_{0,n}|\bm{Y}, \mathscr{D}_0\} 
\stackrel{\bar{P}}{\rightarrow} 0,
\end{equation*}

\noindent
and 
\begin{equation*}
\Pi \{|| D^{-1}-{D_0}^{-1}|| _{(2,2)} \ge K_1 \epsilon_n\epsilon_{0,n} |\bm{Y}, 
\mathscr{D}_0\} \stackrel{\bar{P}}{\rightarrow} 0,
\end{equation*}

\noindent
for some large enough constant $K_1$.

Now, let $\delta_n = \frac{\epsilon_n\epsilon_{0,n}}{d} = \frac{d\sqrt{\frac{\log{p}}{n}}}{\epsilon_{0,n}}$. Then, for a 
large enough constant $K_1$, we get 
\begin{align} \label{p5}
\begin{split}
&\Pi\{|| D^{-1}-{D_0}^{-1} || _{(2,2)} \ge K_1 \epsilon_n\epsilon_{0,n}|\bm{Y}, \mathscr{D}_0\} \\
\le& \Pi\{\mbox{max}_i (2|D_{ii}^{-\frac12} - (D_0)_{ii}^{-\frac12}| + 
(D_0)_{ii}^{\frac12})|D_{ii}^{-\frac12} - (D_0)_{ii}^{-\frac12}| \ge K_1\delta_n|\bm{Y}, 
\mathscr{D}_0\} \\
\le& \Pi\{\mbox{max}_i |D_{ii}^{-\frac12} - (D_0)_{ii}^{-\frac12}| \ge 
\frac12\sqrt{K_1\delta_n} | \bm{Y}, \mathscr{D}_0\} \stackrel{\bar{P}}{\rightarrow} 0, 
\end{split}
\end{align}

\noindent
by (\ref{d2}), \cite[Lemma 3.9]{XKG:2015}, and the fact that $\sqrt{\delta_n} 
\ge \delta_n =  \frac{d\sqrt{\frac{\log{p}}{n}}}{\epsilon_{0,n}}$ for large enough $n$ by (\ref{assumptioneigenvalue}). 

Let $|| A ||_{\mbox{max}} = \mbox{max}_{1 \le i,j \le p}|A_{ij}|$ for any $p \times p$ matrix $A$. In order to show 
$\Pi \{|| L-L_0 || _{(2,2)} \ge K_1\epsilon_n\epsilon_{0,n}|\bm{Y}, \mathscr{D}_0\} 
\stackrel{\bar{P}}{\rightarrow} 0$, first, by \cite[Lemma 3.1]{XKG:2015} and the triangle inequality, we have 
\begin{align}
&\Pi \{|| L - L_0 ||_{(2,2)} \ge  K_1\epsilon_n\epsilon_{0,n} | \bm{Y}, \mathscr{D}_0  \} \nonumber\\
\le& \Pi \{d|| L - L_0 ||_{\mbox{max}} \ge  K_1\epsilon_n\epsilon_{0,n} | \bm{Y}, \mathscr{D}_0 \} 
\nonumber\\
\le& \Pi \{|| L - M ||_{\mbox{max}} \ge \frac{K_1\delta_n}{2} | \bm{Y}, 
\mathscr{D}_0 \} + \Pi \{|| M - L_0 ||_{\mbox{max}} \ge \frac{K_1\delta_n}{2} | 
\bm{Y}, \mathscr{D}_0 \}, \label{p6}
\end{align}

\noindent
where $M$ is a $p \times p$ lower triangular matrix with unit diagonals defined as
$$
M_{\cdot i}^{\ge} = 
\begin{Bmatrix}	
1\\
- (\tilde{S}^{>i})^{-1} \tilde{S}_{\cdot i}^{>}
\end{Bmatrix},
$$ 

\noindent
for $1 \le i \le p$ and $M_{ji} = 0 \mbox{ for } j \notin pa_i (\mathscr{D}_0), 1 \le i < j \le 
p$. For the first term of (\ref{p6}),  by the union-sum inequality, we have
\begin{align*} 
&\Pi \{|| L - M ||_{\mbox{max}} \ge \frac{K_1\delta_n}{2} | \bm{Y}, \mathscr{D}_0 \} \\
=& \Pi \{\mbox{max}_{j \in pa_i (\mathscr{D}_0), 1\le i < j \le p} \lvert L_{ji} - M_{ji} 
\rvert \ge \frac{K_1\delta_n}{2} | \bm{Y}, \mathscr{D}_0 \}\\
\le& \Pi \{\mbox{max}_{j \in pa_i (\mathscr{D}_0), 1\le i < j \le p} \frac{|L_{ji} - M_{ji}|}
{\sqrt{D_{ii}}} (|\sqrt{D_{ii}} - \sqrt{(D_0)_{ii}}| + \sqrt{(D_0)_{ii}}) \ge \frac{K_1\delta_n}
{2} | \bm{Y}, \mathscr{D}_0 \}\\
\le& pd \mbox{max}_{j \in pa_i (\mathscr{D}_0), 1\le i < j \le p} \Pi \left\{ \frac{|L_{ji} - 
	M_{ji}|}{\sqrt{D_{ii}}} \ge \sqrt{\frac{K_1\delta_n}{4}} | \bm{Y}, \mathscr{D}_0 \right\}\\
&+ \Pi \{\mbox{max}_{1\le i \le p} |\sqrt{D_{ii}} - \sqrt{(D_0)_{ii}}| \ge 
\sqrt{\frac{K_1\delta_n}{4}} | \bm{Y}, \mathscr{D}_0 \}\\
&+ pd \mbox{max}_{j \in pa_i (\mathscr{D}_0), 1\le i < j \le p} \Pi \left\{ \frac{|L_{ji} - 
	M_{ji}|}{\sqrt{D_{ii}}} \ge \frac{K_1\delta_n}{4 \sqrt{\mbox{max}_i (D_0)_{ii}}} | \bm{Y}, 
\mathscr{D}_0 \right\}. 
\end{align*}

\noindent
Note that $\frac{\epsilon_{0,n}}{2} \le (D_0)_{ii}^{-1} \le \frac{2}{\epsilon_{0,n}}\mbox{ for } 1 
\le i \le p,$ and $\sqrt{\delta_n} \ge \delta_n$ for large enough $n$ by Assumption 2. 
By (\ref{p2}), (\ref{p4}), (\ref{p4.1}), (\ref{p5}) and (\ref{p6}), to prove the required 
result, it suffices to show that, for a large enough constant $K_1$ 
\begin{equation*} 
pd \mbox{max}_{j \in pa(i), 1\le i < j \le p} \Pi \left\{ \frac{|L_{ji} - M_{ji}|}{\sqrt{D_{ii}}} 
\ge K_1\delta_n\sqrt{\epsilon_{0,n}} | \bm{Y}, \mathscr{D}_0 \right\} \stackrel{\bar{P}}{\rightarrow} 0, 
\end{equation*}

\begin{equation*}
\Pi \{\mbox{max}_{1\le i \le p} |\sqrt{D_{ii}} - \sqrt{(D_0)_{ii}}| \ge K_1 \delta_n | \bm{Y}, 
\mathscr{D}_0 \} \stackrel{\bar{P}}{\rightarrow} 0, 
\end{equation*}

\noindent
and 
\begin{equation*}
\Pi \{|| M - L_0 ||_{\mbox{max}} \ge K_1\delta_n | \bm{Y}, \mathscr{D}_0 \} = 
1_{\{|| M - L_0 ||_{\mbox{max}} \ge K_1\delta_n\}} \stackrel{\bar{P}}{\rightarrow} 0. 
\end{equation*}

\noindent
We will now prove each of these three statements. 

\medskip

\noindent
1) First, let $Z_{ji} = \frac{\sqrt{n}|L_{ji} -M_{ji}|}{\sqrt{D_{ii}} \sqrt{\left( \tilde{S}^{>i} 
		\right)_{jj}^{-1} }}, c = \mbox{max}_{1\le i < j \le p}\left( \tilde{S}^{>i} \right)_{jj}^{-1}$, 
where $\left(\tilde{S}^{>i} \right)_{jj}^{-1}, j \ge i$ represent the diagonal element of $
\left( \tilde{S}^{>i} \right)^{-1}$ corresponding with vertex $j$. It follows from (\ref{d1}) 
that  $Z_{ji} \sim N(0,1)$ conditional on $\bm{Y}, \mathscr{D}_0, 
D_{ii}$. Since this distribution does not depend on $D_{ii}$, it follows that the 
distribution of $Z_{ji}$ conditional on $\bm{Y}, \mathscr{D}_0$ is also standard 
normal. Hence, 
\begin{align} \label{p8}
\begin{split}
&pd \mbox{max}_{j \in pa(i), 1\le i < j \le p} \Pi \left\{ \frac{|L_{ji} - M_{ji}|}{\sqrt{D_{ii}}} 
\ge K_1\delta_n\sqrt{\epsilon_{0,n}} | \bm{Y}, \mathscr{D}_0 \right\} \\
\le& pd \mbox{max}_{j \in pa(i), 1\le i < j \le p}  \Pi \left\{ \frac{\sqrt{n}|L_{ji} - M_{ji}|}
{\sqrt{D_{ii}} \sqrt{\left(\tilde{S}^{>i} \right)_{jj}^{-1} }} \ge \frac{\sqrt{n} K_1 \delta_n\sqrt{\epsilon_{0,n}}}
{\sqrt{c}} | \bm{Y}, \mathscr{D}_0 \right\} \\
\le& 2pd\left(1-\phi\left( \frac{\sqrt{n} K_1 \delta_n\sqrt{\epsilon_{0,n}}}{\sqrt{c}} \right)  \right) \\
\le& 2pd \exp\left( - \frac{nK_1^2 \delta_n^2 \epsilon_{0,n}}{2c}\right) \\
\le& 2p^{2 - \frac{K_1^2d^2\log p}{2c\epsilon_{0,n}}},	
\end{split}
\end{align}

\noindent
where $\Phi$ is the standard normal distribution function and $1 - \Phi(t) \le \exp (- 
\frac{t^2}2)$. The last inequality follows from $\delta_n = \frac{\epsilon_n}{d} = 
d\frac{\sqrt{\frac{\log{p}}{n}}}{\epsilon_{0,n}}$. Next, for any $t > 0$, 
\begin{align} \label{p9}
\begin{split}
&\bar{P}(2p^{2 - \frac{K_1^2d^2\log p}{2c\epsilon_{0,n}}} > t) \\
\le& \bar{P}(2p^{2 - \frac{K_1^2d^2\log p}{2c\epsilon_{0,n}}} > t, c < 2\epsilon_{0,n}^{-1}) + \bar{P}(c \ge 
2\epsilon_{0,n}^{-1}) \\
\le &\bar{P}(2p^{2 - \frac{K_1^2d^2\log p}{4}} > \eta) + \bar{P}(c \ge 
2\epsilon_{0,n}^{-1}),	
\end{split}
\end{align}

\noindent
It follows from \cite[Lemma 3.5]{XKG:2015} that $\bar{P} \{c \ge 2\epsilon_{0,n}^{-1} \} 
\rightarrow 0$ as $n \rightarrow \infty$. Also, $p^{2 - \frac{K_1^2d^2\log p}{4}} > \eta \rightarrow 0$, 
as $n \rightarrow \infty$. It follows by (\ref{p9}) that $2p^{2 - \frac{K_1^2d^2\log p}{2c\epsilon_{0,n}}} 
\stackrel{\bar{P}}{\rightarrow} 0$. By (\ref{p8}), we get 
\begin{equation} \label{p15}
pd \mbox{max}_{j \in pa(i), 1\le i < j \le p} \Pi \left\{ \frac{|L_{ji} - M_{ji}|}{\sqrt{D_{ii}}} \ge 
K_1\delta_n\sqrt{\epsilon_{0,n}} | \bm{Y}, \mathscr{D}_0 \right\} \stackrel{\bar{P}}{\rightarrow} 0. 
\end{equation}

\noindent
for a large enough constant $K_1$. 

\medskip

\noindent
2) Next, let $r = \mbox{max}_{1 \le i \le p} \sqrt{D_{ii} (D_0)_{ii}}$. Then, 
\begin{align} \label{p10}
\begin{split}
&\Pi \{\mbox{max}_{1\le i \le p} |\sqrt{D_{ii}} - \sqrt{(D_0)_{ii}}| \ge K_1\delta_n | 
\bm{Y}, \mathscr{D}_0 \} \\
\le& \Pi \{\mbox{max}_{1\le i \le p} |\sqrt{D_{ii}^{-1}} - \sqrt{(D_0)_{ii}^{-1}}| \ge \frac1r 
K_1\delta_n | \bm{Y}, \mathscr{D}_0 \} \\
\le& \Pi \{\mbox{max}_{1\le i \le p} |\sqrt{D_{ii}^{-1}} - \sqrt{(D_0)_{ii}^{-1}}| \ge \frac1r 
K_1\delta_n, r < \frac4{\epsilon_{0,n}} | \bm{Y}, \mathscr{D}_0 \} \\
&+ \Pi \{ r \ge \frac4{\epsilon_{0,n}} | \bm{Y}, \mathscr{D}_0\} \\
\le &\Pi \{\mbox{max}_{1\le i \le p} |\sqrt{D_{ii}^{-1}} - \sqrt{(D_0)_{ii}^{-1}}| \ge 
\frac{\epsilon_{0,n}}{4} K_1\delta_n | \bm{Y}, \mathscr{D}_0 \} \\
&+ \Pi \{\mbox{min}_{1\le i \le p}(D_{ii}(D_0)_{ii})^{-\frac12} \le  \frac{\epsilon_{0,n}}4 | 
\bm{Y}, \mathscr{D}_0\}.	
\end{split}
\end{align}

\noindent
As observed before, by \cite[Lemma 3.9]{XKG:2015}, 
\begin{equation*}
\Pi \{\mbox{max}_{1\le i \le p} |\sqrt{D_{ii}^{-1}} - \sqrt{(D_0)_{ii}^{-1}}| \ge 
\frac4{K_1}d\sqrt{\frac {\log p} n} | \bm{Y}, \mathscr{D}_0 \} \rightarrow 0.
\end{equation*}

\noindent
for a large enough $K_1$. Also, using $\frac{\epsilon_{0,n}}{2} \leq (D_0)_{ii}^{-1} \leq 
\frac{2}{\epsilon_{0,n}}$ and \cite[Lemma 3.9]{XKG:2015} again, we have 
\begin{align} \label{p11}
\begin{split}
&\Pi \{\mbox{min}_{1\le i \le p}(D_{ii}(D_0)_{ii})^{-\frac12} \le  \frac{\epsilon_{0,n}}4 | 
\bm{Y}, \mathscr{D}_0\} \\
\le& \Pi \{\mbox{min}_{1\le i \le p}(D_0)_{ii}(D_{ii}(D_0)_{ii})^{-\frac12} \le  \frac12 | 
\bm{Y}, \mathscr{D}_0\} \\
\le& \Pi \{\mbox{max}_{1\le i \le p}|1 - (D_0)_{ii}^{\frac12}D_{ii}^{-\frac12}| \ge  \frac12 | 
\bm{Y}, \mathscr{D}_0\} \\
\le& \Pi \{\mbox{max}_{1\le i \le p}|(D_0)_{ii}^{-\frac12} - D_{ii}^{-\frac12}| \ge  
\frac12\sqrt{\frac{\epsilon_{0,n}}2} | \bm{Y}, \mathscr{D}_0\} \stackrel{\bar{P}}{\rightarrow} 
0.	
\end{split}
\end{align}

\noindent
Hence, by (\ref{p10}) and (\ref{p11}), we get, 
\begin{equation} \label{p12}
\Pi \{\mbox{max}_{1\le i \le p} |\sqrt{D_{ii}} - \sqrt{(D_0)_{ii}}| \ge K_1\delta_n | \bm{Y}, 
\mathscr{D}_0 \} \stackrel{\bar{P}}{\rightarrow} 0. 
\end{equation}

\medskip

\noindent
3) Last, by \cite[Lemma 3.2]{XKG:2015} and \cite[Lemma 3.6]{XKG:2015}, for $1 \leq i 
\leq p-1$, 
$$
(L_0)_{\cdot i}^{\ge} = 
\begin{Bmatrix}	
1\\
- ((\Sigma_0)^{>i})^{-1} (\Sigma_0)_{\cdot i}^{>}
\end{Bmatrix}.
$$

\noindent
Hence, by  \cite[(3.24)]{XKG:2015}, it follows that 
\begin{align}
&\bar{P}\{|| M - L_0 ||_{\mbox{max}} \geq K_1\delta_n \} \nonumber\\
\le& \bar{P}\{\mbox{max}_{j \in pa_i (\mathscr{D}_0) 1\le i < j \le p} | [(\tilde{S}
^{>i})^{-1}\tilde{S}_{.i}]_j - [(\tilde{\Sigma}^{>i})^{-1}\tilde{\Sigma}_{.i}]_j |  \ge 
\frac{K_1\delta_n}{2}\} \rightarrow 0 \label{p13}
\end{align}

\noindent
for a large enough constant $K_1$. 
\end{document}